\documentclass[pra,aps,floatfix,superscriptaddress,11pt,tightenlines,longbibliography,onecolumn,notitlepage,accepted=2017-09-20]{quantumarticle}
\pdfoutput=1
\usepackage[utf8]{inputenc}
\usepackage[english]{babel}
\usepackage[T1]{fontenc}
\usepackage{amsmath,amsthm,amsfonts,amssymb, xcolor}
\usepackage{graphicx}
\usepackage{microtype}
\usepackage{braket}
\usepackage[numbers,sort&compress,square]{natbib}

\definecolor{darkblue}{rgb}{0.,0.,0.4}
\definecolor{darkred}{rgb}{0.5,0.,0.}
\usepackage[pdftex,colorlinks=true,linkcolor=darkblue,citecolor=darkred,urlcolor=darkblue]{hyperref}

\numberwithin{equation}{section}
\numberwithin{figure}{section}
\numberwithin{table}{section}

\newtheorem{theorem}{Theorem}[section]
\newtheorem{lemma}[theorem]{Lemma}
\theoremstyle{definition}
\newtheorem{definition}[theorem]{Definition}
\newtheorem*{lemma*}{Lemma}

\DeclareMathAlphabet{\mathpzc}{OT1}{pzc}{m}{it}

\newcommand{\FF}{\mathbb{F}}

\newcommand{\nout}{{n_\text{\rm outer}}}
\newcommand{\nin}{{n_\text{\rm inner}}}

\newcommand{\kin}{{k_\text{\rm inner}}}

\newcommand{\nT}{n_T}
\newcommand{\nphys}{n_\text{\rm phys}}
\newcommand{\calS}{\mathcal{S}} % stabilizer group or self-orthogonal subspace

\newcommand{\Sw}{{\rm Swap}}
\newcommand{\sO}{{\cal O}^*}

\newcommand{\dm}{{\rm diag}}

\begin{document}

\title{Magic State Distillation with Low Space Overhead and Optimal Asymptotic Input Count}
\author{Jeongwan Haah}
\affiliation{Quantum Architectures and Computation, Microsoft Research, Redmond, WA 98052, USA}

\author{Matthew B.~Hastings}

\affiliation{Station Q, Microsoft Research, Santa Barbara, CA 93106-6105, USA}
\affiliation{Quantum Architectures and Computation, Microsoft Research, Redmond, WA 98052, USA}

\author{D. Poulin}
\affiliation{D\'{e}partement de Physique \& Institut Quantique, Universit\'{e} de Sherbrooke, Quebec, Canada
}

\author{D. Wecker}
\affiliation{Quantum Architectures and Computation, Microsoft Research, Redmond, WA 98052, USA}

\date{26 September 2017}

\begin{abstract}
We present an infinite family of protocols to distill magic states for $T$-gates 
that has a low space overhead and uses an asymptotic number of input magic states to achieve a given target error that is conjectured to be optimal.  The space overhead, defined as the ratio between the physical qubits to the number of output magic states, is asymptotically constant, while both the number of input magic states used per output state and the $T$-gate depth of the circuit scale linearly in the logarithm of the target error $\delta$ (up to $\log \log 1/\delta$).
Unlike other distillation protocols, this protocol achieves this performance without concatenation and the input magic states are injected at various steps in the circuit rather than all at the start of the circuit. The protocol can be modified to distill magic states for other gates at the third level of the Clifford hierarchy, with the same asymptotic performance.  The protocol relies on the construction of weakly self-dual CSS codes with many logical qubits and large distance, allowing us to implement control-SWAPs on multiple qubits.  We call this code the ``inner code".  The control-SWAPs are then used to measure properties of the magic state and detect errors, using another code that we call the ``outer code".
Alternatively, we use weakly-self dual CSS codes which implement controlled Hadamards for the inner code, reducing circuit depth.  We present several specific small examples of this protocol.
\end{abstract}

% \maketitle

The possibility of a large scale quantum computer relies on
fault-tolerant architectures,
in which errors are corrected faster than they are created%
~\cite{Shor1996Fault-tolerant,AharonovBenOr1996,KnillLaflammeZurek1996Threshold}.
The standard approach is to use stabilizer codes 
to protect logical qubits from noise%
~\cite{Gottesman1996Saturating,CalderbankRainsShorEtAl1997Quantum},
and perform quantum gates at the encoded level.
The overhead of the fault-tolerance is \emph{only} polynomial
in the logarithm of the desired accuracy,
but in practice the overhead is estimated to be overwhelmingly large~\cite{JonesEtAl2012,OGormanCampbell2016}.
Particularly expensive operations are non-Clifford gates
such as the $\pi/4$-rotation (T-gate) and Toffoli gate.
A compelling approach is to inject a special state, called a magic state,
into a Clifford-only circuit,
and pass the cost of implementing the non-Clifford operation
to the preparation of the magic states, which are distilled from noisy 
ones~\cite{Knill2004a,Knill2004b,BravyiKitaev2005Magic}.

There exist several distillation protocols for the magic state ($T$-state)
using specialized quantum error correcting codes%
~\cite{Knill2004a,BravyiKitaev2005Magic,MeierEastinKnill2012Magic-state,BravyiHaah2012Magic}.
Each code can provide a fixed degree of fidelity improvement 
that is given by the code distance.
In order to achieve arbitrary good fidelity,
one typically concatenates small routines.
In terms of the number of input magic states of low fidelity 
per one output magic state of high fidelity,
the best protocols to date are those in Refs.~\cite{BravyiHaah2012Magic, Jones2012}.
However, these protocols require a large batch 
of thousands of magic states to be useful.

In this paper, we introduce an infinite family of distillation protocols,
extending one of the very first protocols by Knill~\cite{Knill2004a},
and another by Meier, Eastin, and Knill~\cite{MeierEastinKnill2012Magic-state}.
Our protocol produces $n$ $T$-magic states using at most $cn$ qubits
and achieves at least $c'n$-th order error suppression
under the assumption that the sole noise source is the $T$ gate,
where $c,c'$ are small universal constants.
Since the degree of error suppression is high,
there is no need to concatenate small routines,
reducing the space overhead significantly.

Our protocol is also asymptotically superior (conjectured to be optimal) in terms of noisy $T$ count.
For any fixed odd $d \ge 5$,
we show that the number of noisy $T$ gates per one output magic state
with error suppressed to $d$-th order
converges to $d$ exactly in the large code length limit.

Beyond the magic states for $T$ gates,
our protocol can distill magic states for rotation by $\pi/2^k$
for $k=3,4,\ldots$ adapting the idea of Ref.~\cite{DP2014},
and any gate in the third level of Clifford hierarchy~\cite{ChuangGottesman1999}.
(See also Ref.~\cite{LandahlCesare2013Tgate}
for smaller angle ($k \ge 3$) rotations,
though we do not use ideas there.)
For the latter, the asymptotic performance is similar to the $T$ gate case.

Small instances of our family demonstrates 
reduction of space overhead, 
with a modest input $T$ count.
If noisy $\pi/4$ rotations can be directly done on qubits,
an instance of our family operates on 34 qubits including measurement ancillas,
produces 15 $T$-magic states with 5th order error suppression,
and requires 29 noisy $T$ gates per output.
In comparison, to the authors' knowledge, 
any previous protocol 
that can operate on less than 50 qubits
have either have lower order of error suppression,
or requires more non-Clifford gates per output.

Recent innovations show that the $\pi/4$-rotation and Toffoli gate 
can be implemented fault-tolerantly on a class of error correcting codes~\cite{
	PaetznickReichardt2013Universal,
	JochymOConnorLaflamme2014,
	ADP14a,
	Bombin2015,
	BravyiCross2015,
	JochymOConnorBartlett2016,
	JonesBrooksHarrington2016,
	Bombin2016,
	Yoder2016}.
These schemes achieve computational universality
through local operations
while circumventing no-go theorems~\cite{
	EastinKnill2009Restrictions,
	BravyiKoenig2012Classification}
by going back and forth between two code spaces.
This approach removes the need for magic states,
but, it is not a simple question to tell which approach is better.
This question depends on an architecture and underlying physical qubits' characteristic,
and thus we leave the realistic cost analysis and comparison to future work.

The organization of the present paper is as follows.
We start in Section~\ref{sec:basic} with a basic explanation 
of our ideas by exhibiting examples of small sizes.
Section~\ref{sec:cliff3} explains how to convert magic state distillation protocols for $T$ gates
into those for Toffoli gates.
In Section~\ref{sec:innercodes} we show that any weakly self-dual CSS code
can be used in distillation protocols by implementing measurement of Clifford operators.
In Section~\ref{sec:asymptotic}, 
we give asymptotic constructions of the codes, 
in the limit of either large distance or large code length.  
In Section~\ref{sec:simulation}, we give results of numerical simulations; 
in this section we also present some specific additional small size protocols that are not described elsewhere in the paper.
We conclude with discussion in Section~\ref{sec:discussion}.
In Appendix~\ref{app:codes}, we give details, including stabilizer checks, for some of the specific codes used in the paper.
Appendix~\ref{app:circuits} gives circuits for some of the protocols used.
Appendix~\ref{app:coincidence} describes unexpected relations among different distillation protocols.
Appendix~\ref{app:qudits} explains an extension to qudits,
using classification of symmetric forms over finite fields in Appendix~\ref{symmformfinite}.

Throughout this paper, 
all classical codes that we consider will be linear codes over the binary field $\FF_2$ 
and all quantum codes will be qubit stabilizer codes
(except for Appendix~\ref{app:qudits}).
Given a bit vector $v$, we use $|v|$ to denote its Hamming weight.
Our magic state is the $(+1)$-eigenstate
$\ket H = \cos \frac \pi 8 \ket 0 + \sin \frac \pi 8 \ket 1$
of the Hadamard operator $H$.
We use matrices
\begin{align}
X = \begin{pmatrix} 0 & 1 \\ 1 & 0 \end{pmatrix}, \quad
Y = \begin{pmatrix} 0 & -i \\ i & 0 \end{pmatrix}, \quad
Z = \begin{pmatrix} 1 & 0 \\ 0 & -1 \end{pmatrix},\\
H = \frac{1}{\sqrt{2}}\begin{pmatrix} 1 & 1 \\ 1 & -1 \end{pmatrix}, 
\quad
T = e^{-i \pi Y/8} = 
\begin{pmatrix} 
\cos \frac \pi 8 & - \sin \frac \pi 8 \\
\sin \frac \pi 8 &   \cos \frac \pi 8 
\end{pmatrix}.
\end{align}
A self-orthogonal subspace of some vector space 
is a subspace such that the inner product of any two vectors in the subspace vanishes.
A weakly self-dual CSS code is a quantum code 
whose stabilizer group is invariant under conjugation by the product of the Hadamard operator on all qubits; 
we call this product a ``transversal Hadamard".

\section{Basic Distillation Protocols}
\label{sec:basic}

Distillation protocols for magic states to date can be put roughly into three classes.
Those in the first class implement a non-Clifford $\pi/4$-rotation to a stabilizer state 
such as $\ket +$ or $\ket 0$%
~\cite{BravyiKitaev2005Magic,BravyiHaah2012Magic,LandahlCesare2013Tgate}.
The non-Clifford rotation must be done fault-tolerantly,
so the protocols in this class focus on finding error correcting codes that admits a transversal non-Clifford rotation.
This requires the underlying code to have a special symmetry, which is rather rare.
The protocols in the second class%
~\cite{Knill2004a,Knill2004b,MeierEastinKnill2012Magic-state,Jones2012,DP2014,Campbell2016}
implement measurements of ``stabilizers'' of the magic state,
based on the fact that a magic state in the third level of Clifford hierarchy~\cite{ChuangGottesman1999}
is an eigenstate of a Clifford operator.
To measure a Clifford operator one needs a non-Clifford operation
which has been implemented fault-tolerantly by a distance-two code.
The third class uses yet different symmetries of codes~\cite{BravyiKitaev2005Magic,Reichardt2005},
and has high threshold for distillation, but the success probability does not reach 1
even with perfect input magic states.

Our scheme in the present paper belongs to the second class,
and is an extension of the idea of Knill~\cite{Knill2004a}.
There are two levels of error correcting codes in our scheme,
which we call \emph{inner} and \emph{outer} codes.
 Roughly speaking, 
the outer codes specify a certain set of measurements of  Clifford operators on a set of input magic states,
and the inner codes specify how to implement these measurements.
We illustrate aspects of our ideas by two examples.
They are not the best protocols in regards to 
e.g. the total number of non-Clifford gates and states,
but will be simple to explain.
A more general class of protocols is presented in later sections.

Without loss of generality,
by a standard Clifford twirling argument, we can assume that 
each $\pi/4$ rotation and undistilled magic state 
suffers from independent $Y$ errors with probability $\epsilon$.
We refer to this error model as the {\bf stochastic} error model.

\subsection{Trivial outer code}
If we could implement the control-Hadamard,
then the distillation is trivial:
Prepare an ancilla qubit in $\ket +$ state, 
apply the control-Hadamard with the control on the ancilla
and the Hadamard on an arbitrary target qubit,
and measure the ancilla in $X$-basis to accept $+1$ outcome.
The accepted target qubit is projected onto the magic state.

The control-Hadamard belongs to the third level of Clifford hierarchy,
and thus cannot be implemented with Clifford operations.
To obtain an approximate control-Hadamard with noisy non-Clifford rotations,
we must use an error correcting code that can implement $H$ on the logical qubits
fault-tolerantly.

\subsubsection{$[[7,1,3]]$ inner code}
To this end, 
we observe that the Steane code~\cite{Steane_1996} admits a transversal Hadamard~\cite[Sec. IX]{Knill2004a}.
The stabilizers are
\begin{align}
 \begin{tabular}{ccccccc}
 	I&I&I&X&X&X&X \\
 	I&X&X&I&I&X&X\\
 	X&I&X&I&X&I&X
 \end{tabular},
 \quad
 \begin{tabular}{ccccccc}
 	I&I&I&Z&Z&Z&Z\\
 	I&Z&Z&I&I&Z&Z\\
 	Z&I&Z&I&Z&I&Z
 \end{tabular},
\end{align}
the group generated by which is fixed under the $H^{\otimes 7} : X_i \leftrightarrow Z_i$,
and the logical operator pair is
\begin{align}
  \begin{tabular}{ccccccc}
 	Z&Z&Z&Z&Z&Z&Z
 \end{tabular},
 \quad
   \begin{tabular}{ccccccc}
 	X&X&X&X&X&X&X
 \end{tabular}
\end{align}
which are interchanged by the transversal Hadamard.
Using an identity
\begin{align}
 H = T Z T^\dagger = e^{-i \pi Y / 8} ~Z~ e^{ i \pi Y/8} 
 \label{eq:decomp-H}
\end{align} 
we see that the logical control-Hadamard is possible by replacing the middle $Z$ by
the control-$Z$. The $T$ gate can be noisy 
as they act on the physical qubits of the Steane code.

This way, we have built a Hadamard measurement routine that is fault-tolerant.
Then, a magic state distillation protocol is as follows:
\begin{enumerate}
	\item Prepare a noisy magic state in the ``data'' register, and $\ket 0$ in $6$ check registers, 
		and embed them into the Steane code
	\item Prepare an ancilla in $\ket +$ and implement control-$H^{\otimes 7}$ using Eq.~\eqref{eq:decomp-H}, 
		where the control is the ancilla and the targets are the physical qubits of the Steane code.
	\item Inverse the embedding of the Steane code.
	\item Measure the ancilla in the $X$ basis, the check qubits in the $Z$ basis.
	\item Upon $+1$ outcome in all 7 measurements, a distilled magic state is in the data qubit.
\end{enumerate}
Let us examine the pattern of errors that may go undetected. There are two possibilities.
\begin{itemize}
	\item The initial magic state is faulty, and this is undetected due to malfunction of the control-Hadamard.
	\item The noisy $\pi/4$ rotations induce a logical error.
\end{itemize}
The first possibility is because
 a pair of simultaneous errors sandwiching the control-$Z$ can alter the ancilla measurement:
\begin{align}
 &\ket 0 \bra 0 \otimes I + \ket 1 \bra 1 \otimes e^{-i \pi Y / 8} Y ~Z~ Y e^{ i \pi Y/8} \nonumber \\
 &= (Z \otimes I)\left( \ket 0 \bra 0 \otimes I + \ket 1 \bra 1 \otimes e^{-i \pi Y / 8} ~Z~ e^{ i \pi Y/8} \right).
 \label{eq:errorInControl}
\end{align}
Thus, the first possibility occurs with probability $\epsilon^3$ to leading order.
One can easily see that this is the only possibility for weight 2 errors from the
control-Hadamard to escape.
The second possibility occurs at order $\epsilon^3$ since Steane's code has distance 3.
Overall, the protocol operates on $8$ qubits
(assuming $T$ gates are applied in place),
consuming $14$ $T$-gates and $1$ state of error rate $\epsilon$,
producing $1$ output $T$-state whose error rate is $O(\epsilon^3)$.

It is useful to think of the above protocol as
a Hadamard measurement ($H$-measurement) routine
that introduces a new error of order $\epsilon^3$ to the target,
and another error of order $\epsilon^2$ to the control.
The error on the control is easy to fix; repeat the measurement.%
\footnote{This corresponds to having redundant checks in the outer code.}
The error on the target is inherent to the choice of the inner quantum code,
and should thus be overcome by another quantum code.

\subsubsection{$[[17,1,5]]$ inner code}
\label{1715}
There exists a distance 5 code on 17 qubits with $H^{\otimes 17}$ being the logical Hadamard.
It is an instance of the color code~\cite{Bombin_2006,BravyiCross2015}.
We include the binary matrix for this code in Appendix~\ref{app:codes}.
In a similar way as above,
this $H$-measurement routine has error rate $O(\epsilon^5)$ on the target,
and $O(\epsilon^2)$ on the control.
By repeating the $H$-measurement twice using this inner code,
the control's error rate becomes $O(\epsilon^4)$.
The control error goes undetected only if the initial magic state is faulty.
Overall, only weight 5 errors may be undetected.
This protocol consumes $17\times 4 + 1 = 69$ noisy $T$'s.

In fact, we can {\it pipeline} the two $H$-measurement routines:
First, $H$-measure a noisy magic state using $[[7,1,3]]$ code,
and then $H$-measure the outcome using $[[17,1,5]]$ code.%
\footnote{
Interestingly, the protocol using Steane code,
and the pipelined protocol appear to have deep relation with triply even codes.
See Appendix~\ref{app:coincidence}.}
Hence, we have obtained a distillation routine with fifth order error suppression
that operates on 18 qubits in total,
consuming $1+ 7 \times 2 + 17 \times 2 = 49$ noisy $T$'s.
It is worth comparing this with the $69 \to 1$ protocol
in the preceding paragraph.
By using codes of smaller distance in the early stage of the protocol,
we obtain a more efficient protocol.
We loosely call this modification as a pipelined protocol.
The circuit is in
Fig.~\ref{1715fig} in Appendix \ref{app:circuits}.

\subsection{Repetition outer code}
Imagine we have $\nout$ noisy magic states that are to be distilled.
Under the stochastic error model,
we can think of the noisy magic states as an probabilistic ensemble of $\nout$-bit strings
where $0$ denotes a good magic state, and $1$ denotes a bad one.
The protocol in the previous subsection examines one qubit at a time,
and in terms of the bit strings, this amounts to checking individual bits.
If our goal is to suppress the overall error to $d$-th order where $d < \nout$,
the bit-wise check might be overkill.
A better way is to devise a measurement routine that can check the parity of several bits.

\subsubsection{$[[4,2,2]]$ inner code}

The simplest case is when $\nout = 2$ and the desired error suppression is quadratic.
If we can measure $H^{\otimes 2}$, 
then by postselecting on $+1$ outcome the noisy state is projected to the even parity subspace,
which is $O(\epsilon^2)$ away from the pair of perfect magic states.
We can describe the situation by saying that we have a repetition code on $\nout = 2$ bits with one parity check.
This is an outer code.

A corresponding inner code should implement control-$H^{\otimes 2}$ to accuracy $O(\epsilon^2)$,
both in the target and the control.
Meier, Eastin, and Knill have designed such a measurement routine~\cite{MeierEastinKnill2012Magic-state}.
The four qubit code $[[4,2,2]]$ whose stabilizers are $X^{\otimes 4}$ and $Z^{\otimes 4}$ admits
the transversal Hadamard $\bar H = H^{\otimes 4}$ as a logical operator.
If we choose the logical operators as
\begin{align}
\left\{
 \begin{matrix}
  \bar X_1 = X & X & I & I \\
  \bar Z_1 = I & Z & Z & I
 \end{matrix}
 \right.
 \qquad
 \left\{
	 \begin{matrix}
  \bar X_2 = Z & Z & I & I \\
  \bar Z_2 = I & X & X & I 
	 \end{matrix}
 \right.
\end{align}
then the transversal Hadamard swaps the two logical qubits.
Using Eq.~\eqref{eq:decomp-H}, this means that we can implement control-$\Sw$ to accuracy $O(\epsilon^2)$.
Now, a trick is to use the control-$\Sw$ twice sandwiching the Hadamard:
\begin{align}
 [^C\Sw_{12}] H_1 [^C\Sw_{12}] = [^C(H_1 \otimes H_2)] H_1
 \label{eq:CSwap-to-CProduct}
\end{align}
where the superscript $C$ denotes the control that is common for both control-$\Sw$s, 
and the subscripts $1$ and $2$ denote the qubits the operator acts on.
The extra $H_1$ does no harm since the magic state is its eigenstate.
The obtained control-$H^{\otimes 2}$ is accurate up to error $O(\epsilon^2)$ on the target 
since the distance of the four-qubit code is 2,
and also $O(\epsilon^2)$ on the control due to Eq.~\eqref{eq:errorInControl}.
This is a quadratic distillation protocol
operating on 5 qubits, consuming 18 noisy $T$'s to produce $2$ outputs.%
\footnote{Meier, Eastin, and Knill~\cite{MeierEastinKnill2012Magic-state}
compresses the circuit for the control-$H^{\otimes 2}$
to reduce the number of $T$'s of the protocol from 18 to 10,
but we do not explain it here.
See also Campbell and O'Gorman~\cite{Campbell2016}.}

\subsubsection{$[[16,6,4]]$ inner code and pipelining}
\label{1664}

The classical Hadamard code $[16,5,8]$ has a property that every code word has even overlap with any other code word.
By the CSS construction, using these classical codewords as stabilizers, we obtain a $[[16,6,4]]$ code; 
see Appendix \ref{app1664} for the stabilizers.
We will later show that there is a choice of logical operators such that the transversal Hadamard $H^{\otimes 16}$
implements simultaneous pairwise swaps on the three pairs of logical qubits.
This implies that we can measure any even product 
$\tilde H^{\otimes 2}$, $\tilde H^{\otimes 4}$, or $\tilde H^{\otimes 6}$ of Hadamards
on $\kin = 6$ magic states.
For example, we can generalize Eq.~(\ref{eq:CSwap-to-CProduct}) to
\begin{align}
 [^C\Sw_{12}] [^C\Sw_{34}] H_1 H_3 [^C\Sw_{12}] [^C\Sw_{34}] = [^C(H_1 \otimes H_2 \otimes H_3 \otimes H_4)] H_1 H_3.
\end{align}
The $H$-measurement routine puts quadratic error to the control and quartic error to the target.

Imagine that $\nout = 6$ magic states are laid on a ring.
We measure $H^{\otimes 2}$ on every nearest neighboring pair of the magic states.
There are six checks in total.
The measurement pattern follows the parity checks of the classical repetition code; there is a redundancy in the checks, which turns out to be necessary.  Let us see how this achieves quartic error suppression.
In order for an error on one of $\nout$ magic states to pass the measurement routines,
the two checks that involve that input state must both be faulty.
This process gives an  $O(\epsilon^5)$ error, i.e., the probability of both checks being faulty is   $O(\epsilon^4)$, so including the error on the input magic state the error is  $O(\epsilon^5)$.  Note that if we did not have a redundancy in the checks of the outer code, using only $5$ checks, one qubit would be checked only once and we would achieve only third order
error suppression.
More generally, any process involving one or more input magic state errors gives an error which is at most $O(\epsilon^5)$.
The dominant error after all the $H$-measurements
is then from the logical error by the $H$-measurement routine,
which happens with probability $O(\epsilon^4)$.
Overall, the protocol consumes 
$6 + 6 \times (16 \times 4) = 390$ $T$'s
to produce $6$ outputs.

We can pipeline the $[[4,2,2]]$ code routine in front of the $[[16,6,4]]$ code routine
to lower the complexity of the distillation circuit.
For instance, we can run the three $H$-measurement routines by the $[[4,2,2]]$ code
on pairs of magic states $(12)$, $(34)$, and $(56)$, 
and then run the three $H$-measurement routines by $[[16,6,4]]$ code
on pairs of magic states $(23)$, $(45)$, and $(61)$.
The number of $T$'s consumed is now
$ 6 + 3 \times (4 \times 4) + 3 \times (16 \times 4) = 246$,
while the number of outputs is still $6$.
It is left to the readers to show that
the modified version also achieves quartic error suppression.
We have simulated the modified version, and
the results can be found in Section~\ref{sec:simulation}.
The circuit is in
Fig.~\ref{1664fig} in Appendix \ref{app:circuits}.

\section{Third level of Clifford hierarchy}
\label{sec:cliff3}

The protocol above can be straightforwardly generalized to distilling other magic states 
to implement gates at the third level of the Clifford hierarchy.
Consider a state $\ket \psi$ on $q$ qubits such that $\ket \psi = U \ket \phi$ 
where $U$ is a gate at the third level of the Clifford hierarchy~\cite{ChuangGottesman1999},
and $\ket \phi$ is a stabilizer state.
Here we show that any such state $\ket \psi$ can be distilled.
An example of such a state $\ket \psi$ is the magic state to produce a CCZ gate, 
which is equivalent to the Toffoli gate up to a Hadamard on the target~\cite{Shor1996Fault-tolerant}.
\begin{align}
 \ket \psi = ~ \underbrace{^{CC}Z}_{U} 
 \underbrace{\left( \frac{\ket 0 + \ket 1}{\sqrt{2}} \right)^{\otimes 3}}_{\ket \phi} .
\end{align}

As $\ket \phi$ is a stabilizer state,
we can identify $q$ operators, $S(1), S(2), \ldots, S(q)$,
which are products of Paulis, generating the stabilizer group of $\ket \phi$,
so that $\ket \phi$ is the unique (up to global phase) 
$+1$ eigenstate of those operators.
For the CCZ, we see $S(1)=X_1$, $S(2)=X_2$, and $S(3)=X_3$.
Then, the state $\ket \psi$ is the unique $+1$ eigenstate of the
operators
\begin{align}
 W(a) \equiv U S(a) U^\dagger, \text{ for } a = 1,\ldots, q.
\end{align}
These operators $W(a)$ commute with each other by construction,
and belong to the second level of the hierarchy, the Clifford group.
For the CCZ, we see $W(1) = X_1 (^C Z)_{23}$, $W(2) = X_2 (^C Z)_{13}$, and $W(3) = X_3 (^C Z)_{12}$.

Here is an example protocol for CCZ state distillation using three copies of the $[[4,2,2]]$ code, 
comprising a $[[12,6,2]]$ code.
We regard the three copies as a single $[[12,6,2]]$ code
and index the logical qubits by $1,\ldots,6$.
We encode one CCZ state stabilized by $W(a)$ into logical qubits $1,3,5$ 
and another stabilized by $W'(a)$ into logical qubits $2,4,6$, where $a = 1,2,3$.
Consider a variant of Eq.~\eqref{eq:CSwap-to-CProduct}
\begin{align}
 &[^C(\Sw_{12}\Sw_{34}\Sw_{56})] (W(a)_{135} \otimes I_{246}) [^C(\Sw_{12}\Sw_{34}\Sw_{56})] \nonumber \\
 &= [^C(W(a)_{135} \otimes W(a)_{246})] (W(a)_{135} \otimes I)
\end{align}
where the control qubit is common for every gate.
The simultaneous control-$\Sw$s are implemented by the control-$H^{\otimes 12}$ on the $[[12,6,2]]$ code,
where the control-$H^{\otimes 12}$ in turn is implemented by noisy $T$ gates.
Thus, we obtain a measurement routine for $W(a) \otimes W'(a)$.%
\footnote{
Since $W(1),W(2),W(3)$ are the same up to permutations of qubits,
the measurement routine can in fact measure any product $W(a)\otimes W'(b)$ for $a,b=1,2,3$ on the pair of CCZ states.
}
Then the protocol is to measure $W(1)W'(1)$, $W(2)W'(2)$, and $W(3)W'(3)$.
Overall, this protocol takes 2 noisy CCZ states
and $3\times (12 \times 4) = 144$ noisy $T$ gates
to produce 2 CCZ states at a lower error rate.
Note that this explicit example protocol 
performs worse than existing ones in terms of 
non-Clifford gate count~\cite{Eastin2013,Jones2013Tof2,CampbellHoward2017}.

By applying the Clifford stabilizers $W$ uniformly at random
to a noisy magic state for CCZ, it becomes a mixture of eigenstates of $W$'s.
Hence we may assume an error model 
where an error flips at least one of $W(1),W(2),W(3)$ with probability $\epsilon$.
Since the measurement routine puts measurement error at rate $O(\epsilon^2)$
and logical error at rate $O(\epsilon^2)$,
the protocol achieves quadratic error reduction for CCZ state.
For higher order reduction,
one should use inner and outer code of higher distances.

A related discussion on the error model for the $T$ state is given in Section~\ref{sec:magicStateFidelity}.

In passing, we note that the Clifford unitary 
$ V_{123} = ~(^{C_1}Z_2) (^{C_1}X_2)  X_2 X_3 $
on three qubits $1,2,3$ is a stabilizer of the CCZ state 
$\ket{ \text{CCZ}} = \frac{1}{\sqrt{8}}\sum_{a,b,c=0,1} (-1)^{abc} \ket{abc}$.
% It is verified as
% \begin{align}
% V_{123} \sum_{a,b,c} (-1)^{abc} \ket{abc}
% &=\sum_{a,b,c} (-1)^{abc}(^{C_1}Z_2) (^{C_1}X_2)   \ket{a,b+1,c+1} & \nonumber \\
% % &=\sum_{a,b,c} (-1)^{abc}(^{C_1}Z_2) \ket{a,a+b+1,c+1} & \nonumber \\
% &=\sum_{a,b,c} (-1)^{abc}(-1)^{a(a+b+1)} \ket{a,a+b+1,c+1} & \nonumber \\
% &=\sum_{a,b,c} (-1)^{abc+ab} \ket{a,a+b+1,c+1} & a(a+1)=0 \mod 2\nonumber \\
% &=\sum_{a,b,c} (-1)^{abc} \ket{a,b,c} & b\to b+a+1,~c \to c+1 .
% \end{align}
Since the CCZ state is permutation invariant,
we obtain six such stabilizers.
They do not commute, 
but any triple of them uniquely determines the CCZ state.
The controlled version can be implemented with only four $T$ gates~\cite{Eastin2013} 
(See also \cite{Jones2013Tof2}):
$
^{C_0} V_{123} =  T_2 (^{C_1} Z_2) T_2 (^{C_0}Z_2) T_2^\dagger (^{C_1}Z_2) T_2^\dagger (^{C_0} X_3) .
$
It might be possible to use these stabilizers 
with normal codes such as $[[7,1,3]]$ and $[[17,1,5]]$,
but because they do not commute the resulting measurement routine rejects faulty inputs with a probability less than 1,
even in the limit $\epsilon \to 0$.

\section{Inner Codes}
\label{sec:innercodes}

In this section,
we find a general class of inner codes that can be used in distillation protocols.
On the first read, a reader may wish to skip the discussion on symmetric forms,
noting only the magic basis in Definition~\ref{def:magicBasis},
and the construction of codes in Theorem~\ref{thm:magicCSScodes}.

\subsection{Symmetric forms over $\FF_2$}

We consider finite dimensional vector spaces over the binary field $\FF_2$.
The space $\FF_2^n$ is equipped with a symmetric dot product $v \cdot w = \sum_{i} v_i w_i \in \FF_2$.
This dot product on $\FF_2^n$ is \emph{non-degenerate}, 
i.e., for any nonzero vector $v \in \FF_2^n$
there is a vector $w \in \FF_2^n$ such that $v \cdot w \neq 0$.
Let $\calS$ be a null (self-orthogonal) subspace of $\FF_2^n$, 
on which the dot product identically vanishes.
Since $\calS$ is null, 
the dot product of $\FF_2^n$ canonically induces a dot product on the quotient space $\FF_2^n / \calS$
by $[v] \cdot [w] := v \cdot w$ where $[v]$ and $[w]$
denote the equivalence classes (members of the quotient space)
represented by $v$ and $w$, respectively.
Let $\calS^\perp$ denote the orthogonal complement of $\calS$ with respect to the dot product.
\begin{lemma}
The induced dot product on $\calS^\perp / \calS$ is non-degenerate.
\label{lem:nondegeneracy}
\end{lemma}
\begin{proof}
First, we claim that $(\calS^\perp)^\perp = \calS$.
It is clear by definition that $\calS \subseteq (\calS^\perp)^\perp$.
Interpreting the orthogonal complement as the solution space of a system of linear equations,
we see that the claim holds by dimension counting.
For $[v] \in \calS^\perp/\calS$, if $[v] \cdot [w] =0$ for any $w$,
then $v$ belongs to $(\calS^\perp)^\perp = \calS$,
implying that $[v] = 0 \in \calS^\perp / \calS$. 
\end{proof}
For any basis $\{ [v^{(1)}], \ldots, [v^{(k)}] \}$ of $\calS^\perp / \calS$,
we consider the symmetric matrix $\Lambda$ representing the dot product:
\begin{align}
 \Lambda^{ab} = v^{(a)} \cdot v^{(b)}.
\end{align}
Lemma~\ref{lem:nondegeneracy}
is equivalent to saying that the matrix $\Lambda$ is non-singular.
Any basis change of $\calS^\perp / \calS$ induces a congruent transformation $\Lambda \to M^T \Lambda M$ 
where $M$ is the invertible matrix of the basis change.
We consider equivalence classes of $\Lambda$ under the congruent transformations.
\begin{lemma}[Classification of symmetric forms over $\FF_2$]
A non-degenerate symmetric form over $\FF_2$ is equivalent to one of the two non-equivalent choices:%
\footnote{
Any finite dimensional symmetric forms 
over the binary field
is the intersection form of a closed topological surface $X$,
defined by the cup product $H^1(X;\FF_2) \times H^1(X;\FF_2) \to H^2(X;\FF_2) \cong \FF_2$ of cohomology.
Indeed, $I_n$ corresponds to the connected sum of $n$ copies of
$\mathbb R P^2$, and $\lambda_n$ the connected sum of $n/2$ copies of 2-torus.
The classification in the lemma is the orientability of a given surface, and Eq.~\eqref{eq:2RP2-to-torus}
is expressing the fact that two $\mathbb R P^2$'s can be turned into a torus in the presence of another $\mathbb R P^2$.
We thank Michael Freedman for pointing this out.
}
\begin{align}
I_n = \begin{pmatrix} 1 & & \\ & \ddots & \\ & & 1 \end{pmatrix}, \quad
\lambda_n = I_{n/2} \otimes \begin{pmatrix} 0 & 1 \\ 1 & 0 \end{pmatrix} .
\end{align}
\label{lem:dot-classification}
\end{lemma}
\begin{proof}
The two options are not equivalent since $\lambda_n$ means that the every vector is self-orthogonal,
whereas $I_n$ implies that not every vector is self-orthogonal.
For completeness, 
we give an elementary algorithmic proof by manipulating symmetric matrices.

First, we claim that any symmetric matrix can be brought to a direct sum of $I_p$
and $\lambda_q$ for some $p \ge 0$ and $q \ge 0$, where $q$ is even.
If there is a nonzero diagonal element one can bring this to the top-left by permutation.
Gaussian elimination on the first column and row reveals that $I_1$ is a direct summand.
Induction gives a direct summand $I_p$, and we are left with a symmetric matrix $\Lambda'$ with the zero diagonal.
Any column cannot be zero since $\Lambda'$ is non-singular,
and thus some permutation brings $1$ to $(2,1)$ and $(1,2)$ entries of $\Lambda'$.
Gaussian elimination on the first and second columns and rows reveals a direct summand $\lambda_2$.
By induction, our first claim is proved.

The second claim is that $I_{p+2} \oplus \lambda_{q-2} \cong I_{p} \oplus \lambda_{q}$ whenever $p,q > 0$,
whose proof is immediate:
\begin{align}
 \begin{pmatrix}
 	1 & 1 & 1 \\ 1 & 1 & 0 \\ 1 & 0 & 1
 \end{pmatrix} 
 \begin{pmatrix}
 	1 &  & \\
 	& 1 & \\
 	& & 1
 \end{pmatrix}
 \begin{pmatrix}
 	1 & 1 & 1 \\ 1 & 1 & 0 \\ 1 & 0 & 1
 \end{pmatrix}
 =
 \begin{pmatrix}
 	1 & 0 & 0 \\
 	0 & 0 & 1 \\
 	0 & 1 & 0
 \end{pmatrix}
\label{eq:2RP2-to-torus}
\end{align}
Therefore, whenever $p > 0$, we have $I_p \oplus \lambda_q \cong I_{p+q}$.
If $p = 0$, there is nothing more to prove.
\end{proof}

The classification motivates the following notion of bases.
\begin{definition}\label{def:magicBasis}
Given a null subspace $\calS \subseteq \FF_2^n$,
a basis of $\calS^\perp / \calS$ is called {\bf $(p,q)$-magic}
if the symmetric matrix $\Lambda$ representing 
the dot product on $\calS^\perp / \calS$ among the basis vectors
is equal to $I_p \oplus \lambda_q$ for some $p \ge 0$ and $q \ge 0$.
We call a magic basis {\bf normal} if $q = 0$, or {\bf hyperbolic} if $p =0$.
\end{definition}

We summarize the results of this section into a theorem.
\begin{theorem}\label{thm:magicBasis}
For any self-orthogonal subspace $\calS \subseteq \FF_2^n$,
there exists a $(p,q)$-magic basis for $\calS^\perp / \calS$, 
where $p + q = \dim_{\FF_2} \calS^\perp / \calS$.
If $p>0$ and $q>0$, then a $(p+2,q-2)$-magic basis exists.
\end{theorem}

\subsection{CSS codes from self-orthogonal matrices}

It is standard to associate a bit string $v = (v_1, \ldots, v_n)$ to a Pauli operator:
$X(v) = X_1^{v_1} \cdots X_n^{v_n}$ where $X_j$ is the Pauli $\sigma^x$ on qubit $j$,
and $Z(v) = Z_1^{v_1} \cdots Z_n^{v_n}$ where $Z_j$ is the Pauli $\sigma^z$ on qubit $j$.
The commutation relation is that
\begin{align}
 X(v) Z(w) = (-1)^{v \cdot w} Z(w) X(v).
\end{align}
The CSS construction of quantum codes applies to a self-orthogonal (null) subspace $\calS \subseteq \FF_2^n$:
For every vector $v \in \calS$, we define an $X$-stabilizer $X(v)$, as well as $Z$-stabilizer $Z(v)$.
The set of equivalence classes of $X$-type ($Z$-type) logical operators
is then in one-to-one correspondence with $\calS^\perp / \calS$.
The number of logical qubits is thus 
\begin{align}
k = \dim_{\FF_2} \calS^\perp / \calS = n - 2 \dim_{\FF_2} \calS.
\end{align}

We encode logical qubits 
by choosing a complete set of logical operators $\tilde X^{(j)}$ and $\tilde Z^{(j)}$ as follows.
Choose a $(p,q)$-magic basis $\{ v^{(1)}, \ldots, v^{(p)}, w^{(1)}, \ldots, w^{(q)} \}$ of $\calS^\perp/\calS$.%
\footnote{
Here we have abused notation to denote an equivalence class (a member of the quotient space) by a representative.
}
Then, we define
\begin{align}
&\begin{cases}
 \tilde X^{(i)} &= X(v^{(i)})\\
 \tilde Z^{(i)} &= Z(v^{(i)}) 
\end{cases}
&\text{ for } i = 1,\ldots, p, \nonumber
\\
&\begin{cases}
 \tilde X^{(p+2j-1)} &= X(w^{(2j-1)})\\
 \tilde Z^{(p+2j-1)} &= Z(w^{(2j)})\\	
\end{cases}
\quad \begin{cases}
 \tilde X^{(p+2j)} &= Z(w^{(2j-1)})\\
 \tilde Z^{(p+2j)} &= X(w^{(2j)})	
\end{cases}
&\text{ for } j = 1,\ldots,q/2. \label{eq:DefLOP}
\end{align}
By definition of the magic basis,
these logical operators obey the canonical commutation relation of Pauli operators on $k$ qubits:
\begin{align}
  \tilde X^{(a)} \tilde Z^{(b)}= (-1)^{\delta_{ab}} \tilde Z^{(b)} \tilde X^{(a)} .
\end{align}
Note that the commutation relation can be realized 
with arbitrary signs $\pm$ in the choice of the logical operators,
but induced Clifford logical operators will depend on the signs.
We enforce \eqref{eq:DefLOP} in order for the transversal Hadamard $\bar H = H^{\nin}$
to be the logical Hadamard $\prod_{a=1}^{\kin} \tilde H^{(a)}$.

We have defined CSS codes based on self-orthogonal subspaces over $\FF_2$:
\begin{theorem}
Let $\calS \subseteq \FF_2^n$ be a self-orthogonal subspace 
with a $(p,q)$-magic basis of $\calS^\perp / \calS$.
Then, there exists a CSS code on $n$ qubits with $p+q$ logical qubits
and a choice of logical operators such that transversal Hadamard $H^{\otimes n}$
implements the logical Hadamards for the logical qubits $1,\ldots,p$, and simultaneously 
the swaps between the logical qubit $p+2j-1$ and $p+2j$ where $j = 1,\ldots, q/2$.
\label{thm:magicCSScodes}
\end{theorem}
We will call a weakly self-dual CSS code normal if a normal magic basis exists,
and hyperbolic otherwise.
It is possible for a normal code 
to have an even number of logical qubits, an even number of physical qubits,
and an even distance.
Every hyperbolic code, however, must have an even number of logical qubits,
an even number of physical qubits, and an even distance.
For instance, in Sec.~\ref{sec:basic},
the Steane code $[[7,1,3]]$ and
the $[[17,1,5]]$ color code are normal.
The $[[4,2,2]]$ code and the $[[16,6,4]]$ are hyperbolic.
We have used $(0,2)$-magic basis for $[[4,2,2]]$,
and $(0,6)$-magic basis for $[[16,6,4]]$.
The ``$H$-code'' by Jones~\cite{Jones2012}
is a normal code with parameters $[[k+4,k,2]]$ where $k$ is even.
Below, we will mostly use normal codes with $(p,0)$-magic basis
for distillation protocols.

We note that the CSS codes derived from a self-orthogonal matrices
are not too restrictive.
By representing each qubit in any stabilizer code of parameters $[[n,k,d]]$ by Majorana modes,
we obtain a weakly self-dual CSS code of parameters $[[4n,2k,2d]]$~\cite{BravyiLeemhuisTerhal2010Majorana}.
We will briefly review this mapping in Section~\ref{innerconstruct},
where we will also present other families of such codes with improved rate.

\section{Coding Theory and Asymptotic Performance}
\label{sec:asymptotic}

\subsection{Asymptotic Performance}
In this section we consider the asymptotic properties of the class of protocols defined above, 
for appropriate choice of inner and outer codes.  
We ignore all possibilities of pipelining,
and use only a single inner and outer code to define each protocol;
this will reduce the question of asymptotic properties 
to the question of the existence of code families with certain properties.

``Asymptotic" will refer to one of two limits.
In the first limit, we consider a family of protocols parametrized by $d$,
the order of reduction in error.
An instance in the family reduces error probability from 
$\epsilon$ to a constant times $\epsilon^d$ in the limit of small $\epsilon$.
We prove that
\begin{theorem}
\label{thm:dasymptotic}
There is a family of protocols parametrized by an integer $d \ge 1$
to obtain a $d$-th order reduction in error,
using a total of $\Theta(d)$ physical qubits,
producing $\nout = \Theta(d)$ magic states.
The total number of $T$ gates used is $\nT = \Theta(d^2)$,
so that the number of $T$ gates per magic state is $\Theta(d)$.
The $T$-gate depth of the circuit is also $\Theta(d)$, where the $T$-gate depth refers to the circuit depth assuming that an arbitrary Clifford can be executed in depth $1$.
\end{theorem}
In the second limit, we fix $d$ and consider a family of protocols parametrized by $\nout$, the number of magic states produced.
We prove that
\begin{theorem}
\label{thm:noutasymptotic}
For any odd $d\geq 5$, there is a family of protocols 
using $\nout \cdot (1+o(1))$ physical qubits, producing $\nout$ magic states 
with a $d$-th order reduction in error.  The total number of $T$ gates used is
\begin{align}
\nT= d(1+o(1)) \nout. \label{idealasymptotic}
\end{align}
\end{theorem}
The reason $d$ is odd is that the minimal weight 
of an error that is not caught in the protocol 
due to wrong $H$-measurement outcomes is always odd.

Given one particular protocol with 
$\epsilon_\text{out} = C \epsilon_\text{in}^d$ consuming $\nT / \nout$ $T$ gates per output,
an infinite family of protocols can be defined by concatenation with itself.
For this concatenated family,
the number of $T$ gates to achieve an arbitrarily small error rate $\delta$ in output magic states
scales like $O( (\log 1/\delta )^\gamma)$ where the scaling exponent~\cite{BravyiKitaev2005Magic,BravyiHaah2012Magic} 
is
\begin{align}
\gamma = \log_d( \nT / \nout ).
\end{align}
Smaller values of $\gamma$ reflect asymptotically more efficient distillation protocols.
The triorthogonal codes~\cite{BravyiHaah2012Magic} achieve $\gamma \rightarrow \log_2(3)$,
and ``multilevel'' protocol~\cite{Jones2012} achieves $\gamma \to 1^+$. 
We will comment on this multilevel protocol in Discussion~\ref{sec:discussion}.
It was conjectured that no protocol could achieve $\gamma<1$~\cite{BravyiHaah2012Magic}.
Both families in Theorems~\ref{thm:dasymptotic} and \ref{thm:noutasymptotic} achieve $\gamma \to 1^+$.

We note that the measure $\gamma$ slightly underestimates 
the $T$-count efficiency of the family in Theorem~\ref{thm:dasymptotic}.
In order to achieve an arbitrary small final error rate $\delta$ from a fixed initial error rate, say, $\epsilon = 0.01$,
we can pick a member $P_d$ of the family of error reduction degree $d$ such that $\delta > C_d (\epsilon/2d^2)^d$.
Here $C_d$ is the leading coefficient of the output error probability of the protocol $P_d$,
which is at most the number of ways that weight $d$ errors occur among $\nT = O(d^2)$ $T$ gates;
$C_d \le \alpha d^{2d}$ for some $\alpha > 0$ independent of $d$.
For the condition $\delta > C_d(\epsilon/d^2)^d$, it suffices that $d > (\log(1/\delta) + \log \alpha)/\log(1/\epsilon)$.
We initially distill magic states to suppress the error rate from $\epsilon$ to $\epsilon' = \epsilon/d^2$,
by using a concatenated protocol $P_\text{init}$.
This takes $n_\text{init} = O(\log d)^\gamma$ input magic states per output magic states for some $\gamma > 1$.
We can then feed $P_d$ with the outputs from $P_\text{init}$ at error rate $\epsilon'$.
It follows that
\begin{align}
 \nT / \nout = O(d) \cdot n_\text{init} = O(\log(1/\delta) (\log \log 1/\delta)^\gamma ) 
 \end{align}
 magic states at error rate $\epsilon$ suffice to achieve final accuracy $\delta$.
Thus, the scaling of $\nT / \nout$ is linear in $\log (1/\delta)$ up to a logarithmic correction.
(One can iterate the argument recursively to further slow down the dependency on $1/\delta$.)

Theorem \ref{thm:noutasymptotic} will use normal codes.  The reduced number of $T$ gates required to implement checks with a normal code is essential to obtaining the number of $T$ gates in the theorem (we would need roughly twice as many using hyperbolic codes).  This explains why $d$ is chosen odd.  The case of $d=1$ is of course trivial: no codes are needed.  Thus, the reader may wonder why the case $d=3$ is not used; this is explained further below. 

Theorems \ref{thm:dasymptotic},~\ref{thm:noutasymptotic} will follow almost immediately given certain families of inner and outer codes obeying certain properties of the codes that we define below.
We will prove these theorems given these properties in this subsection and we construct families of inner and outer codes with these properties in subsections \ref{innerconstruct}, \ref{outerconstruct}.

Consider first the inner code.  This code will have
$\kin$ logical qubits and $\nin$ physical qubits.  The distance of the inner code will be at least $d$.
Consider then the effect of errors in the $T$ gates inside the inner code; i.e., in the $T$ gates acting on the encoded state.  To obtain $d$-th order reduction in error, it suffices to consider the case that fewer than $d$ errors occur in such $T$ gates.  Since the inner code distance is at least $d$, these errors cannot produce a logical error.
There is one way, however, in which these errors can have an effect without being detected by the inner code.  It is possible that a pair of errors act inside the inner code, both on $T$ gates acting on the same qubit.  The effect of
these errors is to cause an error in the check being measured by the inner code, i.e., if the check was measuring a given product of $W$ operators specified by the outer code, we instead measure the opposite sign; we call this a ``measurement error".

The possibility of measurement errors affects some of the properties that we require of the outer code.
Of course we need the outer code to have distance at least $d$, as otherwise a pattern of fewer than $d$ errors in the input magic states could cause an undetectable error, but this is not sufficient.  It is necessary that
a pattern of fewer than $d$ errors causes enough checks to be violated so that even a small number of measurement errors
will lead to an error detected by the code.  This is defined by the property of ``sensitivity" that we now define.

The outer code will have $m$ parity checks, encoded in an $m$-by-$\nout$ parity check matrix $M$, where each
row of the matrix indicates a given check.
We can measure rows of this matrix with even weight 
using an hyperbolic inner code and rows with odd weight using a normal inner code.
For simplicity we will either have all rows have even weight 
or have all rows use odd weight so that we can use the same inner code for all checks.
(More generally, one could use both a hyperbolic and a normal code.)
Then, this inner code must have $\kin$ greater than or equal to the maximum row weight of $M$.
The difference between the row weight of $M$ and $\kin$
must be an even number.
In this case, we say that the inner code can implement the checks of the outer code.
\begin{definition}
An $m$-by-$\nout$ parity check matrix $M$ for a classical linear code is said to be
{\bf $(\tilde d,s)$-sensitive} if any nonzero bit vector $v$ of length $\nout$
with $|v| \leq \tilde d$, we have
$|M v| \geq s$.  That is, for any such vector, the number of violated parity checks is at least $s$.
\end{definition}
We emphasize that sensitivity is a property of the check matrix of the outer code, 
rather than the codewords of the outer code, and in some examples the rows of the check matrix may be linearly dependent.
A $(\tilde d,s)$-sensitive parity check matrix is $(\tilde d -1,s)$-sensitive by definition.

\begin{lemma}
\label{senslemma}
Given an $m$-by-$\nout$ parity check matrix $M$  
such that $2|Mv|+|v| \ge d$ for any nonzero $v$
(e.g. $(d-1,\frac{d-1}{2})$-sensitive $M$),
and given an inner code of parameters $[[\nin,\kin,d]]$
that can implement the checks defined by $M$,
the protocol yields $d$-th order reduction in error.
% Given an outer code which is $(\tilde d,s)$-sensitive, with $\tilde d\geq d-1$ and $s\geq (d-1)/2$, and given an inner code with distance $d$ that is able to implement the checks defined by that outer code, the protocol yields
% a $d$-th order reduction in error.
The protocol overall takes $\nout$ noisy magic states,
$2\nin m$ noisy $T$ gates when the inner code is normal,
or $4\nin m$ when hyperbolic,
and outputs $\nout$ magic states.
\end{lemma}
\begin{proof}
Any error pattern inside the inner codes with weight less than $d$ cannot cause a logical error.
Thus, if an error pattern inside the inner code does not violate a stabilizer of the inner code, 
it either has no effect or it leads to an error in measurement of a check of the outer code;
the latter possibility requires at least two errors inside the inner code.
Any input state with $|v| \ge 1$ errors
will violate at least $|Mv| \ge (d - |v|)/2$ checks of the outer code.
If no violation of these checks is detected, 
there must be at least $2|Mv|$ errors on $T$ gates inside the inner code.
Thus, there must be at least $d$ errors in total.

The input $T$ gate and state count is clear from Section~\ref{sec:basic}.
\end{proof}

We now define some asymptotic properties of the codes needed.
\begin{definition}
A family of quantum error correcting codes with increasing number of qubits $n$ has {\bf good rate}
if the number of encoded qubits $k$ is $\Theta(n)$ and has good distance if the distance $d$ is $\Theta(n)$.
\end{definition}

\begin{definition}
Given a family of outer codes with increasing $\nout$, we say that this family has {\bf good sensitivity} if
each code in the family is $(\tilde d,s)$-sensitive for $\tilde d=\Theta(\nout)$ and $s=\Theta(\nout)$.
\end{definition}

\begin{definition}
Given a family of outer codes with increasing $\nout$, we say that this family has {\bf good check rate}
if the parity check matrix is $m$-by-$\nout$ with
$m=\Theta(\nout)$.
\end{definition}

\begin{proof}[Proof of Theorem~\ref{thm:dasymptotic}]
In subsection~\ref{innerconstruct}, 
we show that families of both hyperbolic and normal inner codes 
with good rate and distance exist and 
in subsection~\ref{outerconstruct} 
we show that families of outer codes with good check rate, good sensitivity, and even row weight exist.
Combining these results with Lemma~\ref{senslemma}, Theorem~\ref{thm:dasymptotic} follows.
\end{proof}

\begin{proof}[Proof of Theorem~\ref{thm:noutasymptotic}]
In subsection \ref{innerconstruct}, we show that, for any $d$,
there exist families of both hyperbolic and normal inner codes with increasing $\nin$ such that $\kin/\nin\rightarrow 1$.
To prove this theorem, we will only need the result for normal inner codes.
Consider some code from this family with given $\kin,\nin$.
In subsection \ref{outerconstruct} we show Lemma~\ref{bulkoc} which we reproduce here:
\begin{lemma*}
Given integers $\tilde d, w \ge 1$ and $s \ge 2$,
there exists an $m \times \nout$ parity check matrix $M$ that is $(\tilde d, s)$-sensitive
where $m = \nout \cdot s / w$ and every row of $M$ has weight $w$ exactly.
\end{lemma*}

Choose $w=\kin$.
Choosing $\tilde d\geq d-1$ and $s= (d-1)/2$, this gives us an outer code such that the checks can be performed by the given inner code and we need to perform $\nout (s/w)=\nout(s/\kin)$ checks with the inner code.
Each such check with the inner code requires using $2 \nin$ $T$ gates, so that the total number of $T$ gates needed to perform the checks with the inner code is equal to $2 \nout s (\nin/\kin)$.  Additionally, we need to perform $\nout$
$T$ gates to create the input magic states to the outer code.  Thus, the total number of $T$ gates is
\begin{align}
\nT&= \nout (1+2s \nin/\kin) \nonumber \\
&= \nout (d + (d-1)(\nin/\kin-1)).
\end{align}
Taking $\nin$ large so $\nin/\kin \rightarrow 1$, we conclude $\nT \rightarrow \nout d$. 
\end{proof}

We can now see better why we needed $d\geq 5$ in Theorem~\ref{thm:noutasymptotic}.
This is because for $d=3$, 
we have $s=1$ and Lemma~\ref{bulkoc} does not apply.
The reader will see later why the case $s=1$ is excluded from that lemma;
roughly, this is because in this case, 
each bit participates in only a single check 
and we would lack certain expansion properties for a certain graph defined later.

\subsection{Inner Codes}
\label{innerconstruct}
In this subsection, we give asymptotic constructions of inner codes.
There are at least two constructions of weakly self-dual codes with good rate and distance in the literature.  We review these, before giving an alternative probabilistic proof which has some advatanges.

First, in Ref.~\cite{CalderbankShor1996Good}, it is shown that given any ratio $d/n$, one can find a family of weakly self-dual CSS codes with $n$ qubits and distance $d$ and given ratio $d/n$ achieving a rate
$k/n \rightarrow 1-2H_2(d/n)$, where $H_2$ is the binary entropy function.
The codes found in that paper all are hyperbolic codes.  However, we can obtain normal codes from them by a ``puncturing procedure" (see also \cite[Sec.~3.5]{gottesmanthesis}):
\begin{definition}
Given a hyperbolic weakly-self-dual CSS code $C$ on $n$ qubits with $k$ logical qubits, define a ``punctured code" $C'$ as follows.  Choose a qubit $i$ (the code $C'$ may depend upon the choice of $i$).  Write the stabilizer generators of $C$ such that only one $X$-type and one $Z$-type generator is supported on $i$.  Define $C'$ by removing qubit $i$ and removing the stabilizer generators support on $i$.  Then $C'$ has $n'=n-1$ qubits and $k+1$ logical qubits.  The code $C'$ is a normal code by construction.
\end{definition}
If $C$ is non-degenerate with distance $d$, then $C'$ has distance $d' \geq d-1$.  More generally, $d'+1$ is greater than or equal to the minimum weight of
an operator which commutes with the stabilizer group of $C$, because given an $X$-type logical operator $O$ in $C'$ then either $O$ or $O X_i$ must commute with the stabilizer group of $C$.
Indeed, one may show that puncturing the codes of Ref.~\cite{CalderbankShor1996Good} reduces the distance by at most $1$.

The only disadvantage of this proof is that it is a greedy proof that we not know how to implement efficiently.
While it would be desirable to find an explicit family of codes achieving this rate, we do not know how to do this.  However, another construction in the literature is a randomized construction
which allows us to give codes which, with high probability, have the desired distance.  Unfortunately, this construction
will only achieve
$k/n \rightarrow 1/2-H_2(2d/n)$.
This construction uses 
a general method to construct weakly self-dual CSS codes in Ref.~\cite{BravyiLeemhuisTerhal2010Majorana}.

Consider a stabilizer code $C_{qubit}$ which acts on $n$ physical qubits and has $k$ logical qubits and distance $d$.
From this code, one can derive a code for Majorana fermions $C_{Majorana}$ which acts on $4n$ Majorana modes 
and has $k$ logical qubits and distance $2d$, 
where now the distance refers to minimum weight of a product of Majorana operators that is a logical operator.
The code $C_{Majorana}$ is derived in the following way: 
For each physical qubit of $C_{qubit}$, one introduces four Majorana modes, 
$\gamma_0,\gamma_1,\gamma_2,\gamma_3$, 
and declares that the product $\gamma_0 \gamma_1 \gamma_2 \gamma_3$ is a stabilizer of $C_{Majorana}$.
For each stabilizer of $C_{qubit}$, 
one defines a stabilizer of $C_{Majorana}$ 
by replacing an operator 
$X$ on a qubit by $i \gamma_0 \gamma_1$,
$Y$ by $i \gamma_0 \gamma_2$,
and $Z$ by $i \gamma_0 \gamma_3$.
The stabilizer generators of $C_{Majorana}$
are given by bit strings of length $4n$ such that the dot product over $\FF_2$ 
of any pair of such bit strings is $0$. 
Thus, from $C_{Majorana}$, one can define a weakly self-dual CSS code $C_{wsd}$ with $4n$ physical qubits,
$2k$ logical qubits and distance $2d$.
Since a randomized construction (see, for example, Eq.~7.200 of Ref.~\cite{preskill})
gives stabilizer codes $C_{qubit}$ with
$k/n \rightarrow 1-H_2(d/n)-(d/n) \log_2(3)$, 
mapping these stabilizer codes $C_{qubit}$ to weakly-self dual codes $C_{wsd}$ gives 
$k/n \rightarrow (1/2)[1-H_2(2d/n)-(2d/n)\log_2(3)]$.
Since the randomized construction gives a lower bound to the weight 
of any operator commuting with the stabilizer group, 
we can puncture these codes and reduce the distance by at most $1$.

Here we give another proof of the existence of such good weakly self dual-codes. 
This will lead to rate $k/n \rightarrow 1-2 H_2(d/n)$. 
For any fixed distance $d$, 
one can obtain families of stabilizer codes with 
$n$ physical qubits and $k$ logical qubits with the ratio $k/n \rightarrow 1$ as $n \rightarrow \infty$.
While this improvement is only by constant factors over the construction via Majorana codes,
it will lead to nice asymptotic expressions for the number of $T$-gates, $\nT$, 
required to attain $d$-th order suppression in error. 
It is also a randomized construction, 
showing that codes in a certain ensemble have the desired properties with high probability.

Define a random ensemble of $c$-by-$n$ self-orthogonal matrices as follows, where a matrix $M$ is defined to be self-orthogonal if $M M^T=0$.
Choose the first row of the matrix to be the all $1$s vector $\vec 1$. 
Choose the second row uniformly at random subject to the constraint 
that it have vanishing dot product with the first row.
Continue in this fashion, 
choosing the $j$-th row uniformly at random subject to the constraint 
that it have vanishing dot product with the first $j-1$ rows.
(Remark: the requirement that the first row be the all $1$s vector 
is simply chosen to simplify some notation, 
so that we do not need to add the requirement that each row have even weight.)

\begin{lemma}
\label{rwsdlemma}
Consider a fixed $n$-component vector $v$, 
with $v \neq 0$ and $v\neq \vec 1$.
For a random $c$-by-$n$ self-orthogonal $M$,
the probability that $M v=0$ 
is at most $2^{-c+1}+2^{-n+c+1}$.
\end{lemma}
\begin{proof}
Let $w_1,\ldots,w_c$ be the rows of $M$.
Let $V_j$ be the self-orthogonal subspace which is the span of the first $j$ rows of $M$.
We will estimate the desired probability by a union bound, considering separately the event that $v\in V_c^\perp$ and $v\not\in V_c$, and the event that $v\in V_c^\perp$ and $v\in V_c$.

Consider the first event.
Let $j>1$.  Then
\begin{align}
\label{est1}
\Pr[w_j \cdot v =0 | v \not \in V_{j-1}]=\frac{1}{2},
\end{align}
because the constraint that $(v,w_j)=0$ is independent of the constraints on the vector $w_j$.
Thus, for any $k$,
\begin{align}
\label{estk}
\Pr[v \in V_k^\perp \, {\rm and} \, v \not \in V_k] \leq \prod_{j=2}^k \frac{1}{2}=2^{-k+1}.
\end{align}
For $k=c$, we find in particular that
\begin{align}
\label{estc}
\Pr[v \in V_c^\perp \, {\rm and} \, v \not \in V_c] \leq 2^{-c+1}.
\end{align}

Now we estimate the probability of the second event.  
Note that if $v \in V_c$, there is a least $j$ such that $v\in V_j$.  
So,
\begin{align}
\label{sumeq}
\Pr[v \in V_c^\perp \, {\rm and} \, v \in V_c] 
\leq 
\sum_{j=2}^c 
\Pr[v\in V_j \, {\rm and} \, v \in V_{j-1}^\perp \, {\rm and} \, v \not \in V_{j-1}].
\end{align}
We have
\begin{align}
&\Pr[v \in V_{j} \, {\rm and} \, v\in V_{j-1}^\perp \, {\rm and} \, v \not \in V_{j-1}] \nonumber \\ 
&= 
\Pr[v \in V_{j-1}^\perp \, {\rm and}\, v\not \in V_{j-1}] \cdot  
\Pr [v\in V_j  ~|~ v \in V_{j-1}^\perp \, {\rm and}\, v\not \in V_{j-1}] \nonumber\\
&\leq 
2^{-j+2}  \Pr [ v \in V_j  ~|~ v \in V_{j-1}^\perp \, {\rm and}\, v\not \in V_{j-1}],
\end{align}
where we used Eq.~(\ref{estk}).

Now we estimate the probability  $\Pr [v\in V_j  ~|~ v \in V_{j-1}^\perp \, {\rm and}\, v\not \in V_{j-1}]$.
This is possibly nonzero only if $v \cdot v = 0$.
Consider the space of all $n$-component vectors modulo vectors in $V_{j-1}$; 
this quotient space has dimension at least $n-(j-1)$.
Let $\pi$ be the natural map from the space of all vectors to this quotient space.
The vector $\pi v$ is nonzero by assumption.
The vector $w_j$ is subject to at most $j-1$ independent constraints from $V_{j-1}$.
Consider the space of possible $\pi w_j$, given that $w_j$ obeys those constraints;
this space has dimension at least $n-2(j-1)$
and so the probability that a random vector in this space 
is equal to $\pi v$ is at most $2^{-(n-2j+2)}$.
Hence, $\Pr [v\in V_j  ~|~ v \in V_{j-1}^\perp \, {\rm and}\, v \notin V_{j-1}] \leq 2^{-(n-2j+2)}$,
so 
$\Pr[v \in V_j \, {\rm and} \, v\in V_{j-1}^\perp \, {\rm and} \, v \notin V_{j-1}] \leq 2^{-n+j}$.
So by Eq.~(\ref{sumeq}),
\begin{align}
\label{event2}
\Pr[v \in V_c^\perp \, {\rm and} \, v \in V_c] \leq \sum_{j=2}^c 2^{-n+j} \leq 2^{-n+c+1}.
\end{align}
By a union bound, adding probabilities in Eqs.~(\ref{estc},\ref{event2}), the lemma follows.
\end{proof}

\begin{lemma}\label{lem:randomcodes-distance}
Let $n,c,d$ be such that
\begin{align}
\label{ncd}
(2^{-n+c+1}+2^{-c+1}) \sum_{w=1}^d \binom{n}{w}  <1.
\end{align}
Then, there exists a $c$-by-$n$ matrix $M$ such that $M M^T=0$ 
and such that $M v \neq 0$ for any $v\neq 0$ with $v$ having Hamming weight at most $d$.
\begin{proof}
This follows from Lemma~\ref{rwsdlemma} and by a first moment bound.  
For a random $M$ from the above ensemble, 
the expected number of vectors $v\neq 0$ 
with Hamming weight at most $d$ such that $M v=0$ is at most
$(\sum_{w=1}^{d} \binom{n}{w}) \cdot (2^{-n+c+1}+2^{-c+1})$.
\end{proof}
\end{lemma}

\begin{lemma}
\label{fixedd}
For any fixed $d$, one can find a family of $M$ with increasing $n$ such that the ratio $c/n$ tends asymptotically to zero and such that Eq.~(\ref{ncd}) is obeyed.  
Hence, for any distance $d$, one can find a family of 
hyperbolic or normal weakly self-dual CSS codes 
such that the ratio $\kin/\nin \rightarrow 1$ as $\nin \rightarrow \infty$.
\begin{proof}
Immediate for the hyperbolic case.
Since the lemma \ref{rwsdlemma} upper bounds the probability that an operator commutes with the stabilizer group one can also puncture these codes to obtain a normal code.
\end{proof}
\end{lemma}

\subsection{Outer Codes}
\label{outerconstruct}
In this subsection, we construct families of outer codes with  good check rate and sensitivity.  We begin with a randomized construction, and then show how to construct explicit families using previous results in coding theory.

\begin{lemma}
\label{random}
There exist families of outer codes with good check rate and sensitivity and even row weight.
Similarly, there exist families of outer codes with good check rate and sensitivity and odd row weight.
\begin{proof}
We only give a proof for check matrices of even row weight.
The proof for odd row weight case is completely analogous.

Consider a random $m$-by-$\nout$ parity check matrix $M$.
Let $\tilde d = \nout-1$.
Choose each row independently but with the constraint that it should be of even weight. For any vector $v$ 
with
$|v| \leq \tilde d$, the syndrome vector $Mv$ has independent entries from the uniform distribution.
Thus, the probability that $|Mv| \leq s$ for $s\leq m/2$ is bounded by
\[
	2^{-m} \sum_{i \leq s} {m \choose i}=\sO(2^{m(H(s/m)-1)}),
\]
where $H(p)=-p \log_2(p)-(1-p) \log_2(1-p)$ is the binary entropy function, and $\sO$ hides polynomial factors.
The number of such vectors $v$ is bounded by $2^{\nout}$.
By a union bound, the probability that there is an error vector $v$ of weight less than $\tilde d$ such that the syndrome has weight less than $s$ is
bounded by
\[
\sO(2^{m(H(s/m)-1)})2^{\nout} =
\sO(2^{\nout (1 + (m/\nout)(H(s/m)-1))}).
\]
For sufficiently large ratio $m/\nout$ and sufficiently small ratio $s/m$,
this quantity is exponentially small in $\nout$.
\end{proof}
\end{lemma}

The above randomized construction is very similar to randomized constructions of classical codes with good rate and distance, where we define
\begin{definition}
A family of classical error correcting codes with increasing number of bits $n$ has {\bf good rate} 
if the number of encoded bits $k$ is $\Theta(n)$ and has {\bf good distance} if the distance $d$ is $\Theta(n)$.
\end{definition}
That is, even though we are considering very different properties (number of violated checks rather than distance of the code), the first moment argument above is very similar to standard first moment arguments to construct such codes with good rate and distance, with some additional technicalities required to ensure even weight of the parity checks.
This is not a coincidence.  As we now show, given a family of codes with good rate and distance, one can construct a family of codes with good check rate and sensitivity.

\begin{lemma}
\label{transpose}
Let $C$ be a classical error correcting code that encodes $k$ bit messages into $n$ bit codewords.  Let $C$ have distance $d$.  Let $v_1,\ldots,v_k$ be a basis for the codewords of $C$.
Let $M$ be the $n$-by-$(k+1)$ matrix
whose columns are the vectors $v_1,\ldots,v_k,w$ where $w=v_1+\ldots+v_k$.
Then, all rows of $M$ have even weight and
$M$ is a parity check matrix for a code with
$\nout=k+1$ bits which is $(\tilde d,s)$ sensitive with $s=d$ and $\tilde d=\nout-1$.  Thus, the code with parity checks encoded by $M$ has only two codewords (the all $0$ vector and the all $1$ vector) and any message which is not a codeword will violate at least $d$ checks.
\begin{proof}
For any $(k+1)$-bit vector $v$, the vector $M v$ is a codeword of $C$.
If $v$ is nonzero and is not equal to the all $1$ vector, then $Mv$ is a nonzero codeword of $C$ and hence has weight at least $d$.
\end{proof}
\end{lemma}
Since $\nout=k+1$, in order to obtain an even $\nout$, if $C$ has $k$ even, we can simply define a new code $C'$ which encodes $k-1$ bit messages into $n$ bit codewords by using any $(k-1)$-dimensional subspace of the codewords of $C$, in this way obtaining a parity check matrix for a code with $\nout=k-1+1=k$.

Using lemma \ref{transpose}, we can construct explicit families of codes with good check rate and good sensitivity given any explicit family of codes with good rate and good distance.  As an example of such a code family, we can use the expander codes of Ref.~\cite{sipser1996expander}.

\begin{lemma}
\label{bulkoc}
Given integers $\tilde d, w \ge 1$ and $s \ge 2$,
there exists an $m \times \nout$ parity check matrix $M$ that is $(\tilde d, s)$-sensitive
where $m = \nout \cdot s / w$ and every row of $M$ has weight $w$ exactly.
\end{lemma}
\begin{proof}
A parity check matrix $M$ defines a bipartite graph $G$, often called a Tanner graph.  
One set of vertices of the graph (which we call $B$ labeled by the columns of $M$)
corresponds to bits of the code 
and the other set (which we call $C$ labeled by the rows of $M$) corresponds to checks, 
with an edge between a pair of vertices 
if $M$ is nonzero in the corresponding entry.
Equivalently, given such a bipartite graph $G$, this defines a parity check matrix.
We claim that given a bipartite graph with all vertices in $B$ having degree $s$ 
and all vertices in $C$ having degree $w$ and with girth $> 2 \tilde d$, 
the corresponding parity check matrix defines a code with the desired properties.
Once we have shown this, the lemma follows, 
since Ref.~\cite{furedi1995graphs} shows the existence of such graphs.

Note first that the degree of vertices in $C$ corresponds to the row weight of $M$.
Next, note that if all vertices in $C$ have degree $w$ and all in $B$ have degree $s$,
then $m=|C|=\nout \frac{s}{w}$ with $\nout=|B|$.

To prove the claim, let $V \subseteq B$ be a nonempty set of erroneous bits.
By assumption, $1 \le |V| \le \tilde d$.
Consider a subgraph $H$ of $G$ defined by all vertices of $V$ and its neighbors.
By the girth condition on $G$, the subgraph $H$ has to be a collection of disjoint trees.
Thus, it suffices to prove the claim in case where $H$ is connected.
If $|V| = 1$, then the error violates $s$ checks, and we are done.
If $|V| \ge 2$, let $v_1, v_2 \in V$ be a pair that are the furthest apart.
The choice of the pair ensures that each of $v_1$ and $v_2$ has $s-1$ leaves attached to it.
Therefore, $V$ violates at least $2s-2 \ge s$ checks.
\end{proof}

Note that the ratio $m/\nout=s/w$ in lemma \ref{bulkoc} is the best possible, 
because each bit must participate in at least $s$ checks 
(i.e., every column of the parity check matrix must have weight at least $s$).
Note also that though the existence of desired graphs is guaranteed,
they might be too large in practice; 
$w^s \le \nout \le \mathrm{poly}(w^s)$~\cite{furedi1995graphs}.
However, one does not have to be too strict
on the biregularity of the graph in practice.
If small violation of the biregularity
gives a much smaller graph,
then it might be more useful.

\section{Numerical Simulation}
\label{sec:simulation}
In this section, we give results of numerical simulations.
We begin by explaining the error model we used for simulations.
We then explain two protocols that we simulate that are not explained previously; 
one of these protocols uses a $[[21,3,5]]$ code.  
Then we give the simulation results.   
One interesting result of the simulation is 
how little effect the subleading terms have, even at fairly large noise values.

\subsection{Magic state fidelity}
\label{sec:magicStateFidelity}

When we inject a magic state $\mu$ for a 
$\pi/4$ rotation into a quantum circuit,
there is a probability for correction $K$ by angle $\pi/2$ to be applied.
If we represent the overall procedure by a quantum channel $\mathcal C_\mu$,
it is $\mathcal C_\mu(\rho) = \Pi_+ (\rho \otimes \mu) \Pi_+ + K\Pi_- (\rho \otimes \mu) \Pi_- K^\dagger$,
where $\Pi_\pm$ denotes the measurement combined with a control-Pauli 
on the magic state and a target data qubit.
Let $\ket{\mu_0}$ be the ideal magic state, and $\ket{\mu_0^\perp}$ 
be the orthogonal state.
Then, it is straightforward to calculate that 
$\mathcal C_{\ket{\mu_0}\bra{\mu_0^\perp} + \ket{\mu_0^\perp}\bra{\mu_0}}(\rho) = 0$.

This implies that for \emph{any} initial approximate magic state $\mu$,
the result of the injection is the same as if 
$\mu$ had been through a twirling channel $\mathcal E$
that dephases the magic state in the basis $\{ \ket{\mu_0}, \ket{\mu_0^\perp} \}$:
\begin{align}
 \mu = 
 \begin{pmatrix} 1-\epsilon & * \\ * & \epsilon \end{pmatrix} 
 \xrightarrow{~~~~\mathcal E~~~~}
 \begin{pmatrix} 1-\epsilon & 0 \\ 0 & \epsilon \end{pmatrix}.
\end{align}
The twirled state is $\epsilon$ away from the ideal state in the trace distance,%
\footnote{The trace distance is defined as $T(\rho,\sigma) = \frac 1 2 \| \rho - \sigma \|_1$.  }
resulting in error at most $\epsilon$ to the quantum circuit's outcome.
The error $\epsilon$ can be expressed by the squared fidelity%
\footnote{The fidelity is defined as $F(\rho,\sigma) = \| \sqrt \rho \sqrt \sigma \|_1$,
which is equal to $|\braket{\rho|\sigma}|$ for pure $\rho$ and $\sigma$.}
as
\begin{align}
 1-\epsilon = F^2(\mu_0, \mu) = \bra{\mu_0} \mu \ket{\mu_0} = \bra{\mu_0} \mathcal E (\mu) \ket{\mu_0}.
\end{align}
 This formula is convenient in that it yields the same answer regardless of whether or not twirling is applied to $\mu$ (this is the last equality in the above formula).
When a state $\mu_n$ that approximates $\mu_0^{\otimes n}$ is injected,
the error from this multi-qubit magic state is given by $1 - F^2(\mu_0^{\otimes n}, \mu_n)$.
Note that $F^2(\mu_0,\mu)$ is linear in $\mu$.
Below, we use $1-F^2$ as the probability of error to report our simulation results.

\subsection{Error Models}
The typical model to analyze distillation protocols is the stochastic error model.
In typical distillation protocols, 
one has only a single output magic state, 
and so one is interested in the probability that 
the output magic state has an error as a function of the input, 
conditioned on no error being detected by the code; 
the error probability is a ratio of polynomials in $\epsilon$, 
with the leading term being of order $\epsilon^d$ for some $d$, with an integer coefficient.

For our purposes, since the codes used are fairly large, 
enumeration of all possible error patterns becomes difficult, 
especially if one wishes to go beyond leading order in $\epsilon$.  
For this reason, we use numerical simulation.  
One could simulate a mixed state, using a quantum channel to describe an approximate $T$-gate; however, this is numerically prohibitive and so we prefer to use an approach that involves only pure states.
One could numerically simulate pure states using the stochastic error model 
by choosing errors to occur with probability $p$, 
and sampling the output error probability.
However, this simulation also becomes difficult, 
precisely because the codes lead to a high suppression in the error.
For example, if the target error probability is $10^{-10}$, 
one would require $\sim 10^{10}$ samples, 
with a fairly large number of qubits needed to be simulated in each run, 
to determine the output error probability accurately.

While there may be ways to overcome this sampling issue using importance sampling, 
we use another method.
Instead of rotating by either $\pi/4$ or by $5\pi/4$ as in the stochastic error model, 
each $T$ gate rotates by an angle chosen uniformly in the interval $[\pi/4-\theta,\pi/4+\theta]$,
for some angle $\theta>0$.
Then, conditioned on the code not detecting an error, 
we determine the error in the output state.

In fact, the model 
with input angles $[\pi/4-\theta,\pi/4+\theta]$ 
and the stochastic error model describe the same average input state, 
assuming an appropriate choice of $\epsilon$ and $\theta$.
\begin{align}
\mu = 
 \frac{1}{2\theta}\int_{[-\theta,+\theta]} {\rm d}x
 \underbrace{
 \begin{pmatrix} 
 \cos^2 \frac{x}{2} & \sin \frac{x}{2} \cos \frac{x}{2} \\
 \sin \frac{x}{2} \cos \frac{x}{2} & \sin^2 \frac{x}{2}
 \end{pmatrix}
 }_{\rho_x}
&=
  (1-\epsilon) 
\underbrace{
 \begin{pmatrix}
1 & 0 \\ 0 & 0 
\end{pmatrix}
}_{\mu_{0}} 
+ \epsilon
\underbrace{
\begin{pmatrix}
0 & 0 \\ 0 & 1 
\end{pmatrix}
}_{\mu_1}
 \nonumber\\
\frac{1}{2}-\frac{\sin \theta}{2\theta} &= \epsilon \label{epsis}
\end{align}
Hence, one wants $\epsilon \approx \theta^2/12$.
(We emphasize that this is in a notation where $\theta$ is the rotation angle in the Bloch sphere;
the $T$-gate is a rotation by $\pi/4$, not by $\pi/8$.)
In the stochastic error model with small $\epsilon$, 
one must do roughly $1/\epsilon$ runs to obtain meaningful statistics, 
while here, one needs only a constant number of runs.
The reason is as follows.
Since $\theta$ is small, the simulated circuit can be approximated
by an analytic series in $\theta$, and the linear term amounts to a single error,
which is projected out by the post-selection on measurement outcomes
as our protocol always has $d \ge 2$.
Thus, in our post-selected simulation,
a circuit with $\rho_{x \in [-\theta,\theta]}$
is equivalent at leading order to a circuit with $\rho'_x$
where $\rho'_x$ is $\rho_x$ with the linear term in $x$ dropped.
Then, the distance from $n$-sample average of $\rho'_x$ to $\mu$ 
is $O(\epsilon/\sqrt{n})$,
whereas the distance from $n$-sample average of $\mu_i$ ($i=0,1$) to $\mu$
is $O(\sqrt{\epsilon/n})$.
The acceptance probability depends on the fidelity to the ideal magic state $\mu_0$,
which is $1-O(\epsilon)=1-O(\theta^2)$ in any case.

\subsection{Other Protocols}
\subsubsection{$[[16,2,4]]$  Inner Code}
\label{1624}
In subsubsection \ref{1664} we explained a protocol using a $[[16,6,4]]$ inner code.  
This required using a total of $17$ physical qubits, namely $16$ for the code and one ancilla.
We can also modify this inner code to a $[[16,2,4]]$ inner code, by turning some of the logical operators into checks. 
This inner code suffices to implement the $H$-measurements on pairs of states $(23),(45),(61)$
and so it can implement the checks of the outer code used in subsubsection \ref{1664}.
Using a $[[16,2,4]]$ inner code, if we want to have $\nout=6$, 
we need a total of $21$ physical qubits, 
since we need $16$ for the code, plus $4$ for the logical qubits not encoded in the code, plus one ancilla.
Thus, this requires additional physical qubits compared to the $[[16,6,4]]$ code.
The reason for considering the $[[16,2,4]]$ code in numerics is to see if 
it reduces the prefactor in the error, 
since the $[[16,2,4]]$ code has fewer logical operators than the $[[16,6,4]]$ code.
The number of $T$ gates in the protocol using $[[16,2,4]]$
is the same as that using $[[16,6,4]]$.
We pipeline the protocol with the $[[16,2,4]]$ inner code 
in the same way as we did with the $[[16,6,4]]$ inner code.
See Fig.~\ref{1624fig} in Appendix \ref{app:circuits}.

\subsubsection{$[[21,3,5]]$ Inner Code}
\label{2135}
Another inner code that we used is a $[[21,3,5]]$ inner code, described in Appendix~\ref{app:codes}. 
This allows us to obtain fifth order reduction in error.  
We used $\nout=4$ with the outer code having check matrix
\begin{align}
M=\begin{pmatrix}
1 & 1 & 1 & 0 \\
1 & 1 & 0 & 1 \\
1 & 0 & 1 & 1 \\
0 & 1 & 1 & 1
\end{pmatrix}.
\end{align}
It uses $4$ checks for $4$ qubits.  
This matrix is not $(4,2)$-sensitive,
but is sufficient to achieve fifth order reduction in error
since $2|M v| + |v| \ge 5$ for every nonzero $v$.
This protocol consumes $4 + 4 \times (21 \times 2) = 172$
$T$'s to produce $4$ outputs $T$ states.
See Fig.~\ref{2135fig} in Appendix \ref{app:circuits}.

A simple pipelining can reduce the noisy $T$ gate count compared to this protocol.
Distill three independent magic states using $[[7,1,3]]$ inner code.
(The outer code is trivial in this case.)
The three distilled magic states is then pipelined into the $[[21,3,5]]$ inner code.
This produces 3 magic states with error $O(\epsilon^5)$,
consuming, per output, $28$ $T$ gates and one $T$ state with error $\epsilon$.

Without the pipelining in the preceding paragraph,
but using weight $3$ checks from $[[21,3,5]]$,
we can find an outer code that is $(4,2)$ sensitive using $(2/3) \nout$ checks.
This produces $\nout$ magic states with fifth order error suppression,
consuming $42\cdot 2 / 3 = 28$ $T$ gates and one $T$ state, per output.
The smallest such outer code is explained in the appendix, where it is called the Petersen graph code.
For this outer code, $\nout = 15$,
so the overall protocol consumes $435$ $T$'s 
to produce $15$ outputs.
It is possible to further reduce the number of input $T$'s, 
by replacing some of the early checks
with a smaller code e.g. $[[15,7,3]]$ code,
which can be obtained by puncturing $[[16,6,4]]$ code.

\subsubsection{$[[23,1,7]]$ Inner Code}
\label{2317}
In the appendix, we give a $[[23,1,7]]$ inner code.  Pipelining this code with a $[[7,1,3]]$ and a $[[17,1,5]]$ inner code gives us $\nout=1$ with error $O(\epsilon^{7})$.
This protocol consumes $1 + 7 \times 2 + 17 \times 2 + 23 \times 2 = 95$ $T$'s, to produce 1 output magic state.
We could also apply this code to each of the output bits of any of the other fifth order protocols of section \ref{2135} to obtain error $O(\epsilon^7)$; we do not show results for this here.

\begin{figure}[t]
\centering
%\parbox{0.75\textwidth}{
\includegraphics[width=0.7\textwidth]{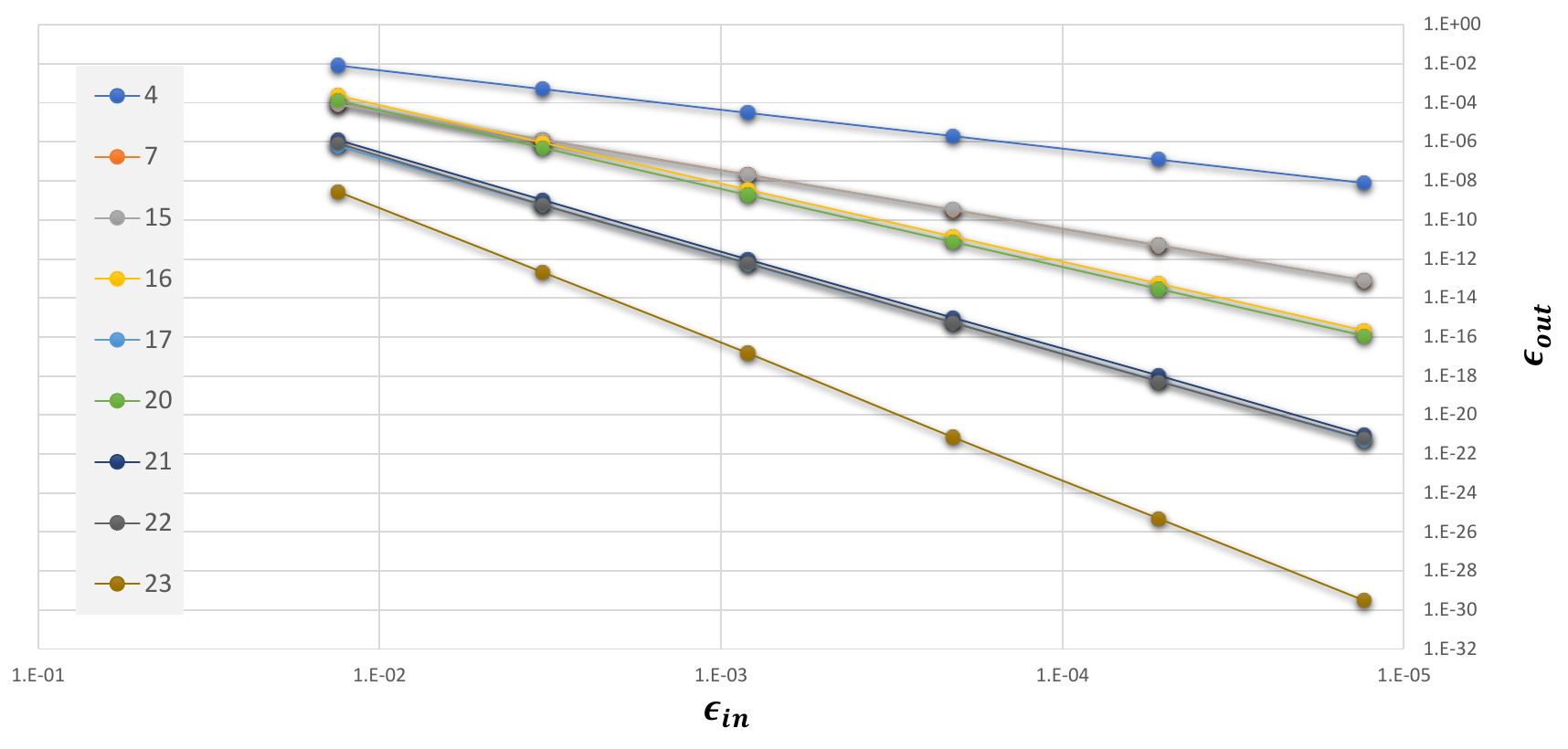}
%}
%\begin{minipage}{0.2\textwidth}{\footnotesize
\begin{tabular}{c|c|c|c|c}
\hline
Protocol & ~$d$~ & $\nT/\nout$ & $\nout$ & Prefactor \\
\hline
``4'' & 2 & 9 & 2 & $45 \times 10^0$ \\
``7'' & 3 & 15 & 1 & $35 \times 10^0$ \\
``15''& 3 & 15 & 1 & $35 \times 10^0$ \\
``16''& 4 & 41 & 6 & $7.3 \times 10^3$\\
``17''& 5 & 49 & 1 & $1.4 \times 10^3$\\
``20''& 4 & 41 & 6 & $3.9 \times 10^3$\\
``21''& 5 & 29 & 3 & $2.5 \times 10^3$\\
``22''& 5 & 43 & 4 & $1.4 \times 10^3$\\
``23''& 7 & 95 & 1 & $49 \times 10^3$\\
\hline
\end{tabular}
%}
%\end{minipage}
\caption{
Results of numerical simulations.
$\epsilon_{in}$ represents input error; 
this is the error $\epsilon$ of Eq.~\eqref{epsis} for the given $\theta$.  
$\epsilon_{out}$ is the average of $1-\bra{\mu_0} \mu \ket{\mu_0}$ over runs.
The fitting function is $\epsilon_{out} = C \epsilon_{in}^d$.
The graph confirms that the leading term dominates
the output error probability in the simulated regime.
The numbers labelling curves indicate the number of physical qubits, not including the ancilla qubit.
Specifically,
``$4$'' indicates protocol using $[[4,2,2]]$ inner code.
(We did not compress the circuit as Ref.~\cite{MeierEastinKnill2012Magic-state} did.)
``$7$'' indicates protocol using $[[7,1,3]]$ inner code.
``$15$'' is the Bravyi-Kitaev 15-to-1 protocol included for comparison purposes; 
``$7$'' and ``$15$'' have almost identical performance. 
``$16$'' is pipelined protocol using $[[16,6,4]]$ inner code.
``$20$'' is pipelined protocol using $[[16,2,4]]$ inner code.
``$17$'' is pipelined protocol using $[[17,1,5]]$ inner code.
``$22$'' is protocol using $[[21,3,5]]$ inner code with $\nout=4$.
``$21$'' is protocol using $[[21,3,5]]$ inner code pipelined with $[[7,1,3]]$ inner code.
``$23$'' is protocol using $[[23,1,7]]$ inner code pipelined
with ``$21$.''
}
\label{runs}
\end{figure}

\subsection{Results}

The results of the simulations are shown in Fig.~\ref{runs}.  Note that the plots are close to linear on a log-log plot, with only small deviations at high error rate.
Each data point represents the average of at least $10^{4}$ runs, with statistical flucutations negligible on the scale of the plot.  The asymptotic behavior is within statistical error of that given by an enumeration of minimum weight error patterns.

The protocol using $[[16,2,4]]$ inner code has a slightly reduced output error, compared to the protocol using the $[[16,6,4]]$ inner code.

We emphasize that $\epsilon_{out}$ indicates the probability that there is any error in the output state which is a multi-qubit state.  Suppose that two protocols give the same value of $\epsilon_{out}$ for a given $\epsilon_{in}$, but one protocol has a large $\nout$.  If the total number of magic states needed in our computation is large compared to $\nout$, the number of times we need to call the protocol is inversely proportional to $\nout$, and so the protocol with the larger $\nout$ for the given $\epsilon_{out}$ is less likely to produce an error.

The probability that no error is detected by the protocol is roughly $(1-\epsilon_{in})^{\nT}$.  This result would be exact if any error in an input $T$ gate led to the protocol detecting an error.  Instead, some high weight error patterns do not lead to any error detected by the code, leading to slight corrections to this formula.

\section{Discussion}
\label{sec:discussion}

We have given a general scheme to construct distillation protocols using inner and outer codes.
If desired, our protocols can be concatenated with other protocols.
However, on their own, 
they achieve asymptotic behavior which is conjectured to be optimal,
as well as having small size examples which perform well.
The concatenation can be useful
in an architecture where Clifford gates are performed at a logical level
so that one can tune the fidelity of Clifford gates
to match the relatively low fidelity of output magic states from 
early stages of concatenation~\cite{RaussendorfHarringtonGoyal2007,Fowler2012,OGormanCampbell2017}.
In contrast, our schemes without concatenation can be useful
if Clifford gates are already of high fidelity.

One of the major advantages of our protocols
is the small number of qubits that they use,
as they maintain a constant ratio of physical to logical qubits in the asymptotic limit.
It is interesting to consider the asymptotics of 
this overhead between physical and logical qubits.
Note that given any distillation protocol, 
there is a trivial way to define a new protocol with a fixed ratio of physical to logical qubits.  
Suppose, for example, that some protocol uses $\nphys$ qubits to produce $1$ output magic state.
Here $\nphys$ includes all physical qubits used in the protocol,
not only initial noisy $T$ states.
Call this protocol $P$.
One can define a new protocol $P'$ that works on $2\nphys$ qubits to produce $\nphys$ output magic states:
First, by working on $\nphys$ qubits,
leaving other $\nphys$ inactive,
we obtain $1$ output, which is put away.
Next, we use $\nphys$ qubits out of $2\nphys -1$ qubits
to obtain the second output, which is put away.
Next, we use $\nphys$ qubits out of $2\nphys -2$ qubits
to obtain the third output, which is put away.
Continuing, we apply $P$ a total of $\nphys$ times sequentially.
However, the circuit depth of $P'$ now 
is proportional to $\nphys$ times the depth of $P$.

In contrast, in Theorem~\ref{thm:dasymptotic},
we obtain $d$-th order error reduction at fixed ratio of physical to logical 
with a $T$-gate depth proportional to $d$.
That is, the protocols we have constructed are space and time efficient
in terms of $T$-gates.
It should be noted that we have ignored all Clifford gates,
and thus the $T$-efficiency claim is meaningful 
when $T$-gates are much more costly than Clifford gates.

Including Clifford gates, the entire circuit of Theorem~\ref{thm:dasymptotic}
has depth $O(d^2 \log d)$ if geometric locality of the gate is ignored.
Indeed, an $n$-qubit stabilizer code's encoding circuit can be constructed
in depth $O(n \log n)$.
It is essentially Gauss elimination of a binary symplectic matrix.
The Gauss elimination for a single column
can be done by a circuit of depth $O(\log n)$,
and hence the entire Gauss elimination can be done 
by a circuit of depth $O(n \log n)$.
Since we are using a good family of codes, 
the circuit depth of an encoding circuit is $O(d \log d)$,
and there are $O(d)$ encoding/decoding steps.
The total gate count (spacetime) per output is $O(d^2 \log d)$.

Jones~\cite{Jones2012} 
constructed a family of protocols giving $\gamma \rightarrow 1$.
(For the definition of $\gamma$, see the discussion below Theorem~\ref{thm:dasymptotic}.)
This protocol builds upon Knill's~\cite{Knill2004a},
and is in fact a subclass of ours.
Implicitly in Ref.~\cite{Jones2012},
the inner code is obtained by concatenating a $[[k+4,k,2]]$ code $\nu$ times
where $k$ is even, 
yielding $[[k^{\nu}+4(k^\nu-1)/(k-1),k^\nu,2^\nu]]$ normal weakly self-dual codes.
(For the definition of normal and hyperbolic codes, see~\ref{sec:innercodes}.)
The outer code is a hypercubic grid
with checks along coordinate axes.
The dimension of the grid is proportional to $\nu \cdot 2^\nu$,
since a check using an inner code of distance $2^\nu$
has to be implemented $2^{\nu-1}$ times
along $2^{\nu -1}$ independent axes,
and the concatenation for the inner codes is also performed on a grid.
For a given $\nu$, in the large $k$ limit,
one obtains a protocol
that consumes $2^\nu + 1 = d + 1$ noisy $T$ gates per output
at $d$-th order of reduction in error.
The asymptotic performance of this family 
is similar to our Theorem~\ref{thm:noutasymptotic}.
If we worked with normal codes of even distance in Theorem~\ref{thm:noutasymptotic},
then we would have concluded that the input $T$ count per output is $d+1$.
Note that the space requirement
of Jones' scheme is much larger than that of Theorem~\ref{thm:dasymptotic},
as the grid outer code used by Jones would require roughly
$k^{\nu 2^\nu} = k^{O(d)}$ qubits.
This exponential dependence on $d$ also holds for the protocols
in the proof of Theorem~\ref{thm:noutasymptotic}.
In contrast, Theorem~\ref{thm:dasymptotic} requires only $O(d)$ qubits.

In comparison with Jones' scheme,
our main technical contribution is
to explicitly separate inner and outer codes
with general criteria for them to be useful in distillation.
The criteria are that inner codes have to be weakly self-dual,
and outer codes have to be sensitive 
(i.e., the parity check matrix $M$ should satisfy $2|Mv|+|v| \ge d$
for any nonzero error vector $v$).
These requirements are rather simple,
so we were able to consider random constructions for them.
In fact, a large pool of existing codes can be incorporated.
In particular, we note that there are quantum BCH codes~\cite{qbch},
some of which have encoding rate greater than $1/2$ 
with code distances 5 and 7 at modest code lengths.

For Theorem~\ref{thm:noutasymptotic} 
we have resorted to a graph theoretic construction of outer codes from 
Ref.~\cite{furedi1995graphs}. 
This is sufficient for the proof, 
but one may wish to have more concrete examples.
In fact, a hypercubic grid of dimension $D \ge 3$ 
yields an outer code of desired sensitivity for $d=2D+1$,
which will be analyzed in detail elsewhere~\cite{enum}.

\bibliographystyle{apsrev4-1}
\nocite{apsrev41Control}
\bibliography{magic-ref}

\begin{thebibliography}{53}%
\makeatletter
\providecommand \@ifxundefined [1]{%
 \@ifx{#1\undefined}
}%
\providecommand \@ifnum [1]{%
 \ifnum #1\expandafter \@firstoftwo
 \else \expandafter \@secondoftwo
 \fi
}%
\providecommand \@ifx [1]{%
 \ifx #1\expandafter \@firstoftwo
 \else \expandafter \@secondoftwo
 \fi
}%
\providecommand \natexlab [1]{#1}%
\providecommand \enquote  [1]{``#1''}%
\providecommand \bibnamefont  [1]{#1}%
\providecommand \bibfnamefont [1]{#1}%
\providecommand \citenamefont [1]{#1}%
\providecommand \href@noop [0]{\@secondoftwo}%
\providecommand \href [0]{\begingroup \@sanitize@url \@href}%
\providecommand \@href[1]{\@@startlink{#1}\@@href}%
\providecommand \@@href[1]{\endgroup#1\@@endlink}%
\providecommand \@sanitize@url [0]{\catcode `\\12\catcode `\$12\catcode
  `\&12\catcode `\#12\catcode `\^12\catcode `\_12\catcode `\%12\relax}%
\providecommand \@@startlink[1]{}%
\providecommand \@@endlink[0]{}%
\providecommand \url  [0]{\begingroup\@sanitize@url \@url }%
\providecommand \@url [1]{\endgroup\@href {#1}{\urlprefix }}%
\providecommand \urlprefix  [0]{URL }%
\providecommand \Eprint [0]{\href }%
\providecommand \doibase [0]{http://dx.doi.org/}%
\providecommand \selectlanguage [0]{\@gobble}%
\providecommand \bibinfo  [0]{\@secondoftwo}%
\providecommand \bibfield  [0]{\@secondoftwo}%
\providecommand \translation [1]{[#1]}%
\providecommand \BibitemOpen [0]{}%
\providecommand \bibitemStop [0]{}%
\providecommand \bibitemNoStop [0]{.\EOS\space}%
\providecommand \EOS [0]{\spacefactor3000\relax}%
\providecommand \BibitemShut  [1]{\csname bibitem#1\endcsname}%
\let\auto@bib@innerbib\@empty
%</preamble>
\bibitem [{\citenamefont {Shor}(1996)}]{Shor1996Fault-tolerant}%
  \BibitemOpen
  \bibfield  {author} {\bibinfo {author} {\bibfnamefont {P.~W.}\ \bibnamefont
  {Shor}},\ }\bibfield  {title} {\enquote {\bibinfo {title} {Fault-tolerant
  quantum computation},}\ }in\ \href {\doibase 10.1109/SFCS.1996.548464} {\emph
  {\bibinfo {booktitle} {Foundations of Computer Science, 1996. Proceedings.,
  37th Annual Symposium on}}}\ (\bibinfo {year} {1996})\ pp.\ \bibinfo {pages}
  {56--65},\ \Eprint {http://arxiv.org/abs/quant-ph/9605011} {quant-ph/9605011}
  \BibitemShut {NoStop}%
\bibitem [{\citenamefont {Aharonov}\ and\ \citenamefont
  {Ben-Or}(2008)}]{AharonovBenOr1996}%
  \BibitemOpen
  \bibfield  {author} {\bibinfo {author} {\bibfnamefont {D.}~\bibnamefont
  {Aharonov}}\ and\ \bibinfo {author} {\bibfnamefont {M.}~\bibnamefont
  {Ben-Or}},\ }\bibfield  {title} {\enquote {\bibinfo {title} {Fault tolerant
  quantum computation with constant error},}\ }\href {\doibase
  10.1137/S0097539799359385} {\bibfield  {journal} {\bibinfo  {journal} {SIAM
  J. Comput.}\ }\textbf {\bibinfo {volume} {38}},\ \bibinfo {pages}
  {1207--1282} (\bibinfo {year} {2008})},\ \Eprint
  {http://arxiv.org/abs/quant-ph/9611025v2} {quant-ph/9611025v2} \BibitemShut
  {NoStop}%
\bibitem [{\citenamefont {Knill}\ \emph {et~al.}(1996)\citenamefont {Knill},
  \citenamefont {Laflamme},\ and\ \citenamefont
  {Zurek}}]{KnillLaflammeZurek1996Threshold}%
  \BibitemOpen
  \bibfield  {author} {\bibinfo {author} {\bibfnamefont {E.}~\bibnamefont
  {Knill}}, \bibinfo {author} {\bibfnamefont {R.}~\bibnamefont {Laflamme}}, \
  and\ \bibinfo {author} {\bibfnamefont {W.}~\bibnamefont {Zurek}},\ }\bibfield
   {title} {\enquote {\bibinfo {title} {Threshold accuracy for quantum
  computation},}\ }\href@noop {} {\  (\bibinfo {year} {1996})},\ \Eprint
  {http://arxiv.org/abs/quant-ph/9610011} {quant-ph/9610011} \BibitemShut
  {NoStop}%
\bibitem [{\citenamefont {Gottesman}(1996)}]{Gottesman1996Saturating}%
  \BibitemOpen
  \bibfield  {author} {\bibinfo {author} {\bibfnamefont {D.}~\bibnamefont
  {Gottesman}},\ }\bibfield  {title} {\enquote {\bibinfo {title} {A class of
  quantum error-correcting codes saturating the quantum hamming bound},}\
  }\href {\doibase 10.1103/PhysRevA.54.1862} {\bibfield  {journal} {\bibinfo
  {journal} {Phys. Rev. A}\ }\textbf {\bibinfo {volume} {54}},\ \bibinfo
  {pages} {1862} (\bibinfo {year} {1996})},\ \Eprint
  {http://arxiv.org/abs/quant-ph/9604038} {quant-ph/9604038} \BibitemShut
  {NoStop}%
\bibitem [{\citenamefont {Calderbank}\ \emph {et~al.}(1997)\citenamefont
  {Calderbank}, \citenamefont {Rains}, \citenamefont {Shor},\ and\
  \citenamefont {Sloane}}]{CalderbankRainsShorEtAl1997Quantum}%
  \BibitemOpen
  \bibfield  {author} {\bibinfo {author} {\bibfnamefont {A.~R.}\ \bibnamefont
  {Calderbank}}, \bibinfo {author} {\bibfnamefont {E.~M.}\ \bibnamefont
  {Rains}}, \bibinfo {author} {\bibfnamefont {P.~W.}\ \bibnamefont {Shor}}, \
  and\ \bibinfo {author} {\bibfnamefont {N.~J.~A.}\ \bibnamefont {Sloane}},\
  }\bibfield  {title} {\enquote {\bibinfo {title} {Quantum error correction and
  orthogonal geometry},}\ }\href {\doibase 10.1103/PhysRevLett.78.405}
  {\bibfield  {journal} {\bibinfo  {journal} {Phys. Rev. Lett.}\ }\textbf
  {\bibinfo {volume} {78}},\ \bibinfo {pages} {405--408} (\bibinfo {year}
  {1997})},\ \Eprint {http://arxiv.org/abs/quant-ph/9605005} {quant-ph/9605005}
  \BibitemShut {NoStop}%
\bibitem [{\citenamefont {Jones}\ \emph {et~al.}(2012)\citenamefont {Jones},
  \citenamefont {Meter}, \citenamefont {Fowler}, \citenamefont {McMahon},
  \citenamefont {Kim}, \citenamefont {Ladd},\ and\ \citenamefont
  {Yamamoto}}]{JonesEtAl2012}%
  \BibitemOpen
  \bibfield  {author} {\bibinfo {author} {\bibfnamefont {N.~C.}\ \bibnamefont
  {Jones}}, \bibinfo {author} {\bibfnamefont {R.~V.}\ \bibnamefont {Meter}},
  \bibinfo {author} {\bibfnamefont {A.~G.}\ \bibnamefont {Fowler}}, \bibinfo
  {author} {\bibfnamefont {P.~L.}\ \bibnamefont {McMahon}}, \bibinfo {author}
  {\bibfnamefont {J.}~\bibnamefont {Kim}}, \bibinfo {author} {\bibfnamefont
  {T.~D.}\ \bibnamefont {Ladd}}, \ and\ \bibinfo {author} {\bibfnamefont
  {Y.}~\bibnamefont {Yamamoto}},\ }\bibfield  {title} {\enquote {\bibinfo
  {title} {Layered architecture for quantum computing},}\ }\href {\doibase
  10.1103/physrevx.2.031007} {\bibfield  {journal} {\bibinfo  {journal}
  {Physical Review X}\ }\textbf {\bibinfo {volume} {2}},\ \bibinfo {pages}
  {031007} (\bibinfo {year} {2012})},\ \Eprint {http://arxiv.org/abs/1010.5022}
  {1010.5022} \BibitemShut {NoStop}%
\bibitem [{\citenamefont {O'Gorman}\ and\ \citenamefont
  {Campbell}()}]{OGormanCampbell2016}%
  \BibitemOpen
  \bibfield  {author} {\bibinfo {author} {\bibfnamefont {J.}~\bibnamefont
  {O'Gorman}}\ and\ \bibinfo {author} {\bibfnamefont {E.~T.}\ \bibnamefont
  {Campbell}},\ }\bibfield  {title} {\enquote {\bibinfo {title} {Quantum
  computation with realistic magic state factories},}\ }\href@noop {} {\
  }\Eprint {http://arxiv.org/abs/1605.07197v2} {1605.07197v2} \BibitemShut
  {NoStop}%
\bibitem [{\citenamefont {Knill}(2004{\natexlab{a}})}]{Knill2004a}%
  \BibitemOpen
  \bibfield  {author} {\bibinfo {author} {\bibfnamefont {E.}~\bibnamefont
  {Knill}},\ }\bibfield  {title} {\enquote {\bibinfo {title} {Fault-tolerant
  postselected quantum computation: Schemes},}\ }\href@noop {} {\  (\bibinfo
  {year} {2004}{\natexlab{a}})},\ \Eprint
  {http://arxiv.org/abs/quant-ph/0402171v1} {quant-ph/0402171v1} \BibitemShut
  {NoStop}%
\bibitem [{\citenamefont {Knill}(2004{\natexlab{b}})}]{Knill2004b}%
  \BibitemOpen
  \bibfield  {author} {\bibinfo {author} {\bibfnamefont {E.}~\bibnamefont
  {Knill}},\ }\bibfield  {title} {\enquote {\bibinfo {title} {Fault-tolerant
  postselected quantum computation: Threshold analysis},}\ }\href@noop {} {\
  (\bibinfo {year} {2004}{\natexlab{b}})},\ \Eprint
  {http://arxiv.org/abs/quant-ph/0404104v1} {quant-ph/0404104v1} \BibitemShut
  {NoStop}%
\bibitem [{\citenamefont {Bravyi}\ and\ \citenamefont
  {Kitaev}(2005)}]{BravyiKitaev2005Magic}%
  \BibitemOpen
  \bibfield  {author} {\bibinfo {author} {\bibfnamefont {S.}~\bibnamefont
  {Bravyi}}\ and\ \bibinfo {author} {\bibfnamefont {A.}~\bibnamefont
  {Kitaev}},\ }\bibfield  {title} {\enquote {\bibinfo {title} {Universal
  quantum computation with ideal {C}lifford gates and noisy ancillas},}\ }\href
  {\doibase 10.1103/PhysRevA.71.022316} {\bibfield  {journal} {\bibinfo
  {journal} {Phys. Rev. A}\ }\textbf {\bibinfo {volume} {71}},\ \bibinfo
  {pages} {022316} (\bibinfo {year} {2005})},\ \Eprint
  {http://arxiv.org/abs/quant-ph/0403025} {quant-ph/0403025} \BibitemShut
  {NoStop}%
\bibitem [{\citenamefont {Meier}\ \emph {et~al.}(2013)\citenamefont {Meier},
  \citenamefont {Eastin},\ and\ \citenamefont
  {Knill}}]{MeierEastinKnill2012Magic-state}%
  \BibitemOpen
  \bibfield  {author} {\bibinfo {author} {\bibfnamefont {A.~M.}\ \bibnamefont
  {Meier}}, \bibinfo {author} {\bibfnamefont {B.}~\bibnamefont {Eastin}}, \
  and\ \bibinfo {author} {\bibfnamefont {E.}~\bibnamefont {Knill}},\ }\bibfield
   {title} {\enquote {\bibinfo {title} {Magic-state distillation with the
  four-qubit code},}\ }\href
  {http://www.rintonpress.com/xxqic13/qic-13-34/0195-0209.pdf} {\bibfield
  {journal} {\bibinfo  {journal} {Quant. Inf. Comp.}\ }\textbf {\bibinfo
  {volume} {13}},\ \bibinfo {pages} {195} (\bibinfo {year} {2013})},\ \Eprint
  {http://arxiv.org/abs/1204.4221} {1204.4221} \BibitemShut {NoStop}%
\bibitem [{\citenamefont {Bravyi}\ and\ \citenamefont
  {Haah}(2012)}]{BravyiHaah2012Magic}%
  \BibitemOpen
  \bibfield  {author} {\bibinfo {author} {\bibfnamefont {S.}~\bibnamefont
  {Bravyi}}\ and\ \bibinfo {author} {\bibfnamefont {J.}~\bibnamefont {Haah}},\
  }\bibfield  {title} {\enquote {\bibinfo {title} {Magic state distillation
  with low overhead},}\ }\href {\doibase 10.1103/PhysRevA.86.052329} {\bibfield
   {journal} {\bibinfo  {journal} {Phys. Rev. A}\ }\textbf {\bibinfo {volume}
  {86}},\ \bibinfo {pages} {052329} (\bibinfo {year} {2012})},\ \Eprint
  {http://arxiv.org/abs/1209.2426} {1209.2426} \BibitemShut {NoStop}%
\bibitem [{\citenamefont {Jones}(2013{\natexlab{a}})}]{Jones2012}%
  \BibitemOpen
  \bibfield  {author} {\bibinfo {author} {\bibfnamefont {C.}~\bibnamefont
  {Jones}},\ }\bibfield  {title} {\enquote {\bibinfo {title} {Multilevel
  distillation of magic states for quantum computing},}\ }\href {\doibase
  10.1103/PhysRevA.87.042305} {\bibfield  {journal} {\bibinfo  {journal} {Phys.
  Rev. A}\ }\textbf {\bibinfo {volume} {87}},\ \bibinfo {pages} {042305}
  (\bibinfo {year} {2013}{\natexlab{a}})},\ \Eprint
  {http://arxiv.org/abs/1210.3388v2} {1210.3388v2} \BibitemShut {NoStop}%
\bibitem [{\citenamefont {Duclos-Cianci}\ and\ \citenamefont
  {Poulin}(2015)}]{DP2014}%
  \BibitemOpen
  \bibfield  {author} {\bibinfo {author} {\bibfnamefont {G.}~\bibnamefont
  {Duclos-Cianci}}\ and\ \bibinfo {author} {\bibfnamefont {D.}~\bibnamefont
  {Poulin}},\ }\bibfield  {title} {\enquote {\bibinfo {title} {Reducing the
  quantum computing overhead with complex gate distillation},}\ }\href
  {\doibase 10.1103/PhysRevA.91.042315} {\bibfield  {journal} {\bibinfo
  {journal} {Phys. Rev. A}\ }\textbf {\bibinfo {volume} {91}},\ \bibinfo
  {pages} {042315} (\bibinfo {year} {2015})},\ \Eprint
  {http://arxiv.org/abs/1403.5280v1} {1403.5280v1} \BibitemShut {NoStop}%
\bibitem [{\citenamefont {Chuang}\ and\ \citenamefont
  {Gottesman}(1999)}]{ChuangGottesman1999}%
  \BibitemOpen
  \bibfield  {author} {\bibinfo {author} {\bibfnamefont {I.~L.}\ \bibnamefont
  {Chuang}}\ and\ \bibinfo {author} {\bibfnamefont {D.}~\bibnamefont
  {Gottesman}},\ }\bibfield  {title} {\enquote {\bibinfo {title} {Demonstrating
  the viability of universal quantum computation using teleportation and
  single-qubit operations},}\ }\href {\doibase 10.1038/46503} {\bibfield
  {journal} {\bibinfo  {journal} {Nature}\ }\textbf {\bibinfo {volume} {402}},\
  \bibinfo {pages} {390--393} (\bibinfo {year} {1999})}\BibitemShut {NoStop}%
\bibitem [{\citenamefont {Landahl}\ and\ \citenamefont
  {Cesare}(2013)}]{LandahlCesare2013Tgate}%
  \BibitemOpen
  \bibfield  {author} {\bibinfo {author} {\bibfnamefont {A.~J.}\ \bibnamefont
  {Landahl}}\ and\ \bibinfo {author} {\bibfnamefont {C.}~\bibnamefont
  {Cesare}},\ }\bibfield  {title} {\enquote {\bibinfo {title} {Complex
  instruction set computing architecture for performing accurate quantum {$Z$}
  rotations with less magic},}\ }\href@noop {} {\  (\bibinfo {year} {2013})},\
  \Eprint {http://arxiv.org/abs/1302.3240} {1302.3240} \BibitemShut {NoStop}%
\bibitem [{\citenamefont {Paetznick}\ and\ \citenamefont
  {Reichardt}(2013)}]{PaetznickReichardt2013Universal}%
  \BibitemOpen
  \bibfield  {author} {\bibinfo {author} {\bibfnamefont {A.}~\bibnamefont
  {Paetznick}}\ and\ \bibinfo {author} {\bibfnamefont {B.~W.}\ \bibnamefont
  {Reichardt}},\ }\bibfield  {title} {\enquote {\bibinfo {title} {Universal
  fault-tolerant quantum computation with only transversal gates and error
  correction},}\ }\href {\doibase 10.1103/PhysRevLett.111.090505} {\bibfield
  {journal} {\bibinfo  {journal} {Phys. Rev. Lett.}\ }\textbf {\bibinfo
  {volume} {111}},\ \bibinfo {pages} {090505} (\bibinfo {year} {2013})},\
  \Eprint {http://arxiv.org/abs/1304.3709} {1304.3709} \BibitemShut {NoStop}%
\bibitem [{\citenamefont {Jochym-O'Connor}\ and\ \citenamefont
  {Laflamme}(2014)}]{JochymOConnorLaflamme2014}%
  \BibitemOpen
  \bibfield  {author} {\bibinfo {author} {\bibfnamefont {T.}~\bibnamefont
  {Jochym-O'Connor}}\ and\ \bibinfo {author} {\bibfnamefont {R.}~\bibnamefont
  {Laflamme}},\ }\bibfield  {title} {\enquote {\bibinfo {title} {Using
  concatenated quantum codes for universal fault-tolerant quantum gates},}\
  }\href {\doibase 10.1103/PhysRevLett.112.010505} {\bibfield  {journal}
  {\bibinfo  {journal} {Phys. Rev. Lett.}\ }\textbf {\bibinfo {volume} {112}},\
  \bibinfo {pages} {010505} (\bibinfo {year} {2014})},\ \Eprint
  {http://arxiv.org/abs/1309.3310} {1309.3310} \BibitemShut {NoStop}%
\bibitem [{\citenamefont {Anderson}\ \emph {et~al.}(2014)\citenamefont
  {Anderson}, \citenamefont {Duclos-Cianci},\ and\ \citenamefont
  {Poulin}}]{ADP14a}%
  \BibitemOpen
  \bibfield  {author} {\bibinfo {author} {\bibfnamefont {J.~T.}\ \bibnamefont
  {Anderson}}, \bibinfo {author} {\bibfnamefont {G.}~\bibnamefont
  {Duclos-Cianci}}, \ and\ \bibinfo {author} {\bibfnamefont {D.}~\bibnamefont
  {Poulin}},\ }\bibfield  {title} {\enquote {\bibinfo {title} {Fault-tolerant
  conversion between the steane and reed-muller quantum codes},}\ }\href
  {\doibase 10.1103/PhysRevLett.113.080501} {\bibfield  {journal} {\bibinfo
  {journal} {Phys. Rev. Lett.}\ }\textbf {\bibinfo {volume} {113}},\ \bibinfo
  {pages} {080501} (\bibinfo {year} {2014})},\ \Eprint
  {http://arxiv.org/abs/1403.2734} {1403.2734} \BibitemShut {NoStop}%
\bibitem [{\citenamefont {Bombin}(2015)}]{Bombin2015}%
  \BibitemOpen
  \bibfield  {author} {\bibinfo {author} {\bibfnamefont {H.}~\bibnamefont
  {Bombin}},\ }\bibfield  {title} {\enquote {\bibinfo {title} {Gauge color
  codes: Optimal transversal gates and gauge fixing in topological stabilizer
  codes},}\ }\href {\doibase 10.1088/1367-2630/17/8/083002} {\bibfield
  {journal} {\bibinfo  {journal} {New J.Phys.}\ }\textbf {\bibinfo {volume}
  {17}},\ \bibinfo {pages} {083002} (\bibinfo {year} {2015})},\ \Eprint
  {http://arxiv.org/abs/1311.0879v6} {1311.0879v6} \BibitemShut {NoStop}%
\bibitem [{\citenamefont {Bravyi}\ and\ \citenamefont
  {Cross}()}]{BravyiCross2015}%
  \BibitemOpen
  \bibfield  {author} {\bibinfo {author} {\bibfnamefont {S.}~\bibnamefont
  {Bravyi}}\ and\ \bibinfo {author} {\bibfnamefont {A.}~\bibnamefont {Cross}},\
  }\bibfield  {title} {\enquote {\bibinfo {title} {Doubled color codes},}\
  }\href@noop {} {\ }\Eprint {http://arxiv.org/abs/1509.03239v1} {1509.03239v1}
  \BibitemShut {NoStop}%
\bibitem [{\citenamefont {Jochym-O'Connor}\ and\ \citenamefont
  {Bartlett}(2016)}]{JochymOConnorBartlett2016}%
  \BibitemOpen
  \bibfield  {author} {\bibinfo {author} {\bibfnamefont {T.}~\bibnamefont
  {Jochym-O'Connor}}\ and\ \bibinfo {author} {\bibfnamefont {S.~D.}\
  \bibnamefont {Bartlett}},\ }\bibfield  {title} {\enquote {\bibinfo {title}
  {Stacked codes: Universal fault-tolerant quantum computation in a
  two-dimensional layout},}\ }\href {\doibase 10.1103/PhysRevA.93.022323}
  {\bibfield  {journal} {\bibinfo  {journal} {Phys. Rev. A}\ }\textbf {\bibinfo
  {volume} {93}},\ \bibinfo {pages} {022323} (\bibinfo {year} {2016})},\
  \Eprint {http://arxiv.org/abs/1509.04255} {1509.04255} \BibitemShut {NoStop}%
\bibitem [{\citenamefont {Jones}\ \emph {et~al.}(2016)\citenamefont {Jones},
  \citenamefont {Brooks},\ and\ \citenamefont
  {Harrington}}]{JonesBrooksHarrington2016}%
  \BibitemOpen
  \bibfield  {author} {\bibinfo {author} {\bibfnamefont {C.}~\bibnamefont
  {Jones}}, \bibinfo {author} {\bibfnamefont {P.}~\bibnamefont {Brooks}}, \
  and\ \bibinfo {author} {\bibfnamefont {J.}~\bibnamefont {Harrington}},\
  }\bibfield  {title} {\enquote {\bibinfo {title} {Gauge color codes in two
  dimensions},}\ }\href {\doibase 10.1103/PhysRevA.93.052332} {\bibfield
  {journal} {\bibinfo  {journal} {Phys. Rev. A}\ }\textbf {\bibinfo {volume}
  {93}},\ \bibinfo {pages} {052332} (\bibinfo {year} {2016})},\ \Eprint
  {http://arxiv.org/abs/1512.04193} {1512.04193} \BibitemShut {NoStop}%
\bibitem [{\citenamefont {Bombin}(2016)}]{Bombin2016}%
  \BibitemOpen
  \bibfield  {author} {\bibinfo {author} {\bibfnamefont {H.}~\bibnamefont
  {Bombin}},\ }\bibfield  {title} {\enquote {\bibinfo {title} {Dimensional jump
  in quantum error correction},}\ }\href {\doibase
  10.1088/1367-2630/18/4/043038} {\bibfield  {journal} {\bibinfo  {journal}
  {New Journal of Physics}\ }\textbf {\bibinfo {volume} {18}},\ \bibinfo
  {pages} {043038} (\bibinfo {year} {2016})},\ \Eprint
  {http://arxiv.org/abs/1412.5079v3} {1412.5079v3} \BibitemShut {NoStop}%
\bibitem [{\citenamefont {Yoder}\ \emph {et~al.}(2016)\citenamefont {Yoder},
  \citenamefont {Takagi},\ and\ \citenamefont {Chuang}}]{Yoder2016}%
  \BibitemOpen
  \bibfield  {author} {\bibinfo {author} {\bibfnamefont {T.~J.}\ \bibnamefont
  {Yoder}}, \bibinfo {author} {\bibfnamefont {R.}~\bibnamefont {Takagi}}, \
  and\ \bibinfo {author} {\bibfnamefont {I.~L.}\ \bibnamefont {Chuang}},\
  }\bibfield  {title} {\enquote {\bibinfo {title} {Universal fault-tolerant
  gates on concatenated stabilizer codes},}\ }\href {\doibase
  10.1103/physrevx.6.031039} {\bibfield  {journal} {\bibinfo  {journal}
  {Physical Review X}\ }\textbf {\bibinfo {volume} {6}},\ \bibinfo {pages}
  {031039} (\bibinfo {year} {2016})},\ \Eprint
  {http://arxiv.org/abs/1603.03948} {1603.03948} \BibitemShut {NoStop}%
\bibitem [{\citenamefont {Eastin}\ and\ \citenamefont
  {Knill}(2009)}]{EastinKnill2009Restrictions}%
  \BibitemOpen
  \bibfield  {author} {\bibinfo {author} {\bibfnamefont {B.}~\bibnamefont
  {Eastin}}\ and\ \bibinfo {author} {\bibfnamefont {E.}~\bibnamefont {Knill}},\
  }\bibfield  {title} {\enquote {\bibinfo {title} {Restrictions on transversal
  encoded quantum gate sets},}\ }\href {\doibase
  10.1103/PhysRevLett.102.110502} {\bibfield  {journal} {\bibinfo  {journal}
  {Phys. Rev. Lett.}\ }\textbf {\bibinfo {volume} {102}},\ \bibinfo {pages}
  {110502} (\bibinfo {year} {2009})},\ \Eprint {http://arxiv.org/abs/0811.4262}
  {0811.4262} \BibitemShut {NoStop}%
\bibitem [{\citenamefont {Bravyi}\ and\ \citenamefont
  {K\"onig}(2013)}]{BravyiKoenig2012Classification}%
  \BibitemOpen
  \bibfield  {author} {\bibinfo {author} {\bibfnamefont {S.}~\bibnamefont
  {Bravyi}}\ and\ \bibinfo {author} {\bibfnamefont {R.}~\bibnamefont
  {K\"onig}},\ }\bibfield  {title} {\enquote {\bibinfo {title} {Classification
  of topologically protected gates for local stabilizer codes},}\ }\href
  {\doibase 10.1103/PhysRevLett.110.170503} {\bibfield  {journal} {\bibinfo
  {journal} {Phys. Rev. Lett}\ }\textbf {\bibinfo {volume} {110}},\ \bibinfo
  {pages} {170503} (\bibinfo {year} {2013})},\ \Eprint
  {http://arxiv.org/abs/1206.1609} {1206.1609} \BibitemShut {NoStop}%
\bibitem [{\citenamefont {Campbell}\ and\ \citenamefont
  {O'Gorman}(2016)}]{Campbell2016}%
  \BibitemOpen
  \bibfield  {author} {\bibinfo {author} {\bibfnamefont {E.~T.}\ \bibnamefont
  {Campbell}}\ and\ \bibinfo {author} {\bibfnamefont {J.}~\bibnamefont
  {O'Gorman}},\ }\bibfield  {title} {\enquote {\bibinfo {title} {An efficient
  magic state approach to small angle rotations},}\ }\href {\doibase
  10.1088/2058-9565/1/1/015007} {\bibfield  {journal} {\bibinfo  {journal}
  {Quantum Science and Technology}\ }\textbf {\bibinfo {volume} {1}},\ \bibinfo
  {pages} {015007} (\bibinfo {year} {2016})},\ \Eprint
  {http://arxiv.org/abs/1603.04230v2} {1603.04230v2} \BibitemShut {NoStop}%
\bibitem [{\citenamefont {Reichardt}(2005)}]{Reichardt2005}%
  \BibitemOpen
  \bibfield  {author} {\bibinfo {author} {\bibfnamefont {B.~W.}\ \bibnamefont
  {Reichardt}},\ }\bibfield  {title} {\enquote {\bibinfo {title} {Improved
  magic states distillation for quantum universality},}\ }\href {\doibase
  10.1007/s11128-005-7654-8} {\bibfield  {journal} {\bibinfo  {journal} {Quant.
  Inf. Proc.}\ }\textbf {\bibinfo {volume} {4}},\ \bibinfo {pages} {251--264}
  (\bibinfo {year} {2005})},\ \Eprint {http://arxiv.org/abs/quant-ph/0411036v1}
  {quant-ph/0411036v1} \BibitemShut {NoStop}%
\bibitem [{\citenamefont {Steane}(1996)}]{Steane_1996}%
  \BibitemOpen
  \bibfield  {author} {\bibinfo {author} {\bibfnamefont {A.~M.}\ \bibnamefont
  {Steane}},\ }\bibfield  {title} {\enquote {\bibinfo {title} {Error correcting
  codes in quantum theory},}\ }\href {\doibase 10.1103/physrevlett.77.793}
  {\bibfield  {journal} {\bibinfo  {journal} {Physical Review Letters}\
  }\textbf {\bibinfo {volume} {77}},\ \bibinfo {pages} {793--797} (\bibinfo
  {year} {1996})}\BibitemShut {NoStop}%
\bibitem [{\citenamefont {Bombin}\ and\ \citenamefont
  {Martin-Delgado}(2006)}]{Bombin_2006}%
  \BibitemOpen
  \bibfield  {author} {\bibinfo {author} {\bibfnamefont {H.}~\bibnamefont
  {Bombin}}\ and\ \bibinfo {author} {\bibfnamefont {M.~A.}\ \bibnamefont
  {Martin-Delgado}},\ }\bibfield  {title} {\enquote {\bibinfo {title}
  {Topological quantum distillation},}\ }\href {\doibase
  10.1103/physrevlett.97.180501} {\bibfield  {journal} {\bibinfo  {journal}
  {Physical Review Letters}\ }\textbf {\bibinfo {volume} {97}},\ \bibinfo
  {pages} {180501} (\bibinfo {year} {2006})},\ \Eprint
  {http://arxiv.org/abs/quant-ph/0605138} {quant-ph/0605138} \BibitemShut
  {NoStop}%
\bibitem [{\citenamefont {Eastin}(2013)}]{Eastin2013}%
  \BibitemOpen
  \bibfield  {author} {\bibinfo {author} {\bibfnamefont {B.}~\bibnamefont
  {Eastin}},\ }\bibfield  {title} {\enquote {\bibinfo {title} {Distilling
  one-qubit magic states into toffoli states},}\ }\href {\doibase
  10.1103/PhysRevA.87.032321} {\bibfield  {journal} {\bibinfo  {journal} {Phys.
  Rev. A}\ }\textbf {\bibinfo {volume} {87}},\ \bibinfo {pages} {032321}
  (\bibinfo {year} {2013})},\ \Eprint {http://arxiv.org/abs/1212.4872v2}
  {1212.4872v2} \BibitemShut {NoStop}%
\bibitem [{\citenamefont {Jones}(2013{\natexlab{b}})}]{Jones2013Tof2}%
  \BibitemOpen
  \bibfield  {author} {\bibinfo {author} {\bibfnamefont {C.}~\bibnamefont
  {Jones}},\ }\bibfield  {title} {\enquote {\bibinfo {title} {Low-overhead
  constructions for the fault-tolerant toffoli gate},}\ }\href {\doibase
  10.1103/physreva.87.022328} {\bibfield  {journal} {\bibinfo  {journal}
  {Physical Review A}\ }\textbf {\bibinfo {volume} {87}},\ \bibinfo {pages}
  {022328} (\bibinfo {year} {2013}{\natexlab{b}})},\ \Eprint
  {http://arxiv.org/abs/1212.5069} {1212.5069} \BibitemShut {NoStop}%
\bibitem [{\citenamefont {Campbell}\ and\ \citenamefont
  {Howard}(2017)}]{CampbellHoward2017}%
  \BibitemOpen
  \bibfield  {author} {\bibinfo {author} {\bibfnamefont {E.~T.}\ \bibnamefont
  {Campbell}}\ and\ \bibinfo {author} {\bibfnamefont {M.}~\bibnamefont
  {Howard}},\ }\bibfield  {title} {\enquote {\bibinfo {title} {Unified
  framework for magic state distillation and multiqubit gate synthesis with
  reduced resource cost},}\ }\href {\doibase 10.1103/PhysRevA.95.022316}
  {\bibfield  {journal} {\bibinfo  {journal} {Phys. Rev. A}\ }\textbf {\bibinfo
  {volume} {95}},\ \bibinfo {pages} {022316} (\bibinfo {year} {2017})},\
  \Eprint {http://arxiv.org/abs/1606.01904} {1606.01904} \BibitemShut {NoStop}%
\bibitem [{\citenamefont {Bravyi}\ \emph {et~al.}(2010)\citenamefont {Bravyi},
  \citenamefont {Leemhuis},\ and\ \citenamefont
  {Terhal}}]{BravyiLeemhuisTerhal2010Majorana}%
  \BibitemOpen
  \bibfield  {author} {\bibinfo {author} {\bibfnamefont {S.}~\bibnamefont
  {Bravyi}}, \bibinfo {author} {\bibfnamefont {B.}~\bibnamefont {Leemhuis}}, \
  and\ \bibinfo {author} {\bibfnamefont {B.~M.}\ \bibnamefont {Terhal}},\
  }\bibfield  {title} {\enquote {\bibinfo {title} {{M}ajorana fermion codes},}\
  }\href {\doibase 10.1088/1367-2630/12/8/083039} {\bibfield  {journal}
  {\bibinfo  {journal} {New J.Phys.}\ }\textbf {\bibinfo {volume} {12}},\
  \bibinfo {pages} {083039} (\bibinfo {year} {2010})},\ \Eprint
  {http://arxiv.org/abs/1004.3791} {1004.3791} \BibitemShut {NoStop}%
\bibitem [{\citenamefont {Calderbank}\ and\ \citenamefont
  {Shor}(1996)}]{CalderbankShor1996Good}%
  \BibitemOpen
  \bibfield  {author} {\bibinfo {author} {\bibfnamefont {A.~R.}\ \bibnamefont
  {Calderbank}}\ and\ \bibinfo {author} {\bibfnamefont {P.~W.}\ \bibnamefont
  {Shor}},\ }\bibfield  {title} {\enquote {\bibinfo {title} {Good quantum
  error-correcting codes exist},}\ }\href {\doibase 10.1103/PhysRevA.54.1098}
  {\bibfield  {journal} {\bibinfo  {journal} {Phys. Rev. A}\ }\textbf {\bibinfo
  {volume} {54}},\ \bibinfo {pages} {1098--1105} (\bibinfo {year} {1996})},\
  \Eprint {http://arxiv.org/abs/quant-ph/9512032} {quant-ph/9512032}
  \BibitemShut {NoStop}%
\bibitem [{\citenamefont {Gottesman}(1997)}]{gottesmanthesis}%
  \BibitemOpen
  \bibfield  {author} {\bibinfo {author} {\bibfnamefont {D.}~\bibnamefont
  {Gottesman}},\ }\bibfield  {title} {\enquote {\bibinfo {title} {Stabilizer
  codes and quantum error correction},}\ }\href@noop {} {\  (\bibinfo {year}
  {1997})},\ \Eprint {http://arxiv.org/abs/quant-ph/9705052} {quant-ph/9705052}
  \BibitemShut {NoStop}%
\bibitem [{\citenamefont {Preskill}()}]{preskill}%
  \BibitemOpen
  \bibfield  {author} {\bibinfo {author} {\bibfnamefont {J.}~\bibnamefont
  {Preskill}},\ }\bibfield  {title} {\enquote {\bibinfo {title} {Lecture notes
  on quantum computation},}\ }\href
  {http://www.theory.caltech.edu/people/preskill/ph229/notes/chap7.pdf}
  {\bibinfo  {journal} {Caltech Ph219}\ }\BibitemShut {NoStop}%
\bibitem [{\citenamefont {Sipser}\ and\ \citenamefont
  {Spielman}(1996)}]{sipser1996expander}%
  \BibitemOpen
\bibfield  {journal} {  }\bibfield  {author} {\bibinfo {author} {\bibfnamefont
  {M.}~\bibnamefont {Sipser}}\ and\ \bibinfo {author} {\bibfnamefont {D.~A.}\
  \bibnamefont {Spielman}},\ }\bibfield  {title} {\enquote {\bibinfo {title}
  {Expander codes},}\ }\href {\doibase 10.1109/18.556667} {\bibfield  {journal}
  {\bibinfo  {journal} {IEEE Transactions on Information Theory}\ }\textbf
  {\bibinfo {volume} {42}},\ \bibinfo {pages} {1710--1722} (\bibinfo {year}
  {1996})}\BibitemShut {NoStop}%
\bibitem [{\citenamefont {Furedi}\ \emph {et~al.}(1995)\citenamefont {Furedi},
  \citenamefont {Lazebnik}, \citenamefont {Seress}, \citenamefont {Ustimenko},\
  and\ \citenamefont {Woldar}}]{furedi1995graphs}%
  \BibitemOpen
  \bibfield  {author} {\bibinfo {author} {\bibfnamefont {Z.}~\bibnamefont
  {Furedi}}, \bibinfo {author} {\bibfnamefont {F.}~\bibnamefont {Lazebnik}},
  \bibinfo {author} {\bibfnamefont {A.}~\bibnamefont {Seress}}, \bibinfo
  {author} {\bibfnamefont {V.~A.}\ \bibnamefont {Ustimenko}}, \ and\ \bibinfo
  {author} {\bibfnamefont {A.~J.}\ \bibnamefont {Woldar}},\ }\bibfield  {title}
  {\enquote {\bibinfo {title} {Graphs of prescribed girth and bi-degree},}\
  }\href {\doibase 10.1006/jctb.1995.1033} {\bibfield  {journal} {\bibinfo
  {journal} {Journal of Combinatorial Theory, Series B}\ }\textbf {\bibinfo
  {volume} {64}},\ \bibinfo {pages} {228--239} (\bibinfo {year}
  {1995})}\BibitemShut {NoStop}%
\bibitem [{\citenamefont {Raussendorf}\ \emph {et~al.}(2007)\citenamefont
  {Raussendorf}, \citenamefont {Harrington},\ and\ \citenamefont
  {Goyal}}]{RaussendorfHarringtonGoyal2007}%
  \BibitemOpen
  \bibfield  {author} {\bibinfo {author} {\bibfnamefont {R.}~\bibnamefont
  {Raussendorf}}, \bibinfo {author} {\bibfnamefont {J.}~\bibnamefont
  {Harrington}}, \ and\ \bibinfo {author} {\bibfnamefont {K.}~\bibnamefont
  {Goyal}},\ }\bibfield  {title} {\enquote {\bibinfo {title} {Topological
  fault-tolerance in cluster state quantum computation},}\ }\href {\doibase
  10.1088/1367-2630/9/6/199} {\bibfield  {journal} {\bibinfo  {journal} {New
  Journal of Physics}\ }\textbf {\bibinfo {volume} {9}},\ \bibinfo {pages}
  {199} (\bibinfo {year} {2007})},\ \Eprint
  {http://arxiv.org/abs/quant-ph/0703143v1} {quant-ph/0703143v1} \BibitemShut
  {NoStop}%
\bibitem [{\citenamefont {Fowler}\ \emph {et~al.}(2012)\citenamefont {Fowler},
  \citenamefont {Mariantoni}, \citenamefont {Martinis},\ and\ \citenamefont
  {Cleland}}]{Fowler2012}%
  \BibitemOpen
  \bibfield  {author} {\bibinfo {author} {\bibfnamefont {A.~G.}\ \bibnamefont
  {Fowler}}, \bibinfo {author} {\bibfnamefont {M.}~\bibnamefont {Mariantoni}},
  \bibinfo {author} {\bibfnamefont {J.~M.}\ \bibnamefont {Martinis}}, \ and\
  \bibinfo {author} {\bibfnamefont {A.~N.}\ \bibnamefont {Cleland}},\
  }\bibfield  {title} {\enquote {\bibinfo {title} {Surface codes: Towards
  practical large-scale quantum computation},}\ }\href {\doibase
  10.1103/PhysRevA.86.032324} {\bibfield  {journal} {\bibinfo  {journal} {Phys.
  Rev. A}\ }\textbf {\bibinfo {volume} {86}},\ \bibinfo {pages} {032324}
  (\bibinfo {year} {2012})},\ \Eprint {http://arxiv.org/abs/1208.0928}
  {1208.0928} \BibitemShut {NoStop}%
\bibitem [{\citenamefont {O'Gorman}\ and\ \citenamefont
  {Campbell}(2017)}]{OGormanCampbell2017}%
  \BibitemOpen
  \bibfield  {author} {\bibinfo {author} {\bibfnamefont {J.}~\bibnamefont
  {O'Gorman}}\ and\ \bibinfo {author} {\bibfnamefont {E.~T.}\ \bibnamefont
  {Campbell}},\ }\bibfield  {title} {\enquote {\bibinfo {title} {Quantum
  computation with realistic magic-state factories},}\ }\href {\doibase
  10.1103/PhysRevA.95.032338} {\bibfield  {journal} {\bibinfo  {journal} {Phys.
  Rev. A}\ }\textbf {\bibinfo {volume} {95}},\ \bibinfo {pages} {032338}
  (\bibinfo {year} {2017})},\ \Eprint {http://arxiv.org/abs/1605.07197}
  {1605.07197} \BibitemShut {NoStop}%
\bibitem [{\citenamefont {Grassl}\ and\ \citenamefont {Beth}(1999)}]{qbch}%
  \BibitemOpen
  \bibfield  {author} {\bibinfo {author} {\bibfnamefont {M.}~\bibnamefont
  {Grassl}}\ and\ \bibinfo {author} {\bibfnamefont {T.}~\bibnamefont {Beth}},\
  }\bibfield  {title} {\enquote {\bibinfo {title} {Quantum bch codes},}\ }in\
  \href@noop {} {\emph {\bibinfo {booktitle} {Proceedings X. International
  Symposium on Theoretical Electrical Engineering, Magdeburg}}}\ (\bibinfo
  {year} {1999})\ pp.\ \bibinfo {pages} {207--212},\ \Eprint
  {http://arxiv.org/abs/quant-ph/9910060} {quant-ph/9910060} \BibitemShut
  {NoStop}%
\bibitem [{\citenamefont {Haah}\ \emph {et~al.}()\citenamefont {Haah},
  \citenamefont {Hastings}, \citenamefont {Poulin},\ and\ \citenamefont
  {Wecker}}]{enum}%
  \BibitemOpen
  \bibfield  {author} {\bibinfo {author} {\bibfnamefont {J.}~\bibnamefont
  {Haah}}, \bibinfo {author} {\bibfnamefont {M.~B.}\ \bibnamefont {Hastings}},
  \bibinfo {author} {\bibfnamefont {D.}~\bibnamefont {Poulin}}, \ and\ \bibinfo
  {author} {\bibfnamefont {D.}~\bibnamefont {Wecker}},\ }\bibfield  {title}
  {\enquote {\bibinfo {title} {Magic state distillation at intermediate
  size},}\ }\href@noop {} {\ }\Eprint {http://arxiv.org/abs/1709.02789}
  {1709.02789} \BibitemShut {NoStop}%
\bibitem [{\citenamefont {Hastings}()}]{Hastings}%
  \BibitemOpen
  \bibfield  {author} {\bibinfo {author} {\bibfnamefont {M.~B.}\ \bibnamefont
  {Hastings}},\ }\bibfield  {title} {\enquote {\bibinfo {title} {Small majorana
  fermion codes},}\ }\href@noop {} {\ }\Eprint
  {http://arxiv.org/abs/1703.00612} {1703.00612} \BibitemShut {NoStop}%
\bibitem [{\citenamefont {Campbell}\ \emph {et~al.}(2012)\citenamefont
  {Campbell}, \citenamefont {Anwar},\ and\ \citenamefont
  {Browne}}]{Campbell2012}%
  \BibitemOpen
  \bibfield  {author} {\bibinfo {author} {\bibfnamefont {E.~T.}\ \bibnamefont
  {Campbell}}, \bibinfo {author} {\bibfnamefont {H.}~\bibnamefont {Anwar}}, \
  and\ \bibinfo {author} {\bibfnamefont {D.~E.}\ \bibnamefont {Browne}},\
  }\bibfield  {title} {\enquote {\bibinfo {title} {Magic state distillation in
  all prime dimensions using quantum reed-muller codes},}\ }\href {\doibase
  10.1103/PhysRevX.2.041021} {\bibfield  {journal} {\bibinfo  {journal} {Phys.
  Rev. X}\ }\textbf {\bibinfo {volume} {2}},\ \bibinfo {pages} {041021}
  (\bibinfo {year} {2012})},\ \Eprint {http://arxiv.org/abs/1205.3104}
  {1205.3104} \BibitemShut {NoStop}%
\bibitem [{\citenamefont {Campbell}(2014)}]{Campbell2014}%
  \BibitemOpen
  \bibfield  {author} {\bibinfo {author} {\bibfnamefont {E.~T.}\ \bibnamefont
  {Campbell}},\ }\bibfield  {title} {\enquote {\bibinfo {title} {Enhanced
  fault-tolerant quantum computing in $d$-level systems},}\ }\href {\doibase
  10.1103/PhysRevLett.113.230501} {\bibfield  {journal} {\bibinfo  {journal}
  {Phys. Rev. Lett.}\ }\textbf {\bibinfo {volume} {113}},\ \bibinfo {pages}
  {230501} (\bibinfo {year} {2014})},\ \Eprint {http://arxiv.org/abs/1406.3055}
  {1406.3055} \BibitemShut {NoStop}%
\bibitem [{\citenamefont {Nebe}\ \emph {et~al.}(2001)\citenamefont {Nebe},
  \citenamefont {Rains},\ and\ \citenamefont {Sloane}}]{NebeRainsSloane2001}%
  \BibitemOpen
  \bibfield  {author} {\bibinfo {author} {\bibfnamefont {G.}~\bibnamefont
  {Nebe}}, \bibinfo {author} {\bibfnamefont {E.~M.}\ \bibnamefont {Rains}}, \
  and\ \bibinfo {author} {\bibfnamefont {N.~J.~A.}\ \bibnamefont {Sloane}},\
  }\bibfield  {title} {\enquote {\bibinfo {title} {The invariants of the
  clifford groups},}\ }\href {\doibase 10.1023/a:1011233615437} {\bibfield
  {journal} {\bibinfo  {journal} {Designs, Codes and Cryptography}\ }\textbf
  {\bibinfo {volume} {24}},\ \bibinfo {pages} {99--122} (\bibinfo {year}
  {2001})},\ \Eprint {http://arxiv.org/abs/math/0001038} {math/0001038}
  \BibitemShut {NoStop}%
\bibitem [{\citenamefont {Nebe}\ \emph {et~al.}(2006)\citenamefont {Nebe},
  \citenamefont {Rains},\ and\ \citenamefont
  {Sloane}}]{NebeRainsSloane2006book}%
  \BibitemOpen
  \bibfield  {author} {\bibinfo {author} {\bibfnamefont {G.}~\bibnamefont
  {Nebe}}, \bibinfo {author} {\bibfnamefont {E.~M.}\ \bibnamefont {Rains}}, \
  and\ \bibinfo {author} {\bibfnamefont {N.~J.~A.}\ \bibnamefont {Sloane}},\
  }\href {\doibase 10.1007/3-540-30731-1} {\emph {\bibinfo {title} {Self-Dual
  Codes and Invariant Theory}}}\ (\bibinfo  {publisher} {Springer-Verlag},\
  \bibinfo {year} {2006})\BibitemShut {NoStop}%
\bibitem [{\citenamefont {Assmus}\ and\ \citenamefont
  {Key}(1992)}]{assmus1992designs}%
  \BibitemOpen
  \bibfield  {author} {\bibinfo {author} {\bibfnamefont {E.~F.}\ \bibnamefont
  {Assmus}}\ and\ \bibinfo {author} {\bibfnamefont {J.~D.}\ \bibnamefont
  {Key}},\ }\href@noop {} {\emph {\bibinfo {title} {Designs and their
  Codes}}},\ \bibinfo {number} {103}\ (\bibinfo  {publisher} {Cambridge
  University Press},\ \bibinfo {year} {1992})\BibitemShut {NoStop}%
\bibitem [{\citenamefont {Vijay}\ and\ \citenamefont {Fu}()}]{RMqubit}%
  \BibitemOpen
  \bibfield  {author} {\bibinfo {author} {\bibfnamefont {S.}~\bibnamefont
  {Vijay}}\ and\ \bibinfo {author} {\bibfnamefont {L.}~\bibnamefont {Fu}},\
  }\bibfield  {title} {\enquote {\bibinfo {title} {Quantum error correction for
  complex and majorana fermion qubits},}\ }\href@noop {} {\ }\Eprint
  {http://arxiv.org/abs/1703.00459} {1703.00459} \BibitemShut {NoStop}%
\bibitem [{\citenamefont {Lang}(2002)}]{Lang}%
  \BibitemOpen
  \bibfield  {author} {\bibinfo {author} {\bibfnamefont {S.}~\bibnamefont
  {Lang}},\ }\href@noop {} {\emph {\bibinfo {title} {Algebra}}},\ \bibinfo
  {edition} {revised 3rd}\ ed.\ (\bibinfo  {publisher} {Springer},\ \bibinfo
  {year} {2002})\BibitemShut {NoStop}%
\end{thebibliography}%

\appendix

\section{Specific Small Inner and Outer Codes}
\label{app:codes}
In this appendix, we give some specific inner and outer codes, 
either giving the stabilizers or referring to the literature.
Some of these codes are explained in the basic distillation section~\ref{sec:basic} 
or in numerical simulations~\ref{sec:simulation} in the body of the paper.
Other codes have other useful properties that we describe for specific codes.

When we give stabilizers for an inner code, each row gives one stabilizer generator.
Each row consists of a binary string, of length equal to the number of qubits, 
with a $1$ indicating that that stabilizer acts on that qubit, i.e., we give the parity check matrix.

\subsection{Inner Codes}

\subsubsection{$[[4,2,2]]$ Inner Code}

This is explained in Section~\ref{sec:basic}.

\subsubsection{$[[16,6,4]]$ Inner Code}
\label{app1664}
This is from Ref.~\cite{Hastings}.
The stabilizer matrix is the classical Hadamard code $[16,5,8]$.
\begin{center}
\begin{tabular}{l}
1111111111111111\\
1111111100000000\\
1111000011110000\\
1100110011001100\\
1010101010101010
\end{tabular}
\end{center}

\subsubsection{$[[7,1,3]]$ Inner Code}
This is explained in Section~\ref{sec:basic}.

\subsubsection{$[[17,1,5]]$ Inner Code}
This is an instance of color code~\cite{Bombin_2006,BravyiCross2015}.
It is the smallest normal code that we found with $\kin=1$ and distance $5$.
The stabilizers are: 
\begin{center}
\begin{tabular}{l}
11011010101000010\\
01100011001100110\\
00110110010011001\\
00010101000111110\\
00001110010011101\\
00000101000110000\\
00000011111011010\\
00000001010100001
\end{tabular}
\end{center}

\subsubsection{$[[21,3,5]]$ and $[[23,1,7]]$ Inner Codes}
The (extended) Golay code is a classical self-dual code which has parameters $[24,12,8]$.
Puncturing a bit by collecting all code words that has zero on that bit,
we obtain a self-orthogonal $[23,11,7]$.
From this, we obtain a weakly self-dual CSS code which is $[[23,1,7]]$.
(Reichardt has used this code in a very different distillation protocol~\cite{Reichardt2005}.)
There are many positions to puncture, but due to high symmetry of the Golay code,
the resulting codes have the same weight enumerators.
One can pipeline the
$[[23,1,7]]$ code after
the protocol of section \ref{1715} 
to give a protocol with one output magic state and seventh order suppression in error.

By puncturing the $[[23,1,7]]$ code twice, we obtain a $[[21,3,5]]$ code.
In a numerical search, we did not find any smaller normal code with $\kin=3$ and distance $5$.  The stabilizers of the $[[21,3,5]]$ code are:

\begin{center}
\begin{tabular}{l}
100000000011110110100\\
010000000001111011010\\
001000000110110011001\\
000100000011011001101\\
000010000001101100111\\
000001000110111000110\\
000000100101010010111\\
000000010100100111110\\
000000001100011101011
\end{tabular}
\end{center}

\subsubsection{Other Inner Codes}
Some other examples of inner codes can be found in Ref.~\cite{Hastings},
from which we reproduce optimal $\kin$ found for given distance and $\nin$ in Table~\ref{tb:evenCodeTable}.
For stabilizers, see Ref.~\cite{Hastings}.

\begin{table}
\centering
\begin{tabular}{|c||c|c|c|c|c||c|c|c|}
\hline
$\nin$ & 16 & 20 & 24 & 28 & 30 & 20 & 28 & 30 \\
\hline
$\kin$ & 6  & 8  & 12 & 14 & 16 & 2  & 4  & 6 \\ 
\hline
$d$    & 4  & 4  & 4  & 4  & 4  & 6  & 6  & 6 \\
\hline
\end{tabular}
\caption{Parameters of small hyperbolic weakly self-dual CSS codes~\cite{Hastings}.
The code $[[20,2,6]]$ can be constructed from the five-qubit code $[[5,1,3]]$
by going through the Majorana operators~\cite{BravyiLeemhuisTerhal2010Majorana},
while the others cannot be constructed in this way.
}
\label{tb:evenCodeTable}
\end{table}

\subsection{Outer Codes}
\subsubsection{Petersen Graph Code}
The outer code in section \ref{2135} has $4$ qubits
uses $4$ checks of weight $3$.
However, from Lemma~\ref{bulkoc}, 
we know that there is some $\nout$ such that there is a code which is $(4,2)$ sensitive with weight-$3$ checks, 
which has only $(2/3) \nout$ checks.  
We now explain this code.
The proof of Lemma~\ref{bulkoc} reduces the problem of finding such a code to finding a bipartite graph $G$.  
Since the set $B$ of that lemma has degree $2$, 
we can equivalently define the code by a graph $H$ 
such that the vertices of the graph $H$ correspond to checks and the edges correspond to bits; 
i.e., in the case that $B$ has degree $2$, 
the possible bipartite graphs $G$ are in one-to-one correspondence with degree-$3$ graphs $H$.  
Then, from the proof of Lemma~\ref{bulkoc} 
we know that if $H$ has girth at least $5$, 
then the corresponding code is $(4,2)$ sensitive.  
The smallest such graph $H$ is known to be the Petersen graph.  
This is a degree-$3$ graph with $15$ vertices and $10$ edges.
This graph can be thought of as the dodecahedron with antipodes identified.
Note that the girth being 5 is optimal in this case,
because if $H$ has girth $4$, then there is a weight $4$ error that violates no checks.

\section{Circuits}
\label{app:circuits}
In this section we give circuits for some of the protocols above.  Boxes labelled $Enc$ or $Enc'$ denote encoding and decoding circuits, which are Cliffords.  The number in the box indicates what code is used.  $H$ denotes Haamard, $M$ denotes measurement in $Z$ basis, $JM_x$ denotes measurement in $X$ basis, $Czs$ denotes control-$Z$ operations.

\begin{figure}
\includegraphics[width=6in]{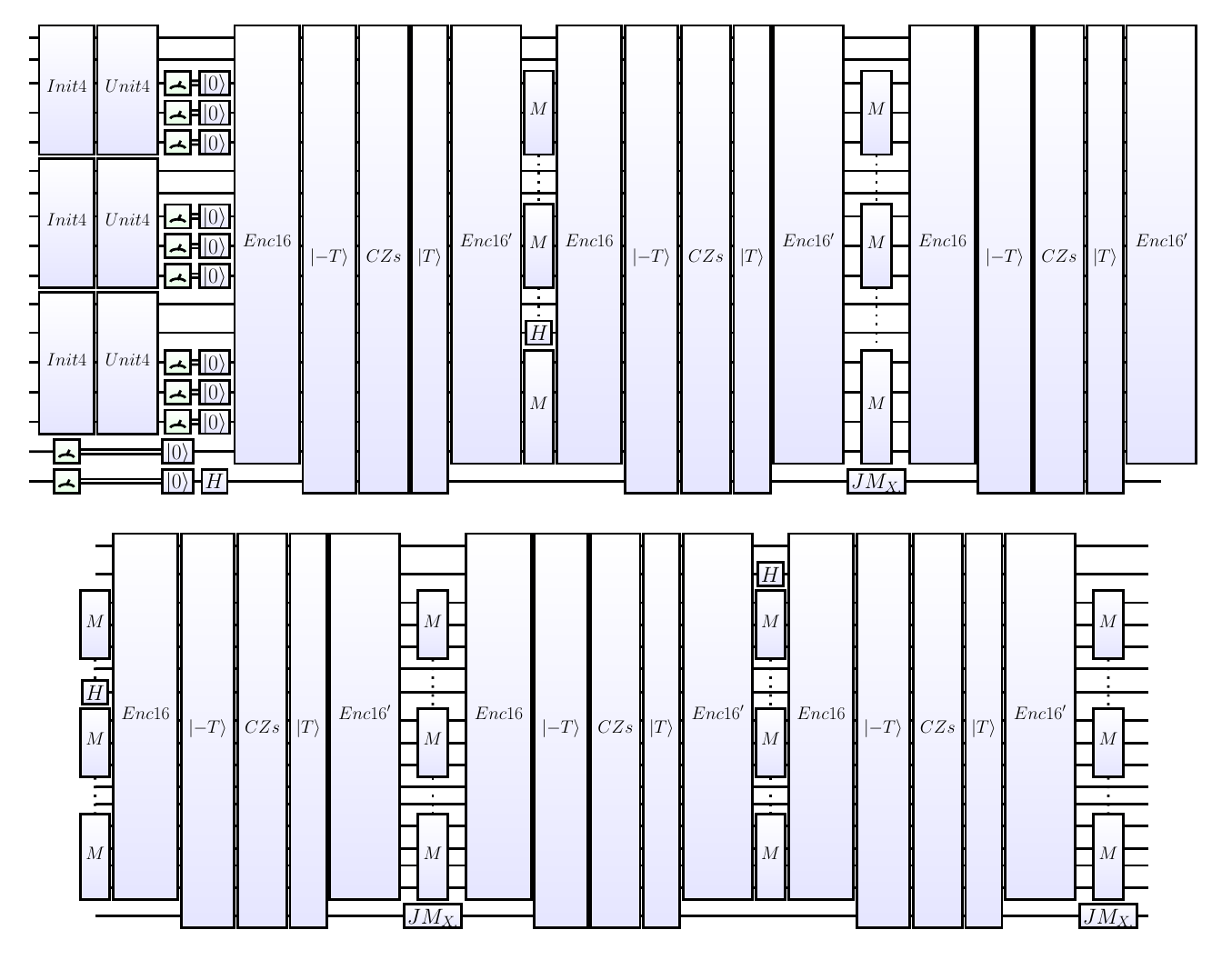}
\caption{Pipelined circuit using $[[16,6,4]]$ code described in section \ref{1664}.}
\label{1664fig}
\end{figure}

\begin{figure}
\includegraphics[width=6in]{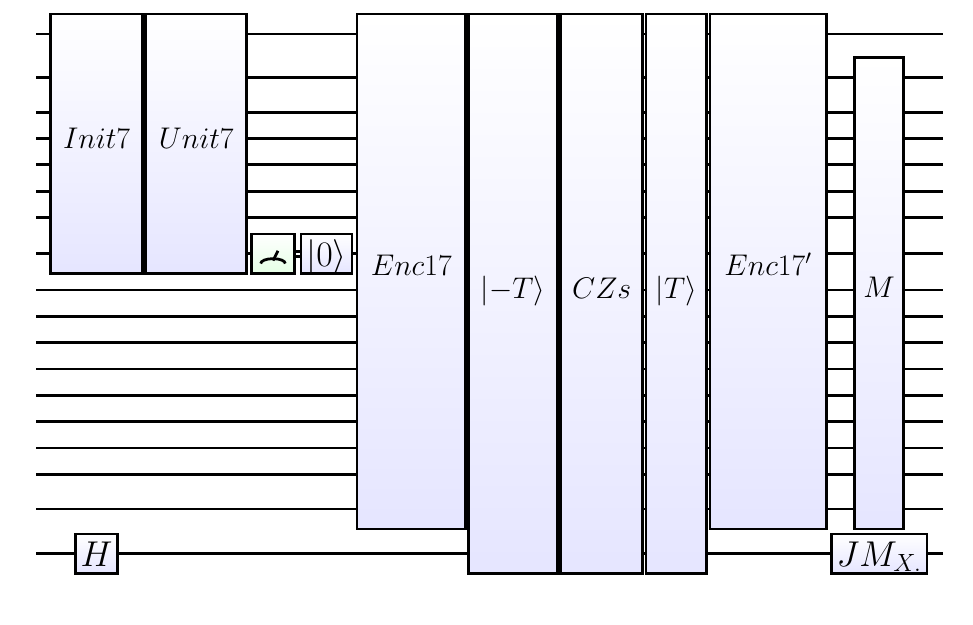}
\caption{Pipelined circuit using $[[17,1,5]]$ code described in section \ref{1715}.}
\label{1715fig}
\end{figure}

\begin{figure}
\includegraphics[width=6in]{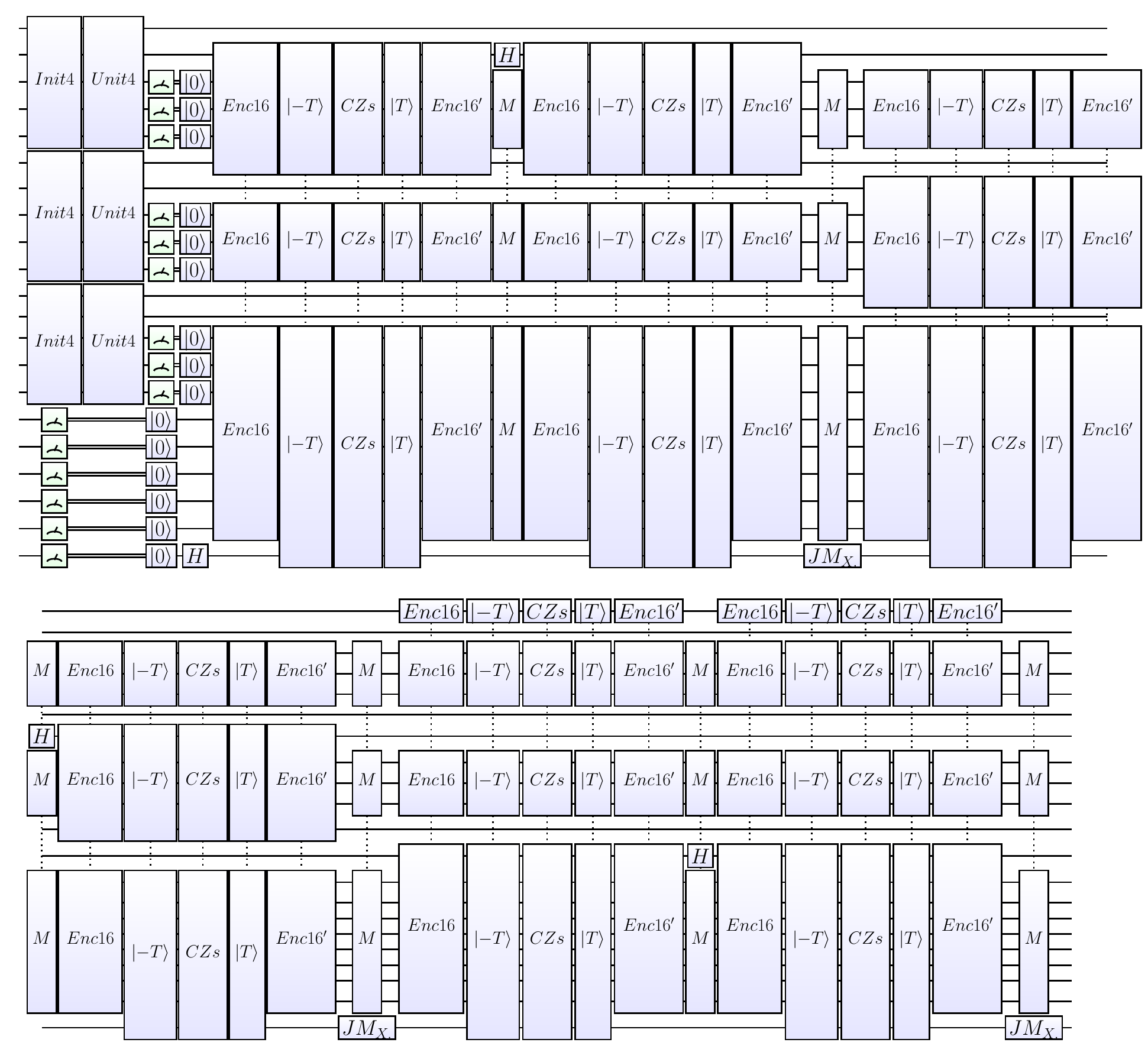}
\caption{Pipelined circuit using $[[16,2,4]]$ code described in section \ref{1624}.}
\label{1624fig}
\end{figure}

\begin{figure}
\includegraphics[width=6in]{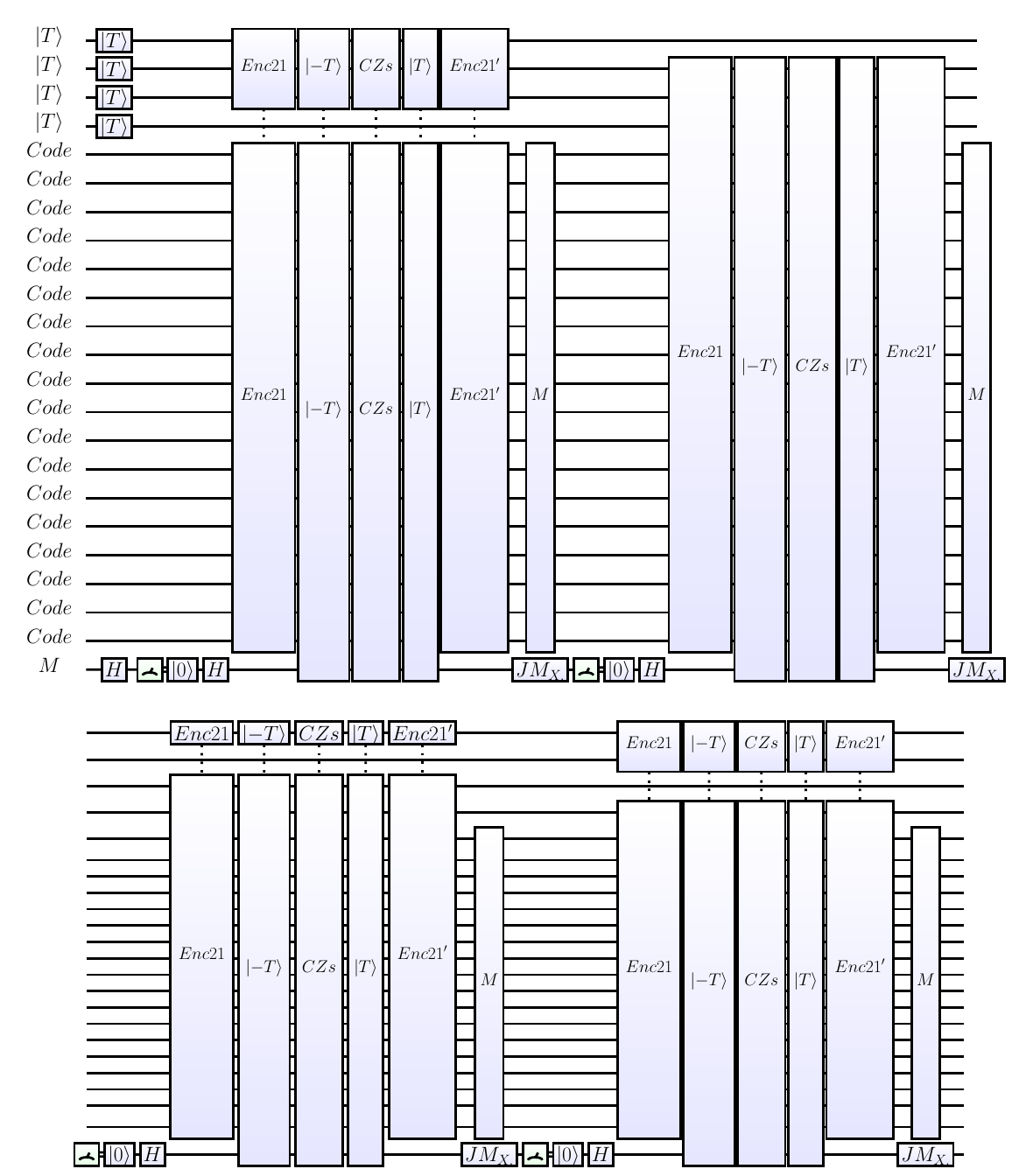}
\caption{Circuit using $[[21,3,5]]$ code described in section \ref{2135}}
\label{2135fig}
\end{figure}

% \FloatBarrier

\section{Coincidence among protocols}
\label{app:coincidence}

The Steane code has 7 $Y$-logical operators of weight 3.
In the distillation protocol using the Steane code as the inner code,
each logical error may appear in 4 different ways in the column that implements control-$H^{\otimes 7}$~\cite{Knill2004a}.
The measurement error at the lowest order can happen in 7 ways.
Overall, the cubic error can happen in $7 \cdot 4 + 7 =35$ ways.
This number matches the number of logical operators of weight 3
in the Bravyi-Kitaev 15-to-1 protocol~\cite{BravyiKitaev2005Magic}.
Reichardt~\cite{Reichardt2005} has noted this equivalence.

When we pipelined $[[7,1,3]]$ to $[[17,1,5]]$,
there are 48 $T$ gates and 1 $T$ states.
The number of logical operators of weight 5 in $[[17,1,5]]$ is 51.
Each logical operator can appear in 16 different configurations 
in the column that implements control-$H^{\otimes 17}$.
The measurement error from the 17-qubit code routine
occurs in 17 ways at the leading order.
Thus, the output error probability has leading term
$(51 \cdot 16 + 35 \cdot 17) \epsilon^5 = 1411 \epsilon^5$.
The coefficient matches the number of $Z$-logical operators of weight 5 
in the $[[49,1,5]]$ code, as reported in \cite[App.~B]{BravyiHaah2012Magic}.

Bravyi and Cross~\cite{BravyiCross2015}
gave a recursive construction for triply even codes.
They showed how to convert a pair of 
a (classical) triply even code of length $n_{t-1}$ with dual distance $2t-1$
and some (classical) self-orthogonal code of length $m_t$ with dual distance $2t+1$
into
a triply even code of length $n_t = 2m_t + n_{t-1} $ with dual distance $2t+1$.
The formula gives another coincidence with our pipeline.
$n_{t-1}$ is the number of $T$ gates/states, sitting before the final $H$-measurement routine
in the pipeline, and $m_t$ is the code length of the final $H$-measurement routine.
Thus, the recursive formula $n_t = 2m_t + n_{t-1}$ correctly counts
the number of $T$ gates/states used in the pipeline.

A similar coincidence was observed by Jones~\cite{Jones2012},
where the leading error probabilities of the distillation protocols by 
a family of weakly self-dual $[[k+4,k,2]]$ codes with $(k,0)$-magic basis
and those by a family of triorthogonal codes~\cite{BravyiHaah2012Magic}
are shown to be the same as $(3k+1)\epsilon^2$.
The total number of $T$ gates/states were also the same as $3k+8$.

\section{Qudits}
\label{app:qudits}

In this section, we consider an extension to qudits with local Hilbert space dimension $p > 2$, 
with $p$ a prime.
Previously, Reed-Muller codes over prime fields were used~\cite{Campbell2012,Campbell2014},
but our approach is more efficient.
In terms of the scaling exponent $\gamma$ 
(see Sec.~\ref{sec:asymptotic}),
previous schemes for a fixed $p$ 
did not achieve $\gamma \to 1$,
whereas our protocols below will.%
\footnote{
In Ref.~\cite{Campbell2014}, it is shown that $\gamma$ can be arbitrary close to $1$
in the limit of large $p$.
}
Specifically, we consider outer codes defined by {\em binary}
parity check matrices $M$ that are sensitive enough.
Those from biregular bipartite graph can be used.
The parity check matrices being binary is
for technical simplicity, 
and is not conceptually crucial ingredient,
though relaxing this condition might involve complicated calculation.
The main difference from the qubit protocols
is in inner codes.
We only consider analogs of normal codes,
where transversal $S$ gates become logical $S$ gates.
We show that measurement in the eigenbasis of $SX$
can be implemented fault-tolerantly 
on an inner code of parameters $[[\nin,\kin,d]]_p$ using $4\nin$ $T$ gates
if $p > 3$, or $2\nin$ $T$ gates if $p = 3$.
Therefore, the Lemma~\ref{senslemma} generalizes to odd prime dimensional qudits, 
where input $T$ count is $\nout + 4 \nin m$ if $p > 3$ 
where $m$ is the number of checks in the outer code, 
or $\nout + 2\nin m$ if $p = 3$.
We complement the construction in this section with a probabilistic existence proof for a good family of inner codes in the sense of Section~\ref{sec:asymptotic};
see Lemma~\ref{lem:qudit-rwsd} and 
combine it with the proof of Lemma~\ref{lem:randomcodes-distance}.
We also point out that quantum Reed-Muller code gives a family whose encoding rate approaches 1 for a given distance.
Hence, the statement of Theorem~\ref{thm:dasymptotic} remains unchanged in either case where $p =3$ or $p > 3$,
but that of Theorem~\ref{thm:noutasymptotic} for $p >3 $ becomes that, for odd $d \ge 5$ the number of input magic states per output approaches 
$1+4(d-1)/2 = 2d-1$ in the large code length limit.
The input $T$ count per output 
is still $d(1+o(1))$ if $p = 3$.
If one wishes to improve the asymptotic input count for $p > 3$,
then one has to solve an equation analogous to 
\eqref{eq:decomposition-CSX}.

Consider a basis of states $\ket j$, where $j=0,1,\ldots, p-1$ is periodic mod $p$.
We use the following operators and phase factor
\begin{align}
\omega&=e^{2\pi i/p}, &
 Z&=\sum_j \omega^j \ket j \bra j ,&
 X&=\sum_{j} \ket{j+1}\bra{j} ,
\nonumber \\
 H&=\frac{1}{\sqrt{p}} \sum_{j,k}\omega^{jk} \ket j \bra k ,&
 S&=\sum_j \omega^{j(j-1)/2}  \ket j \bra j , &
^CX &= \sum_j \ket j \bra j \otimes X^j , \nonumber\\
&& U(n) &= \sum_j \ket{nj}\bra j \quad (n \neq 0)
\end{align}
which generate the Clifford group.
It holds that $ZX=\omega XZ$.

We will work with a generalization of {\it normal} codes throughout this section, 
ignoring hyperbolic codes.  
One reason is that we cannot achieve control-Swap in the same way as we could previously.
The general method in the qubit case was to use some non-Clifford operation such as a $T$ gate,
conjugating controlled Pauli to obtain control-Swap on the code space of some code.
However, Swap is of order 2 while control-$Z$ is of order $p$.
One might hope to obtain a control-permutation of order $p$, 
but we do not consider this possibility.
For normal codes, we do not try to implement the control-Hadamard as was done before, 
because Hadamard is of order $4$ for $p>2$, and hence is not conjugate to control-$Z$.

\subsection{Preliminary}
Let us first define a $T$-gate~\cite{Campbell2012}.
The cases $p=3$ and $p>3$ are going to be different.
Define
\begin{align}
g(j) &:= \sum_{k=0}^{j} \frac 1 2 k(k-1) = \frac 1 6 (j^3-j),\\
g( j + p) &= g(j) \mod p \quad \text{ if $ p > 3$}
\end{align}
where the second line is because $6$ is invertible in $\FF_{p > 3}$,
and ensures that $g$ is a well defined function on $\FF_{p>3}$.
All arithmetic in the exponent of 
$\omega$, $Z$, $X$, and $S$ will be over $\FF_p$ for both $p = 3$ and $p > 3$.
Define the $T$-gate as
\begin{align}
T &= \sum_{j=0}^{p-1} \omega^{g(j)} \ket j \bra j,&
TX T^{-1} X^{-1} &= 
%\ket 0 \bra 0 + \sum_{j=1}^{p-1} \omega^{g(j)-g(j-1)} \ket j \bra j = 
S & \text{if $p > 3$,}\label{thg3}
\\
T &= \ket 0 \bra 0 + e^{-2\pi i/ 9} \ket 1 \bra 1 + e^{2 \pi i / 9} \ket 2 \bra 2,&
T X T^{-1} X^{-1} &= e^{-2 \pi i /9} S &\text{if $p=3$}. \label{th3}
\end{align}
These show that in both cases the $T$ gate is at the third level of the generalized Clifford hierarchy.
More generally, we find
\begin{align}
T^m X T^{-m}&=
\begin{cases}
                   S^m X & \text{ for } p > 3, \\
e^{-2 \pi i m / 9} S^m X & \text{ for } p = 3.
\end{cases} 
\label{eq:TXTeqSX} 
\end{align}
For both $p=3$ and $p>3$, 
define $\ket{\psi_m}$ for $m = 0,1,\ldots, p-1$ be the $(+1)$-eigenstate of $T^m X T^{-m}$:
\begin{align}
T^m X T^{-m} \ket{\psi_m} =  \ket{\psi_m}.  
 % \begin{cases}
 % S^m X \ket{\psi_m} & \text{ if } p > 3, \\
 % e^{-2\pi i m / 9}S^m X \ket{\psi_m}& \text{ if } p = 3.
 % \end{cases}
\end{align}
% $\psi(m)$ for $m=0,\ldots,p-1$
%where $\psi(0)$ is a $+1$ eigenstate of $X$ and $\psi(m)$ is given by $\psi(m)=T^m \psi(0)$.
Any state $\ket{\psi_m}$ for $m=1,\ldots,p-1$ will be a ``magic state.''
%Note that for $p>3$, $T^p \psi(0)$ is equal to $\psi(0)$ times a global phase.

How would one use these magic states?
Suppose $p>3$.
Consider a pair of qudits in a state
$\sum_j a_j \ket j \otimes \ket{\psi_m}$.
Apply a control-$X$ operation
with the first qudit as source and the second qudit as target.
This maps the state to
\begin{align}
\frac{1}{\sqrt{p}} \sum_{j,k} a_j \omega^{m g(k)} \ket{j,k+j}.
\end{align}
Now measure the second qudit in the computational basis, obtaining a result $\ell$.
This gives a state on the first qudit
$\sum_j a_j \omega^{m g(\ell-j)} \ket j$.
Thus, the transformation implemented on the first qudit is $\sum_j \omega^{mg(\ell-j)} \ket j \bra j$.
Expanding the exponent, we have
\begin{align}
\label{inject}
m g(\ell-j)&= mg(\ell)-mg(j)+\frac{m}{2}( \ell j^2 - \ell^2  j) \\ \nonumber
&=
m g(\ell) - mg(j) +m \ell \frac{j(j-1)}{2} - m \frac{\ell(\ell - 1) }{2} j.
\end{align}
The first term on the right-hand side of Eq.~(\ref{inject}) corresponds to an irrelevant global phase factor.
The second term, $-mg(j)$, corresponding to implementing transformation $T^{-m}$ on the first qudit.  
The third term gives a phase factor that can be corrected by applying a power of the $S$ gate
and the last term gives phase factors that can be corrected by a power of the $Z$ gate.
Thus, the state injection procedure works, 
in that we can use a magic state $\ket{\psi_m}$ to produce a transformation $T^{-m}$ up to Clifford corrections.

When $p=3$, we use the same state injection, with $m=1$.
One finds after some calculation that if the measurement outcome is $\ell=0$,
the implemented operations is $T^{-1}$ to the source,
if $\ell=1$, it is $e^{-2\pi i / 9} ST^{-1}$,
and if $\ell=2$, it is $e^{2\pi i / 9} Z^{-1} S^{-1}  T^{-1}$.  
Thus, in all cases, the implemented operation is $T^{-1}$ up to a Clifford correction.

The injected $T$ gates together with Cliffords
form a universal gate set~\cite[App.~D]{Campbell2012}.
This is a corollary of \cite[Thm.~7.3]{NebeRainsSloane2001}
that says the Clifford group is a maximal finite subgroup of $U(p^n)$ up to global phase factors,
and \cite[Cor.~6.8.2]{NebeRainsSloane2006book}
that says any infinite subgroup (even after quotienting out phase factors)
containing the Clifford group is dense in $U(p^n)$.

Note that $T^m$ and $T^{-m}$ are interconvertible by Cliffords.
More generally,
it is possible to use Clifford operations to convert a gate $T^m$ 
into another gate $T^{m'}$ with $m'=mn^3$ for $n \neq 0$,
by $U = U(n) = \sum_j \ket{nj}\bra j$ gate.
For $p > 3$, we have $U^\dagger T^m U = \sum_j \omega^{m g(nj)} \ket j \bra j$
where
\begin{align}
\label{inter}
g(nj) 
= \frac{1}{6}(n^3j^3-nj) 
= n^3 g(j)+\frac{n^3-n}{6} j.
\end{align}
Thus $U^\dagger T^m U = T^{mn^3} Z^{m\frac{n^3-n}{6}}$, 
and so indeed $T^m=C_1 T^{m'} C_2$ for some Cliffords $C_1,C_2$.
For $p=3$, we see $T = U(-1) T^{-1} U(-1)$.
Now, for which pairs $m,m'$ can we find an $n$ such that $m'=mn^3$?
The multiplicative group $\FF_p^\times$ is cyclic of order $p-1$.
Therefore, when $p-1$ is not a multiple of $3$,
then $\FF_p^\times \ni n \mapsto n^3 \in \FF_p^\times$ is a bijection,
and any $T^m$ can be interconverted into any other $T^{m'}$.
If $p-1$ is a multiple of $3$, there are three distinct classes of $T$ gates.
Since $-1 = (-1)^3$, $T^m$ and $T^{-m}$ are always interconvertible.

\subsection{Inner codes}

For arbitrary vector $v \in \FF_p^{\nin}$ 
we write $X(v) = X^{v_1} \otimes \cdots \otimes X^{v_\nin}$,
and $Z(v) = Z^{v_1} \otimes \cdots \otimes Z^{v_\nin}$.
As in the weakly self-dual CSS code construction for qubits,
it is straightforward to define a stabilizer code starting from
a self-orthogonal subspace $\calS \subset \calS^\perp \subset \FF_p^{\nin}$:
The stabilizer group is generated by $X(v)$ and $Z(v)$ where $v \in \calS$.
The quotient space $\calS^\perp / \calS$ is in one-to-one correspondence with
the set of $X$-type ($Z$-type) logical operators,
and the induced dot product on $\calS^\perp/\calS$ is non-degenerate.
In Section~\ref{symmformfinite} below,
we show that there is a basis $\{v^{(1)},\ldots, v^{(\kin)} \}$ of $\calS^\perp / \calS$
such that $ v^{(i)} \cdot v^{(j)} = \alpha_j \delta_{ij} $
where the scalars $\alpha_j$ are all $1$ possibly except the last one.
For simplicity we restrict ourselves to cases where
\begin{align}
 (1,1,\ldots,1) &\in \calS ,\label{eq:all1inclusion}\\
 v^{(i)} \cdot v^{(j)} &= \delta_{ij}. \label{eq:FpSymForm},
\end{align}
i.e., the second condition is that all scalars $\alpha_j$ are equal to $1$.
The first condition demands that $\nin$ to be a multiple of $p$.
The second is a mild restriction,
since $(\calS \oplus \calS)^{\perp} / (\calS \oplus \calS)$
always has a basis such that \eqref{eq:FpSymForm} holds.
Given a basis $\{ v^{(j)}\}$ satisfying \eqref{eq:FpSymForm},
we define logical operators of the inner code as
\begin{align}
 \tilde X^{(j)} &= X(v^{(j)}),\nonumber\\
 \tilde Z^{(j)} &= Z(v^{(j)}),
\end{align}
which indeed obey the commutation relation 
\begin{align}
 \tilde Z^{(a)} \tilde X^{(b)} = \omega^{\delta_{ab}} \tilde X^{(b)} \tilde Z^{(a)}
\end{align}
of the generalized Pauli operators on $\kin$ qudits.
Thus, this is a generalization of the normal codes in the qubit case.
Due to \eqref{eq:all1inclusion},
the transversal gate $\bar S = S^{\otimes \nin}$ is a logical operator:
\begin{align}
S^j X^k S^{-j} 
 &= 
\omega^{-j k (k+1)/2} Z^{jk} X^k \\
%%%%
\bar S X(v) \bar S^{-1} 
&= 
\omega^{- (v \cdot v + v \cdot \vec 1)/2}
Z(v) X(v)
\end{align}
where in the second equation the phase factor vanishes when $v \in \calS$.

We will implement the measurement of the stabilizer $T^m X T^{-m}$ 
of the magic state $\ket{\psi_m}$ using the inner codes.
The measurement becomes feasible if $^C(T^m X T^{-m})$ can be implemented
for \emph{logical} qudits.
We begin searching for its fault-tolerant implementation by observing
an identity $^C(T^m X T^{-m}) = T^m (^C X) T^{-m}$
that enables us to implement some controlled Clifford on logical qudits.
The actual action on logical qubits depends on the inner code,
but our conditions (\ref{eq:all1inclusion},~\ref{eq:FpSymForm}) will make it uniform across all logical qudits. 

Recall $T^m X T^{-m} = \eta^{-1} S^m X$ where $\eta = 1$ if $p >3$ and $\eta = e^{2\pi i /9}$ if $p =3$.
The action of the transversal gate $\bar T^m \bar X \bar T^{-m}$
can be deduced by looking at the logical operators and phase.
The answer is
\begin{align}
\bar T^m \bar X \bar T^{-m} 
= 
\eta^{-\nin} \bar S^m \bar X 
\cong
\eta^{\kin-\nin} \prod_{a=1}^{\kin} \eta^{-1}(\tilde S^{(a)} (\tilde Z^{(a)})^{1/2})^{m} 
\end{align}
because
\begin{align}
\begin{cases}
\bar S^m \tilde X^{(a)} \bar S^{-m} 
=
\omega^{-m/2} (\tilde Z^{(a)})^m \tilde X^{(a)} \\ 
\bar S^m \tilde Z^{(a)} \bar S^{-m} 
=
\tilde Z^{(a)}\\
\bar S^m 
 \underbrace{\sum_{v \in \calS} \frac{X(v)}{\sqrt{|\calS|}} \ket{0}^{\otimes \nin}}_{\ket{\tilde 0}^{\otimes \kin}}
 = \ket{\tilde 0}^{\otimes \kin}
\end{cases},
\quad
\begin{cases}
 (S^m Z^{m/2}) X (S^m Z^{m/2})^{-1} = \omega^{-m/2} Z^m X\\
 (S^m Z^{m/2}) Z (S^m Z^{m/2})^{-1} = Z\\
 (S^m Z^{m/2}) \ket 0 = \ket 0
\end{cases}.
\end{align}

Suppose $p > 3$. In order to implement $^C(\tilde S^m \tilde X)$,
we consider an equation and a solution
\begin{align}
 ^C(S^m X) &= (^C X^{1-y}) (^C Z^{u}) (^C S^x) (^C X^y) (^C S^z) (^C Z^s) (^C \omega^t) \label{eq:decomposition-CSX}\\
 ^C(S^m X) &= (^C X^{1/3}) [(^C Z^{3m/8}) (^C S^{3m/4})] (^C X^{2/3}) [(^C S^{m/4}) (^C Z^{m/8})] (^C \omega^{-m/6})
\end{align}
where the control is common for every gate,
and $u,x,y,z,s,t$ are variables.
(Using $^C A = \sum_j \ket j \bra j \otimes A^j$,
one can evaluate matrix elements on both sides.)
Note that the operators in the brackets are powers of $^C(S^m Z^{m/2})$.
This implies that indeed simultaneous $^C(\tilde S^m \tilde X)$ on all logical qudits can be implemented
using $\bar T^m (^C\bar X^{3m/4}) \bar T^{-m}$, 
$\bar T^m (^C\bar X^{m/4}) \bar T^{-m}$, 
controlled Pauli logical operators,
and a power of $Z$ on the control.

When $p = 3$ it suffices to consider $m=1$. 
To remove the phase factor $\eta^{\kin-\nin}$
we require that the $\kin$ is a multiple of 3.
This can be achieved by considering three copies of a given code if necessary.
$\nin$ is already a multiple of $3$ due to \eqref{eq:all1inclusion}.
We can implement $\prod_{a=1}^\kin ~^C(\eta^{-1} \tilde S^{(a)} \tilde X^{(a)})$ by
an identity
\begin{align}
 ^C(\eta^{-1}SX) = (^C X^{-1}) (^C (\eta^{-1}SZ^{-1})) (^C Z) (^C X^{-1}) (^C Z^{-1}) (^C \omega^2) .
\end{align}

We have shown that it is possible to build a fault-tolerant routine to measure $\tilde T^m \tilde X \tilde T^{-m}$.

We have not yet shown how to construct such inner codes.
It is possible to generalize Lemma~\ref{rwsdlemma}
to the case of matrices over a field $\FF_p$ for $p>2$; 
however, the generalization is more difficult 
since the self-orthogonality constraint implies a nonlinear constraint 
on the rows of the matrix so that each row is null; see Lemma~\ref{lem:qudit-rwsd}.
Let us give an alternative construction 
which achieves the scaling similar to Lemma~\ref{fixedd}, 
namely that for any distance $d$,  
one can find a family of normal weakly self-dual qudit CSS codes 
with $X(\vec 1)$ in the stabilizer group 
such that the ratio $\kin/\nin \rightarrow 1$ as $\nin \rightarrow \infty$.
This construction is derived from Reed-Muller codes.
Let $C=RM_{\FF_p}(r,m)$ be a classical Reed-Muller code over $\FF_p$; 
the codewords have length $p^m$.
The dual code is $C^\perp=RM_{\FF_p}(m(p-1)-r-1,m)$; see Theorem 5.4.2 of Ref.~\cite{assmus1992designs}.
For any fixed $r$, for large enough $m$, $C\subset C^\perp$, so the codespace of $C$ is self-orthogonal, 
and $\vec 1$ is in the codespace of $C$.
We use the codespace of $C$ as the space $\calS$, 
and use the CSS construction to define a weakly self-dual code.
For fixed $r$, the rate of $C$ tends to zero at large $m$, 
so the rate of the resulting weakly self-dual tends to $1$.
See Ref.~\cite{RMqubit} for weakly self-dual qubit codes derived from Reed-Muller codes.
To make~\eqref{eq:FpSymForm} hold, it may be necessary to use $\calS \oplus \calS$ instead of $\calS$.

\subsection{Outer codes}

If the inner code has code distance $d$, then we should use an outer code 
with a parity check matrix that is $(d-1, \lceil \frac{d-1}{2} \rceil)$-sensitive.
In full generality,
one would want to use a parity check matrix with entries in $\FF_p$,
where an entry $\beta \neq 0$ would mean a stabilizer $(\eta^{-1}S^m X)^\beta$.
This makes it necessary to have a different logical operator choice than we have used above.

However, a check matrix that is given by the adjacency matrix of a biregular graph 
with large girth is sufficient for us.
Such a check matrix has only 0 and 1 entries,
so no other choice of logical operator is necessary beyond what we have given above.
Recall that a graph with large girth is locally a tree.
Hence, a bad magic state will be caught by many checks 
because it flips a single stabilizer in these checks,
and the required sensitivity is guaranteed.

\section{Symmetric forms over finite fields}
\label{symmformfinite}

We have classified nondegenerate symmetric forms 
over the binary field $\FF_2$ in Section~\ref{sec:innercodes}.
Over a field of odd characteristic,
the set of all finite dimensional vector spaces 
with nondegenerate symmetric forms ({\em quadratic spaces} for short)
constitute an abelian group under the direct sum, after identifying hyperbolic planes as the identity.
This group is known as the Witt group of the field,
and the group structure is well known.
Here we present a self-contained and elementary treatment of 
the Witt group of $\FF_p$,
and classify the quadratic spaces over fields of odd characteristic.
A {\bf square element}, or a {\bf square} for short,
is any member of the set $\{ x^2 : x \in \FF_p^2 \}$.

It is natural to distinguish two cases depending on whether $-1 \in \FF_p$ is a square,
since a one-dimensional quadratic space is classified by $\FF_p^\times / (\FF_p^\times)^2$,
where $\FF_p^\times := \FF_p \setminus \{0\}$ and $(\FF_p^\times)^2 := \{ x^2 ~|~ x \in \FF^\times\}$.
Since the multiplicative group $\FF_p^\times$ is a cyclic group of order $p-1$,
the element $-1$ being the unique element of $\FF^\times$ with multiplicative order 2,
is a square if and only if $p = 1 \mod 4$.

The part of the argument in Section~\ref{sec:innercodes} applies here without any change
where we have inductively converted any non-degenerate symmetric
matrix to a direct sum of a diagonal matrix and blocks of
$
 \begin{pmatrix}
 0 & 1 \\ 1 & 0
 \end{pmatrix},
$
which represents a {\em hyperbolic plane}.
Below, we assume that symmetric matrices are block diagonal in this form.
It is then easy to explain why quadratic spaces constitute a group:
\begin{align}
 \begin{pmatrix} 1 & 1 \\ 1 & -1 \end{pmatrix}^T
 \begin{pmatrix} a & 0 \\ 0 & -a \end{pmatrix}
 \begin{pmatrix} 1 & 1 \\ 1 & -1 \end{pmatrix}
=
\begin{pmatrix} 0 & 2a \\ 2a & 0 \end{pmatrix}
\simeq
\begin{pmatrix} 0 & 1 \\ 1 & 0 \end{pmatrix}
\end{align}
This means that the one-dimensional quadratic space with form $(-a)$ is the inverse of
the space with form $(a)$.
It is important here that $2$ is an invertible element of the field.

We note that the determinant of the symmetric form up to squares
is a nontrivial invariant valued in the multiplicative group $\FF_p^\times / (\FF_p^\times)^2$
which is isomorphic to the additive group $\mathbb Z/2\mathbb Z$.
Let $\alpha \in \FF_p$ be a non-square.

{\em Case I}: $p = 1 \mod 4$ so that $-1 \in (\FF_p^{\times})^2$.
Consider a block $\dm(a,a)$ of the symmetric matrix.
Since $-1$ is a square, we see $\dm(a,a) \simeq \dm(a,-a) \simeq \dm(1,-1) \simeq \dm(1,1)$
under congruent transformations.
Therefore, there are four classes of symmetric matrices up to hyperbolic planes:
$\dm(1)$, $\dm(\alpha)$, $\dm(1,\alpha)$, and $\dm(1,1)$.
By looking at the determinant of the form and the parity of the dimension,
we see that the four classes are distinct elements of the Witt group,
which is hence isomorphic to $\mathbb Z/ 2 \mathbb Z \oplus \mathbb Z / 2 \mathbb Z$.
Given a dimension of quadratic spaces,
we see there are only two exclusive possibilities:
\begin{align}
 \dm(1,1,\ldots,1,1), \quad\text{and}\quad \dm(1,1,\ldots,1,\alpha).
\end{align}

{\em Case II}: $p = 3 \mod 4$ so that $-1 \notin (\FF_p^{\times})^2$.
In this case, we can set $\alpha = -1$.
We claim that $\dm(1,1)$ is not hyperbolic.
If $v = a v_1 + b v_2$ is a vector in this two-dimensional space,
where $v_1,v_2$ are basis vectors with $v_i^2 = 1$
and $a,b \in \FF_p$,
then $v \cdot v = a^2 + b^2$.
Since $-1$ is not a square,
the equation $a^2 + b^2= 0$ does not have any nonzero solution,
and this proves the claim.
Next, we show that $\dm(1,1) \simeq \dm(-1,-1)$.
To this end, we will find a solution to $a^2 + b^2 +1 = 0$ over $\FF_p$.
Once we have such a solution, then we see
\begin{align}
\begin{pmatrix}
 a  & b \\
 b  & -a 
\end{pmatrix}^T
\begin{pmatrix}
 1  & 0 \\
 0  & 1
\end{pmatrix}
\begin{pmatrix}
 a  & b \\
 b  & -a 
\end{pmatrix}
 =
 \begin{pmatrix}
 -1 & 0 \\
 0 & -1
 \end{pmatrix}.
\end{align}
The existence of the solution follows from 
$(\FF_p^\times)^2 + (\FF_p^\times)^2 \not\subseteq (\FF_p^\times)^2$,
which implies that $(\FF_p^\times)^2 + (\FF_p^\times)^2 \ni -1$.
If $(\FF_p^\times)^2 + (\FF_p^\times)^2 \subseteq (\FF_p^\times)^2$,
then $(\FF_p^\times)^2$ would be a monoid under addition contained in a finite group,
and hence would be a group itself,
which must contain $0 \notin (\FF_p^\times)^2$.
Therefore, quadratic spaces given a dimension are classified by the determinant of the form up to squares.
\begin{align}
 \dm(1,1,\ldots, 1,1), \quad\text{and}\quad \dm(1,1,\ldots,1,-1)
\end{align}
The Witt group of $\FF_p$ is isomorphic to $\mathbb Z / 4 \mathbb Z$
generated by $\dm(1)$.

Now we state and prove some facts about quadratic spaces.
\begin{lemma}[Chapter~XV Theorem~10.2 of Ref.~\cite{Lang}]\label{lem:isometryextension}
Let $Q$ be a nondegenerate quadratic space.
If two subspaces $V$ and $U$ are isomorphic by an isometry $\sigma: V \to U$,
then there exists an isometry $\bar \sigma : Q \to Q$ such that $\bar \sigma | _V = \sigma$.
\end{lemma}

\begin{lemma}
Let $N$ be a null subspace (on which the symmetric form vanishes) 
of a nondegenerate quadratic space $Q$ over $\FF_p$.
Then, $Q$ is isometric to the orthogonal sum of $N^\perp / N$ 
and a minimal hyperbolic subspace that contains $N$.
\end{lemma}
\begin{proof}
Applying Lemma~\ref{lem:isometryextension} to the identity map $\sigma$,
we conclude that any orthogonal set of vectors extends to an orthogonal basis.
Since the form is nondegenerate,
there exists a minimal hyperbolic subspace that includes $N$ ({\em hyperbolic extension}),
and the symmetric form can be written as $\Lambda' \oplus \lambda$.
where $\lambda$ is hyperbolic, and $\Lambda'$ is nondegenerate.
It is then clear that $N^\perp / N$ has the symmetric form $\Lambda'$.
\end{proof}

\begin{lemma}
Let $Q$ be a nondegenerate quadratic space of dimension $n$ over $\FF_p$.
Every maximal null subspace of $Q$ has the same dimension $m$.
Given any null subspace $N$ of dimension $k \le m$,
the number of null vectors of $Q$ that are orthogonal to $N$ is
\begin{align}
\#\mathcal{Z}(Q,N) = p^{n-k-1} +p^{m} - p^{n-m-1} =: \zeta(n,m,k).
\end{align}
\end{lemma}
\begin{proof}
To prove the first claim,
suppose $M, M'$ are maximal null subspaces.
If $\dim M \le \dim M'$, 
then any injection from $M$ to $M'$ is an isometry,
which can be extended to $Q$ as $\bar \sigma$.
Then, $\bar \sigma^{-1}(M')$ is a null superset of $M$,
and hence is $M$ itself since $M$ is maximal. Thus, $\dim M = \dim M'$. 

Let $\mathcal{Z}(Q,N)$ be the set of all null vectors of $Q$ that are orthogonal to $N$.
($\mathcal{Z}$ is not a subspace in general.)
Consider $\phi : \mathcal{Z}(Q,N) \to \mathcal{Z}(N^\perp/N,0)$, 
a restriction of the canonical projection map $Q \to Q/N$.
The map $\phi$ is surjective by definition of $\zeta$.
If $x,y \in \mathcal{Z}(Q,N)$ are mapped to the same element,
then $x-y \in N$.
This implies that $\phi$ maps exactly $\# N$ elements to one.
(Here, $\#$ denotes the number of elements of the finite set.)
Therefore,
\begin{align}
 \# \mathcal{Z}(Q,N) = (\# N) (\# \mathcal{Z}(N^\perp/N,0)).
\end{align}
Due to the preceding lemma, 
the dimension of a maximal null subspace of $N^\perp/N$ is $m-k$.
Thus, it remains only to prove the lemma when $k = 0$ since
\begin{align}
 \#\mathcal{Z}(Q,k) = p^k (p^{n-2k-1}-p^{n-m-k-1}+p^{m-k}) = p^{n-k-1}-p^{n-m-1} + p^{m}.
\end{align}

A {\em definite} quadratic space is one in which $w \cdot w = 0$ implies $w = 0$.%
\footnote{
A definite space corresponds to a nontrivial element of the Witt group.
For example, a one-dimensional nondegenerate space is definite.
If $\alpha$ is not a square, then the symmetric form $\dm(1,-\alpha)$ is definite.
}
To count all null vectors, we work in a basis 
such that the $n$-by-$n$ symmetric matrix is 
\begin{align}
\Lambda = \Lambda' \oplus \Lambda_{2m} 
\end{align}
where $\Lambda'$ definite, and 
$\Lambda_{2m} = \frac 1 2 \begin{pmatrix} 0 & 1 \\ 1 & 0 \end{pmatrix} \otimes I_m$
is an orthogonal sum of $m$ hyperbolic planes.
In this basis, let us write any vector $x$ as $x' \oplus (u,u')$.
The nullity is then expressed by a quadratic equation of coordinates
\begin{align}
 x' \cdot x' + u \cdot u' = 0.
\end{align}
The solutions of this equation are divided into two classes: $x' \cdot x' = 0$ or $x' \cdot x' \neq 0$.
In the former case, $x' = 0$ and $u \cdot u' =0$. 
Given arbitrary $u$ there is $u'$ such that this equation holds.
The number of solutions is $p^m + (p^m-1)p^{m-1}$.
In the latter case, we must have $u \neq 0$,
and 
$u \cdot u' = c = - x' \cdot x' \neq 0$ is a inhomogeneous equation in $u'$,
whose solution always exists.
For any given nonzero $c$, there are thus $(p^m-1)p^{m-1}$ choices of $(u,u')$.
$x'$ can be any nonzero vector, so there are $p^{n-2m}-1$ choices.
In sum, the number of null vectors in an $n$-dimensional quadratic space $Q$ over $\FF_p$ is
\begin{align}
\#\mathcal{Z}(Q,0) &= p^m + (p^m-1)p^{m-1} + (p^{n-2m}-1)(p^m-1)p^{m-1} = p^{n-1} - p^{n-m-1} + p^m.
\end{align}
\end{proof}

\begin{lemma}
\label{lem:qudit-rwsd}
Let $w_1 = \vec 1 \in \FF_p^n$ be the all-1 vector where $n$ is a multiple of $p \ge 3$.
Assume $c < (n - 2)/2$,
and let $w_2, \ldots, w_c$ be null vectors of $\FF_p^n$ chosen inductively
such that $w_j$ is chosen uniformly at random from $\mathcal{Z}(\FF_p^n, V_{j-1})$
where $V_{j-1} = \mathrm{span}(w_1,\ldots,w_{j-1})$.
Let $M$ be a $c$-by-$n$ matrix with rows $w_j$.

Consider a fixed $n$-component vector $v$,
with $v \neq 0$ and $v\neq \vec 1$.
The probability that $M v=0$ is bounded from above by
\begin{align}
 20 \left(\frac 3 5\right)^{n-c} + \left( \frac{11}{15}\right)^{c-1}.
\end{align}
\end{lemma}

\begin{proof}
We will estimate the desired probability by a union bound, 
considering separately the event that $v \in V_c^\perp$ and $v\not\in V_c$, 
and the event that $v\in V_c^\perp$ and $v \in V_c$.
The second event is possible only if $v \cdot v = 0$.
By the classification of symmetric forms,
a maximal null space of $\FF_p^n$ has dimension $m$ such that $n-2 \le 2m \le n$.
The assumption that $c < (n-2)/2$ implies that
\begin{align}
k_j := \dim V_j \le j \le c \le m-1.
\end{align}

Consider the first event, assuming $v\cdot v = 0$.
Let $j>1$.  Then
\begin{align}
\label{est1d}
\Pr[ v \in V_j^\perp \,{\rm and}\, v \notin V_j | v \in V_{j-1}^\perp \setminus V_{j-1} ] 
\le \frac{\zeta(n,m, k_{j-1}+1 )}{\zeta(n, m, k_{j-1})} \le \frac{p}{2p-1} \le \frac 3 5
\end{align}
because $w_j$ has to be orthogonal to $\text{span}(v) + V_{j-1}$, which is null and is a proper superset of $V_{j-1}$.
Thus, for any $t$,
\begin{align}
\label{estkd}
\Pr[v \in V_t^\perp \, {\rm and} \, v \not \in V_t] 
&=\Pr[ v \in V_1^\perp \setminus V_1 ] \prod_{j=2}^t \Pr[v \in V_j^\perp \setminus V_j ~|~ v \in V_{j-1}^\perp \setminus V_{j-1}] \nonumber \\
&\leq \prod_{j=2}^t \frac{\zeta(n,m, k_{j-1}+1 )}{\zeta(n ,m, k_{j-1})}\nonumber \\
&\le \left(\frac 3 5\right)^{t-1}.
\end{align}
For $t=c$, we find in particular that
\begin{align}
\label{estcd}
\Pr[v \in V_c^\perp \, {\rm and} \, v \not \in V_c] \leq \left(\frac 3 5\right)^{c-1}.
\end{align}

Now assume $v \cdot v \neq 0$.
The event that $v \in V_j^\perp $ happens only if $w_j$ is chosen from $v^\perp$.
We bound the decomposition 
$
\Pr[ v \in V_c^\perp] 
= 
\Pr[v \in V_1^\perp] \prod_{j=2}^c \Pr[v \in V_j^\perp | v \in V_{j-1}^\perp]
$.
The first term is bounded by 1 trivially.
For other factors, we observe that the dimension of $v^\perp$ is $n-1$,
and a maximal null subspace in $v^\perp$ has dimension $m' \le m$.
Under the conditioning $v \in V_{j-1}^\perp$,
the null space $V_{j-1}$ is a subspace of $v^\perp$,
and $\#\mathcal{Z}(v^\perp,V_{j-1}) = \zeta(n-1,m',k_{j-1}) \le \zeta(n-1,m,k_{j-1})$.
Hence,
\begin{align}
\Pr[ v \in V_c^\perp] 
&= 
\Pr[v \in V_1^\perp] \prod_{j=2}^c \Pr[v \in V_j^\perp | v \in V_{j-1}^\perp] \nonumber\\
&\le
\prod_{j=2}^c \frac{\zeta(n-1,m,k_{j-1})}{\zeta(n,m,k_{j-1})} \nonumber\\
&\le
\left(\frac{p^2+p-1}{2p^2-p}\right)^{c-1} \le \left(\frac{11}{15}\right)^{c-1} \label{eq:vvnonzero}
\end{align}
where in the second inequality, we used the assumption that $k_{j-1} \le c < (n-2)/2 \le m$.

Let us turn to the second event, assuming $v \cdot v = 0$.
Note that if $v \in V_c$, there is a least $j$ such that $v\in V_j$.  
So,
\begin{align}
\label{sumeqd}
\Pr[v \in V_c^\perp \, {\rm and} \, v \in V_c] 
\leq 
\sum_{j=2}^c 
\Pr[v\in V_j \, {\rm and} \, v \in V_{j-1}^\perp \, {\rm and} \, v \not \in V_{j-1}].
\end{align}
We have
\begin{align}
&\Pr[v \in V_{j} \, {\rm and} \, v\in V_{j-1}^\perp \, {\rm and} \, v \not \in V_{j-1}] \nonumber \\ 
&= 
\Pr[v \in V_{j-1}^\perp \, {\rm and}\, v\not \in V_{j-1}] \cdot  
\Pr [v\in V_j  | v \in V_{j-1}^\perp \, {\rm and}\, v\not \in V_{j-1}] \nonumber\\
&\leq 
\left(\frac 3 5\right)^{j-2}  \Pr [ v \in V_j  | v \in V_{j-1}^\perp \, {\rm and}\, v\not \in V_{j-1}],
\end{align}
where we used Eq.~(\ref{estkd}).
The second factor is bounded as
\begin{align}
 \Pr [ v \in V_j | v \in V_{j-1}^\perp \, {\rm and}\, v \not \in V_{j-1}] 
 \le \frac{p^{1+k_{j-1}}}{\zeta(n,m,k_{j-1})} 
 < \frac{3}{2p^{n-2j}}
\end{align}
because $w_j$ belongs to $\text{span}(v) + V_{j-1}$.
Hence,
$\Pr[v \in V_j \, {\rm and} \, v\in V_{j-1}^\perp \, {\rm and} \, v \notin V_{j-1}] 
\leq 5 (\frac 3 5)^{n-j}$.
So by Eq.~(\ref{sumeqd}),
\begin{align}
\label{event2d}
\Pr[v \in V_c^\perp \, {\rm and} \, v \in V_c] \leq \sum_{j=2}^c 5 \left(\frac 3 5\right)^{n-j} 
< 20 \left(\frac 3 5\right)^{n-c}.
\end{align}

Summing the probabilities of \eqref{eq:vvnonzero} and \eqref{event2d},
we conclude the proof.
\end{proof}

% \bibliography{magic-ref}
\end{document}